\documentclass{article}

\usepackage[margin=1in]{geometry}

\usepackage{amssymb,amsmath,amsthm,amsfonts}
\usepackage{bm}
\usepackage{textgreek}
\usepackage{mathtools}
\usepackage{enumitem}
\usepackage[numbers,comma,sort&compress]{natbib}
\usepackage{authblk}
\usepackage{graphicx}
\usepackage[font=small]{caption}
\usepackage{subcaption} 
\usepackage{float}
\usepackage[ruled,vlined,linesnumbered]{algorithm2e}
\newcounter{proccnt}

\makeatletter
\newenvironment{procedureblock}[1][htbp]{%
    \let\c@algocf\c@proccnt

    \SetAlgorithmName{Procedure}{Procedure}{List of Procedures}%
    
    \begin{algorithm}[#1]%
}{%
    \end{algorithm}%
}
\makeatother

\usepackage[colorlinks]{hyperref}

\usepackage{braket}
\usepackage{float}
\usepackage{tablefootnote}
\usepackage{overpic}
\usepackage{multirow}
\usepackage{afterpage}
\usepackage{multirow}     
\usepackage{booktabs}      
\usepackage{tabularx}
\newcolumntype{Y}{>{\centering\arraybackslash}X}
\usepackage{makecell}
\usepackage{xcolor}

\newtheorem{theorem}{Theorem}
\newtheorem{definition}{Definition}
\newtheorem{lemma}{Lemma}
\newtheorem{proposition}{Proposition}

\newtheorem{remark}{Remark}
\newtheorem{assumption}{Assumption}

\newcommand{\eq}[1]{(\ref{eq:#1})}

\newcommand{\thm}[1]{\hyperref[thm:#1]{Theorem~\ref*{thm:#1}}}
\newcommand{\cor}[1]{\hyperref[cor:#1]{Corollary~\ref*{cor:#1}}}
\newcommand{\defn}[1]{\hyperref[defn:#1]{Definition~\ref*{defn:#1}}}
\newcommand{\lem}[1]{\hyperref[lem:#1]{Lemma~\ref*{lem:#1}}}
\newcommand{\prop}[1]{\hyperref[prop:#1]{Proposition~\ref*{prop:#1}}}
\newcommand{\assum}[1]{\hyperref[assum:#1]{Assumption~\ref*{assum:#1}}}
\newcommand{\fig}[1]{\hyperref[fig:#1]{Figure~\ref*{fig:#1}}}
\newcommand{\tab}[1]{\hyperref[tab:#1]{Table~\ref*{tab:#1}}}
\newcommand{\algo}[1]{\hyperref[algo:#1]{Algorithm~\ref*{algo:#1}}}
\renewcommand{\sec}[1]{\hyperref[sec:#1]{Section~\ref*{sec:#1}}}
\newcommand{\append}[1]{\hyperref[append:#1]{Appendix~\ref*{append:#1}}}
\newcommand{\proc}[1]{\hyperref[proc:#1]{Procedure~\ref*{proc:#1}}}
\newcommand{\fac}[1]{\hyperref[fac:#1]{Fact~\ref*{fac:#1}}}
\newcommand{\lin}[1]{\hyperref[lin:#1]{Line~\ref*{lin:#1}}}

\renewcommand{\arraystretch}{1.25}

\def\>{\rangle}
\def\<{\langle}

\newcommand{\R}{\mathbb{R}}

\newcommand{\Robustalgo}{Approximate }
\newcommand{\robustalgo}{approximate }

\title{Scalable First-Order Interior Point Trust Region Algorithms for Linearly Constrained Optimization}

\author[1,2]{Yuexin Su\thanks{\texttt{yuexinsu@stu.pku.edu.cn}}}
\author[3]{Chenyi Zhang\thanks{\texttt{chenyiz@stanford.edu}}}
\author[4]{Peiyuan Huang\thanks{\texttt{pyhuang@gsm.pku.edu.cn}}}
\author[1,2]{Tongyang Li\thanks{Corresponding author. \texttt{tongyangli@pku.edu.cn}}}
\author[5,6]{Yinyu Ye\thanks{\texttt{yinyu-ye@stanford.edu}}}

\affil[1]{Center on Frontiers of Computing Studies, Peking University}
\affil[2]{School of Computer Science, Peking University}
\affil[3]{Computer Science Department, Stanford University}
\affil[4]{Guanghua School of Management, Peking University}
\affil[5]{Department of Management Science and Engineering, Stanford University}
\affil[6]{Shanghai Institute for Mathematics and Interdisciplinary Sciences}

\date{}
\begin{document}

\maketitle

\begin{abstract}
Computing approximate Karush--Kuhn--Tucker (KKT) points for constrained nonconvex programs is a fundamental problem in mathematical programming. Interior-point trust-region (IPTR) methods are particularly attractive for such problems because they maintain strictly feasible iterates throughout the iterative process and converge to a first-order and second-order KKT solution. Their scalability, however, is limited by the repeated computation of trust-region search directions. In this paper, we propose an \robustalgo first-order IPTR framework that addresses this bottleneck by replacing exact trust-region subproblem solves with an approximate projector maintained through low-rank updates. The resulting method preserves feasibility and the global convergence guarantees of standard IPTR schemes while substantially reducing the per-iteration cost. We further extend the framework to obtain approximate second-order KKT points using only first-order information by integrating a gradient-based negative-curvature routine, thus avoiding explicit Hessian computations. We conduct numerical experiments to demonstrate the scalability of our \robustalgo first-order IPTR framework in large-scale settings, where it achieves up to a $2.48\times$ speedup over the existing first-order IPTR algorithm.
\end{abstract}

\section{Introduction}

In this paper, we consider the following optimization problem: 
\begin{equation}
\begin{aligned}
\label{eq:main-prob}
\min_{\boldsymbol{x}\in \mathbb{R}^n}\quad &f(\boldsymbol{x})  \\
\mathrm{s.t.}\quad &\boldsymbol{A}\boldsymbol{x}=\boldsymbol{b}, \boldsymbol{x}\geq \boldsymbol{0},
\end{aligned}
\end{equation}
where $\boldsymbol{A} \in \mathbb{R}^{m \times n}$ and $\boldsymbol{b} \in \mathbb{R}^{m}$. We assume that $m \leq n$ and that $\boldsymbol{A}$ has full row rank. The objective function $f\colon\mathbb{R}^n \to \mathbb{R}$ is continuous on the nonnegative orthant $\mathbb{R}_{+}^n := \{\boldsymbol{x} \in \mathbb{R}^n \mid \boldsymbol{x} \geq \boldsymbol{0}\}$ and differentiable on the strictly positive orthant $\mathbb{R}_{++}^n := \{\boldsymbol{x} \in \mathbb{R}^n \mid \boldsymbol{x} > \boldsymbol{0}\}$. Problem \eq{main-prob} is a fundamental formulation in mathematical programming, covering a broad class of nonconvex optimization problems over the nonnegative orthant with affine equality constraints, including sparse signal reconstruction \cite{che2022sparse}, nonnegative matrix factorization \cite{guo2024rise}, and portfolio selection \cite{xia2023high}. In many such problems, the objective function $f$ is nonconvex, making the search for global minimizers NP-hard. Consequently, theoretical and algorithmic developments typically focus on identifying local minimizers or stationary points that satisfy necessary optimality conditions. 

In the unconstrained setting, such optimality conditions are well-established through the notions of first- and second-order stationarity. A point $\boldsymbol{x}$ is called an $\varepsilon$-first-order stationary point (FOSP) if $\|\nabla f(\boldsymbol{x})\|\le \varepsilon$. An $(\varepsilon,\sqrt{\varepsilon})$-second-order stationary point (SOSP) further requires that
$\lambda_{\min}(\nabla^2 f(\boldsymbol{x})) \ge -\sqrt{\rho\varepsilon}$, where $\rho$ is the Lipschitz constant of the Hessian. 
An $\varepsilon$-FOSP captures first-order stationarity, but it is generally insufficient in nonconvex optimization since it may still be a saddle point. Such saddle points can be highly suboptimal and thus undesirable in practice \cite{jain2017global,sun2018geometric}.
This motivates algorithms that go beyond first-order stationarity and provably converge to SOSPs. A direct way is to use second-order methods, but explicitly forming and factorizing the Hessian typically costs $\Omega(n^2)$, which can be prohibitive in high-dimensional problems. A large body of work therefore focuses on obtaining second-order guarantees using primarily gradient information.  
When the gradient is small, either the Hessian is nearly positive semi-definite so $\boldsymbol{x}$ is already close to an SOSP, or a direction of negative curvature exists that allows further descent. Such directions can be found using first-order methods. In particular, Refs.~\cite{Zhu2018NEON2,carmon2018accelerated,Liu2018Adaptive,xu2017neon+,Jin2018Accelerated,zhang2021Escape} establishes that  $\widetilde{\mathcal{O}}(1 / \varepsilon^{1.75})$\footnote{Throughout this paper, we use $\widetilde{O}(\cdot)$ to suppress polylogarithmic factors in $\varepsilon^{-1}$ and $n$.} gradient queries suffice to find an $(\varepsilon,\sqrt{\varepsilon})$-SOSP.  Extending this $\widetilde{O}(1/\varepsilon^{1.75})$ complexity to constrained optimization remains an open question.

Constrained optimization, such as constrained learning and training in AI, has become popular and even necessary in practice \cite{NEURIPS2019_cf708fc1,NEURIPS2024_9979a69d,lian2022advances,NEURIPS2024_bcfcf723,liu2024enhancing,pmlr-v162-yu22d}. However, the aforementioned stationarity-based characterization is no longer sufficient in such constrained settings, where feasibility must be taken into account and optimal solutions may lie on the boundary of the feasible region. In such cases, vanishing gradients or positive semidefinite Hessians of the objective alone do not capture optimality. The KKT conditions provide a principled extension of first- and second-order stationarity to constrained problems by jointly incorporating the objective and the constraints through primal–dual optimality conditions. Under standard regularity assumptions, KKT points are necessary for local optimality and thus play a role analogous to that of stationary points in the unconstrained setting \cite{bertsekas1997nonlinear,nocedal2006numerical}. This motivates the study of algorithms that converge to approximate first-order and second-order KKT points in constrained nonconvex optimization.
A variety of algorithms have been proposed to identify approximate KKT points in constrained optimization, including trust-region methods \cite{haeser2019optimality} and augmented Lagrangian approaches \cite{li2021augmented,kong2023iteration}.  
Inspired by unconstrained optimization, where gradient-based perturbations allow efficient escape from saddle points, it is natural to seek methods that achieve similar efficiency in the constrained setting. In particular, we aim at reaching approximate first- and second-order KKT points using only gradient information, which can avoid the heavy computational cost of second-order oracles.

\subsection{Our contributions}
\paragraph{Main results}
We introduce a new framework of first-order IPTR algorithms designed for nonconvex optimization problems with linear and non-negativity constraints. Our primary contribution is improving the computational complexity of finding approximate first- and second-order KKT points, which makes our algorithms highly scalable for large-scale applications.

Our contributions are twofold. On the one hand, we propose an \robustalgo first-order IPTR algorithm that reduces the cost of computing interior iterates via low-rank updates. On the other hand, we develop first-order IPTR algorithms that find approximate second-order KKT points without Hessian information. Both contributions are particularly appealing for large-scale problems in which projection or matrix factorization is computationally expensive.
\begin{enumerate}
    \item \textbf{\Robustalgo first-order IPTR on finding $2\varepsilon$-KKT points.} We introduce an \robustalgo first-order IPTR (\algo{robust-1st-order}). This algorithm maintains the iteration complexity guarantees of existing first-order IPTR approaches while significantly reducing the total computational runtime. The improvement is achieved by replacing the exact solution of the trust-region subproblems with an approximate update scheme. This mechanism not only preserves the feasibility of the iterates but also circumvents the need for the frequent matrix factorizations required in prior schemes, leading to a significant reduction of the average per-iteration computational cost from $\mathcal{O}(nm^{\omega-1})$ to $\widetilde{\mathcal{O}}(nm)$, where $\omega \approx 2.371339$ denotes the most recent matrix multiplication exponent \cite{matrixexp}.
    \item \textbf{First-order IPTR on finding $(2\varepsilon,\sqrt{\varepsilon})$-KKT2 points.} We propose \algo{1st-order interior} and \algo{robust-nega} that compute approximate second-order KKT points using solely first-order information. \algo{1st-order interior} integrates a negative-curvature finding subroutine into the basic IPTR framework, whereas \algo{robust-nega} further enhances this approach with the \robustalgo update mechanism. By leveraging negative-curvature directions, both algorithms escape saddle points of the Lagrangian without explicit Hessian computation. Both algorithms match the iteration complexity of prior first-order IPTR algorithms, and \algo{robust-nega} achieves a better total runtime compared to  prior second-order IPTR algorithms due to the computational efficiency of the \robustalgo update scheme.
\end{enumerate}
\tab{iter-time-tab} summarizes the oracle requirements, iteration bounds, and runtime complexities of algorithms.

\begin{table}[!htb]
  \centering
  \begin{tabularx}{\textwidth}{cccYY}
    \toprule
    \textbf{Criteria} & \textbf{Algorithm}  & \textbf{Oracle} &  \textbf{Iterations} & \textbf{Time complexity} \\
    \midrule
    $2\varepsilon$-KKT & \cite[1st-order IPTR]{haeser2019optimality}  & 1st-order &  $\mathcal{O}\left(\frac{l (f(\boldsymbol{x}_0)-f(\boldsymbol{x}^*))}{\varepsilon^2} \right)$ & ${\mathcal{O}}\left(\frac{nm^{\omega-1}}{\varepsilon^2}\right)$ \\
    $2\varepsilon$-KKT & \algo{robust-1st-order} & 1st-order  &  $\mathcal{O}\left(\frac{l (f(\boldsymbol{x}_0)-f(\boldsymbol{x}^*))}{\varepsilon^2} \right)$ & $\widetilde{\mathcal{O}}\left(nm^{\omega-1}+\frac{nm}{\varepsilon^2}\right)$  \\
    \midrule
    $(2\varepsilon,\sqrt{\varepsilon})$-KKT2 & \cite[2nd-order IPTR]{haeser2019optimality}  & 2nd-order  &  $\mathcal{O}\left(\frac{\max\{\eta,R\}^{3.5} (f(\boldsymbol{x}_0)-f(\boldsymbol{x}^*))}{\varepsilon^{1.5}} \right)$ & $\widetilde{\mathcal{O}}\left(\frac{n^\omega}{\varepsilon^{1.5}}\right)$\\
     $(2\varepsilon,\sqrt{\varepsilon})$-KKT2 & \algo{1st-order interior}  & 1st-order  &  $\widetilde{\mathcal{O}}\left(\frac{l \rho^2(f(\boldsymbol{x}_0)-f(\boldsymbol{x}^*))}{\varepsilon^2} \right)$ & $\widetilde{\mathcal{O}}\left(\frac{nm^{\omega-1}}{\varepsilon^2}\right)$ \\
     $(2\varepsilon,\sqrt{\varepsilon})$-KKT2 & \algo{robust-nega}  & 1st-order  & $\widetilde{\mathcal{O}}\left(\frac{l \rho^2(f(\boldsymbol{x}_0)-f(\boldsymbol{x}^*))}{\varepsilon^2} \right)$ & $\widetilde{\mathcal{O}}\left(\frac{nm^{\omega-1}}{\varepsilon^{1.5}}+\frac{nm}{\varepsilon^{2}}\right)$ \\
    \bottomrule
  \end{tabularx}
  \caption{Iteration bounds and overall time complexities of the IPTR-type algorithms for computing $2\varepsilon$-KKT and $(2\varepsilon,\sqrt{\varepsilon})$-KKT2 points. The parameters $l$, $\rho$, $\eta$, and $R$ are constants specified in the assumptions (see \sec{assump}).} 
\label{tab:iter-time-tab}
\end{table}

We also analyze the special case where $f$ is concave on $\Omega^\circ$. In this setting, the iteration complexity of first-order IPTR algorithms improves from $\mathcal{O}(1/\varepsilon^2)$ to $\mathcal{O}(1/\varepsilon)$; furthermore, we can return either a $2\varepsilon$-KKT point or an iterate $\boldsymbol{x}_t$ such that $f(\boldsymbol{x}_t)-f(\boldsymbol{x}^*)\le \varepsilon$ with time complexities $\mathcal{O}(nm^{\omega-1}/\varepsilon)$ (see \thm{concave-exact}) or $\widetilde{\mathcal{O}}(nm^{\omega-1}+nm/\varepsilon^2)$ (see \thm{concave-approximate}), depending on whether we apply the exact first-order IPTR algorithm~\cite{haeser2019optimality} or the approximate first-order IPTR algorithm (\algo{robust-1st-order}), respectively.

\paragraph{Techniques}
In existing IPTR methods, the search direction for \eq{main-prob} at each iteration is obtained by minimizing a local model of the potential function over a trust-region ball, while maintaining feasibility with respect to the linear equality constraints:
\begin{equation}
\label{eq:existing-IPTR}
\begin{aligned}
\min&~ 
\begin{cases} 
\nabla\phi(\boldsymbol{x}_t)^\top \boldsymbol{X}_{t} \boldsymbol{d} &\text{first-order IPTR}
\\\nabla\phi(\boldsymbol{x}_t)^\top \boldsymbol{X}_{t} \boldsymbol{d} +\frac{1}{2} \boldsymbol{d}^\top \boldsymbol{X}_{t}\nabla^2f(\boldsymbol{x}_t)\boldsymbol{X}_{t} \boldsymbol{d} & \text{second-order IPTR}
\end{cases}
\\
\text{s.t.}&~\boldsymbol{A}\boldsymbol{X}_{t}\boldsymbol{d}=\boldsymbol{0},\ \|\boldsymbol{d}\|\le \beta;
\end{aligned}
\end{equation}
where $\phi(\boldsymbol{x}_t)$ is the potential function defined in \sec{potential-center}, $\boldsymbol{A}$ is the constraint matrix in \eq{main-prob}, and $\beta < 1$ is the trust-region radius.

While the first- and second-order IPTR methods can converge to approximate first- and second-order KKT points respectively, they both suffer from high computational complexity in large-scale settings.
The first-order IPTR subproblem admits a closed-form solution given by $-\beta{\boldsymbol{P}_t \boldsymbol{X}_t\nabla\phi(\boldsymbol{x}_t)}/{\Vert \boldsymbol{P}_t  \boldsymbol{X}_t\nabla\phi(\boldsymbol{x}_t) \Vert}$, where $\boldsymbol{P}_t$ denotes the projection in the null space of $\boldsymbol{A}\boldsymbol{X}_t$. A major computational bottleneck of this step is that computing the search direction requires an explicit basis for the null space of $\boldsymbol{A}\boldsymbol{X}_t$. On the other hand, the second-order IPTR requires access to the exact Hessian $\nabla^2 f(\boldsymbol{x}_t)$, which incurs at least $\Omega(n^2)$ time and space complexity per iteration. Since $\boldsymbol{x}_t$ varies across iterations, continuously recomputing the projection matrix and the Hessian becomes computationally expensive.

To overcome these computational barriers, we first address the projection cost by developing an \robustalgo first-order IPTR algorithm that avoids recomputing the exact projection matrix from scratch. When $\boldsymbol{A}$ has full row rank, the projection matrix takes the explicit form $\boldsymbol{P}_t := \boldsymbol{I} - \boldsymbol{X}_t \boldsymbol{A}^\top (\boldsymbol{A} \boldsymbol{X}_t^2 \boldsymbol{A}^\top)^{-1} \boldsymbol{A} \boldsymbol{X}_t$. The main computational cost arises from forming and inverting the matrix $\boldsymbol{A}\boldsymbol{X}_t^2\boldsymbol{A}^\top$. To mitigate this cost, we maintain an approximate projector $\boldsymbol{R}_t:=\boldsymbol{I} - \boldsymbol{X}_{t}^{-1} \overline{\boldsymbol{X}}^2_{t} \boldsymbol{A}^\top (\boldsymbol{A} \overline{\boldsymbol{X}}^2_{t} \boldsymbol{A}^\top)^{-1} \boldsymbol{A}\boldsymbol{X}_{t}$,  where $\overline{\boldsymbol{X}}_t$ is a maintained approximation of $\boldsymbol{X}_t$. The matrix $\overline{\boldsymbol{X}}_t$ is updated sparsely using the binary-decomposition scheme. This sparsity structure enables efficient low-rank Sherman–Morrison–type updates of $(\boldsymbol{A}\overline{\boldsymbol{X}}_t^{2}\boldsymbol{A}^\top)^{-1}$ in \lem{time-inverse}.  The resulting projector preserves feasibility of the search directions while reducing the computational cost through structured low-rank updates. 

Consequently, at each iteration of our \robustalgo first-order IPTR algorithm, the search direction is computed as $-\beta {\boldsymbol{R}_t\boldsymbol{X}_t\nabla\phi(\boldsymbol{x}_t)}/{\Vert \boldsymbol{R}_t \boldsymbol{X}_t\nabla\phi(\boldsymbol{x}_t) \Vert}$. To ensure that $\boldsymbol{R}_t$ remains close to the exact projector $\boldsymbol{P}_t$, we maintain the multiplicative bounds $e^{-\delta}\boldsymbol{X}_t \leq \overline{\boldsymbol{X}}_t \leq e^{\delta}\boldsymbol{X}_t$, $\delta = \Theta(\varepsilon)$. These bounds imply $\Vert \boldsymbol{R}_t-\boldsymbol{P}_t \Vert = \Theta(\varepsilon)$ in \lem{R-P}, which ensures that the resulting search direction remains close to the ideal one and thereby preserves the robustness of the iteration. We show that the potential function decreases sufficiently at each accepted step and establish that, within $\mathcal{O}(1/\varepsilon^{2})$ iterations, the proposed \robustalgo first-order IPTR method either attains an $2\varepsilon$-KKT point or reaches a global minimizer.

To further compute approximate second-order KKT points without incurring the prohibitive $\Omega(n^2)$ cost of exact Hessian evaluations, we develop a negative-curvature finding subroutine (\proc{nega-curvature}). Given an iterate that satisfies the first-order KKT conditions but lies near a saddle point of the Lagrangian, the subroutine returns a direction that approximates the minimum-eigenvalue eigenvector of the Hessian. Since all iterates are constrained to the affine space $\boldsymbol{A}\boldsymbol{x}=\boldsymbol{b}$, we employ a projected power-iteration scheme to amplify the negative-curvature component while maintaining feasibility. As only first-order oracles are available, Hessian–vector products are approximated using finite differences of gradients, i.e., $\nabla f(\boldsymbol{x}+\tau \boldsymbol{v})-\nabla f(\boldsymbol{x})\approx \tau\nabla^2 f(\boldsymbol{x})\boldsymbol{v}$ for small $\tau$. Within both \algo{1st-order interior} and \algo{robust-nega}, this subroutine is invoked whenever the iterate meets the $2\varepsilon$-KKT conditions and the Lagrangian Hessian exhibits sufficiently negative curvature. The resulting negative-curvature step decreases the Lagrangian value and enables the algorithm to escape saddle-type regions, ultimately leading to a $(2\varepsilon,\sqrt{\varepsilon})$-KKT2 point.

\paragraph{Numerical experiments}
We also perform extensive numerical experiments. Through a variety of test instances, we show that our first-order IPTR algorithms are able to compute second-order approximate KKT points without access to Hessian information. In addition, we test our \robustalgo first-order IPTR algorithms on large-scale optimization problems with $(n,m)=(1000,500)$, $(3000,2000)$ and $(5000,4000)$ of the problem in Eq.~\eq{main-prob} to evaluate their practical efficiency. The experimental results show that our \robustalgo IPTR algorithm achieves up to a $2.48\times$ speedup over the existing first-order IPTR algorithm on these instances, demonstrating its scalability for large-scale problems.

\paragraph{Paper organization}
In \sec{prelim}, we introduce the notions of approximate first-order and second-order KKT points and state the assumptions required for the proposed IPTR framework. In \sec{robust-1st}, we present an \robustalgo first-order IPTR algorithm with a computationally efficient low-rank update. In \sec{2nd-KKT}, we introduce a negative-curvature finding procedure that enables the first-order IPTR framework to converge to approximate second-order KKT points using only first-order information. Numerical experiments are reported in \sec{expri}, where we present numerical evidence demonstrating the computational advantages of the proposed \robustalgo IPTR algorithms over existing IPTR-type methods.

\paragraph{Notation}
Throughout this paper, we use boldface letters to denote matrices and vectors. For any vector $\boldsymbol{x}$, we use $\boldsymbol{X}:=\text{diag}\{\boldsymbol{x}\}$ to denote the diagonal matrix with $\boldsymbol{x}$ on its diagonal. We define the feasible set as $\Omega:=\{\boldsymbol{A}\boldsymbol{x}=\boldsymbol{b}, \boldsymbol{x} \geq \boldsymbol{0}\}$, and its interior as $\Omega^\circ= \{\boldsymbol{A}\boldsymbol{x}=\boldsymbol{b}, \boldsymbol{x} > \boldsymbol{0}\}$. Unless stated otherwise, $\Vert \cdot \Vert$ denotes the Euclidean norm.

\subsection{Related work}
In this section, we review the literature relevant to our work. We begin by discussing prior research on finding approximate KKT points under various conditions, as well as the closely related problem of finding stationary points. Furthermore, given that our primary technical contributions lie in the \robustalgo interior trust-region method and negative curvature finding procedure to find $(2\varepsilon,\sqrt{\varepsilon})$-KKT2 points, we also summarize existing results regarding these techniques.

\paragraph{Approximate KKT points}
We review several works on computing approximate KKT points for constrained nonconvex optimization problems that are closely related to ours. For problems with linear equality constraints, augmented Lagrangian methods are widely used. Li et al.~\cite{li2021augmented} developed a first-order augmented Lagrangian method that attains an $2\varepsilon$-KKT point in $\mathcal{O}(\varepsilon^{-5/2} \log(1/\varepsilon))$ iterations for problems with a nonconvex objective and convex functional constraints. In a similar vein, Kong et al.~\cite{kong2023iteration} proposed an inner accelerated inexact proximal augmented Lagrangian method for solving linearly constrained smooth nonconvex composite optimization problems, achieving an iteration complexity of $\mathcal{O}(\varepsilon^{-5/2}\log^2(1/\varepsilon))$.  
These penalty methods generally do not keep the iterates strictly feasible. In the context of interior point methods, Haeser et al.~\cite{haeser2019optimality} proposed a first-order and second-order interior-point trust-region algorithm for linearly constrained problems without differentiability on the boundary. They established iteration complexities of $\mathcal{O}(1/\varepsilon^2)$ for finding $2\varepsilon$-KKT points and $\mathcal{O}(1/\varepsilon^{1.5})$ for $(2\varepsilon,\sqrt{\varepsilon})$-KKT2 points. More recently, Boob et al.~\cite{boob2025level} proposed a level-constrained first-order scheme that reduces the problem to a sequence of convex subproblems and achieves an $2\varepsilon$-KKT point in $\mathcal{O}(\varepsilon^{-2})$ iterations.
For KKT points in nonlinear programming, the seminal work by W{\"a}chter and Biegler~\cite{wachter2006implementation} developed IPOPT, a primal-dual interior point method with filter line-search globalization for large-scale nonlinear programming. Their work provides a practically robust implementation framework, incorporating feasibility restoration, second-order correction, inertia correction, and other practical enhancements. Both their work and ours use the interior point method, and our work adopts a simpler structure with linear equality constraints. We also utilize an approximate update mechanism to improve efficiency. For future research, building upon their theoretical foundation \cite{wachter2005line} to extend our approach to general nonlinear constraints, and exploring the robust implementation framework presented in their paper, are highly valuable questions.

\paragraph{Approximate stationary points in constrained optimization}
Distinct from the literature focusing on approximate KKT conditions, a parallel line of research aims to identify FOSP and SOSP for constrained optimization problems. The computation of stationary points has become a standard objective in modern mathematical programming, particularly for large-scale and nonconvex models where guarantees of global optimality are typically unavailable. 
In the context of FOSP, first-order methods for constrained optimization can be broadly divided into projection-based and projection-free schemes. Projected gradient methods extend gradient descent by projecting each iterate onto the feasible set \cite{beck2017first}. Projection-free methods such as Frank--Wolfe replace projections with a linear minimization oracle. For possibly nonconvex objectives, Lacoste-Julien et al.~\cite{lacoste2016convergence} established an $\mathcal{O}(\varepsilon^{-2})$ complexity bound for attaining an $\varepsilon$-FOSP. Recent extensions of this framework include Zeng et al.~\cite{zeng2024frank}. 
Beyond projected gradient and Frank--Wolfe methods, Muehlebach et al.~\cite{muehlebach2025accelerated} recently proposed a first-order algorithm for nonlinear constrained problems that avoids global optimization over the feasible region while ensuring convergence to an FOSP. When second-order information is accessible, the convergence rate can be improved; for instance, adaptive cubic regularization methods~\cite{CartisAdaptiveCubic,CartisEvaluationCubic,CartisEvaluationConstrained} are shown to achieve an $\varepsilon$-FOSP within $\mathcal{O}(\varepsilon^{-1.5})$ iterations.

Regarding SOSP, research has focused on leveraging higher-order information to escape saddle points in constrained settings. 
For problems with generic linear inequality constraints, Xie et al.~\cite{NEURIPS2020_1da546f2} analyzed the complexity of finding stationary points using projected gradient-based approaches. They established that first-order methods can reach an $(\varepsilon,\sqrt{\varepsilon})$-SOSP in $\mathcal{O}(\varepsilon^{-2.5})$ iterations, while their second-order counterparts improve this complexity to $\mathcal{O}(\max\{\varepsilon_{G}^{-2}, \varepsilon_H^{-1.5}\})$ for an $(\varepsilon_{G},\varepsilon_H)$-SOSP. 
For nonconvex equality constriant, a Newton-CG based augmented Lagrangian method proposed by \cite{He2023Newton-CG} can find an $(\varepsilon,\sqrt{\varepsilon})$-SOSP in $\widetilde{O}(\varepsilon^{-7/2})$ iterations. 
Mokhtari et al.~\cite{NEURIPS2018_069654d5} investigated minimizing smooth nonconvex functions over convex sets, specifically where a $\rho$-approximate solution to a quadratic program is computable. Their second-order method achieves an $(\varepsilon,\gamma)$-SOSP with a complexity of $\mathcal{O}(\max\{\varepsilon^{-2}, \rho^{-3}\gamma^{-3}\})$.

\paragraph{\Robustalgo interior point method} 

Since the seminal work by Karmarkar~\cite{Karmarkar1984polynomial}, interior point methods for linear programming have been extensively studied~\cite{anstreicher1997volumetric,renegar1988polynomial,nesterov1997self,vaidya1987algorithm,vaidya1993technique,LS2019solving}. In the theoretical computer science literature, a recent line of work studies robust interior point methods. Cohen, Lee, and Song~\cite{cohen2021solving} established a robust IPM framework where each update only needs to be computed approximately, which is used to prove that linear programming problems can be solved in current matrix multiplication time, whose technique is further derandomized and extended by~\cite{jiang2020faster,lee2019solving,LS2019solving,van2020deterministic,van2020solving}. This framework has also been used to develop algorithms for semidefinite programming~\cite{huang2022solving,jiang2020faster} and graph problems~\cite{chen2025maximum,van2020bipartite,van2023dynamic}. Motivated by this line of work, we follow the same intuition and refer to our method as an approximate interior point method.

\section{Preliminaries}
\label{sec:prelim}

\subsection{Approximate KKT points and sufficient conditions}
\label{sec:approx-kkt}
We consider linearly constrained optimization problems in which the objective function may be non-differentiable at the boundary of the feasible region. In such settings, classical KKT conditions are often inadequate. Exact stationarity may be ill-defined at boundary points, and second-order optimality conditions can be difficult to verify. To address these issues, Ref.~\cite{haeser2019optimality} provides a unified framework for characterizing approximate first- and second-order KKT points in the presence of boundary non-differentiability.
Consider the general constrained optimization problem
\begin{equation}
\begin{aligned}
\label{eq:general-prob}
\min_{\boldsymbol{x}\in \mathbb{R}^n}\quad &f(\boldsymbol{x})  \\
\mathrm{s.t.}\quad & h(\boldsymbol{x})=\boldsymbol{0}, c(\boldsymbol{x}) \geq \boldsymbol{0}.
\end{aligned}
\end{equation}
An approximate first-order KKT point is defined as follows.
\begin{definition}[{\cite[Definition 1]
{haeser2019optimality}}]
\label{defn:kkt1}
Given $\varepsilon>0$, a point $x\in\R^n$ is called an $\varepsilon$-KKT point for problem \eq{general-prob} when there exist approximate Lagrange multipliers $\lambda\in\R^m$ and $s\in\R^p_+$ with:
\begin{enumerate}
\item[(i)] $h(x)=0$, $c(x)>0$,
\item[(ii)] $\|\nabla f(x)+\sum_{i=1}^m\lambda_i\nabla h_i(x)-\sum_{i=1}^p s_i\nabla c_i(x)\|_{\infty}\leq\varepsilon$,
\item[(iii)]
$|c_i(x)s_i|\leq\varepsilon$ for all $i=1,\dots,p$.
\end{enumerate}
\end{definition}
\defn{kkt1} relaxes the exact KKT conditions by allowing controlled violations of stationarity and complementarity, measured in terms of the tolerance parameter $\varepsilon$. 
The following definition characterizes approximate second-order KKT points for functions that may be non-differentiable at the boundary. 
It modifies the second-order stationarity condition by incorporating weighted sums of matrices $\nabla c_i(x)\nabla c_i(x)^\top$, thereby yielding an appropriate notion of positive semidefiniteness near the boundary.
\begin{definition}[{\cite[Definition 2]{haeser2019optimality}}]
\label{defn:kkt2}
Given $\varepsilon_1,\varepsilon_2>0$, a point $x\in\R^n$ is called an $(\varepsilon_1,\varepsilon_2
)$-KKT2 point for problem \eq{general-prob} when there exist approximate Lagrange multipliers $\lambda\in\R^m$ and $s\in\R^p_+$ and a parameter $\theta\in\R^p_+$ with:
\begin{enumerate}
\item[(i)] $h(x)=0$, $c(x)>0$,
\item[(ii)] $\|\nabla f(x)+\sum_{i=1}^m\lambda_i\nabla h_i(x)-\sum_{i=1}^p s_i\nabla c_i(x)\|_{\infty}\leq\varepsilon_1$,
\item[(iii)]
$|c_i(x)s_i|\leq\varepsilon_1$ for all $i=1,\dots,p$,
\item[(iv)] $d^\top\left(\nabla^2f(x)+\sum_{i=1}^m\lambda_i\nabla^2h_i(x)-\sum_{i=1}^ps_i\nabla^2c_i(x)+\sum_{i=1}^p\theta_i\nabla c_i(x)\nabla c_i(x)^\top+\varepsilon_2\mathcal{I}\right)d\geq0,$ for all $d\in\R^n$ with $\nabla h_i(x)^\top d=0, i=1,\dots,m,$
\item[(v)] 
$|c_i(x)^2\theta_i|\leq\varepsilon_2$ for all $i=1,\dots,p$.
\end{enumerate}
\end{definition}

While \defn{kkt1} and \defn{kkt2} apply to general nonlinear constraints, \prop{kkt1} and \prop{kkt2} present sufficient conditions for approximate first-order KKT points and approximate second-order KKT points for linearly constrained optimization problem \eq{main-prob}.
\begin{proposition}[{\cite[Proposition 1]{haeser2019optimality}}]
\label{prop:kkt1}
    Let $\varepsilon > 0$. A point $\boldsymbol{x}$ is said to satisfy the sufficient condition for being an $\varepsilon$-KKT point if there exists a vector $\boldsymbol{v} \in \mathbb{R}^m$ such that:
    \begin{enumerate}
        \item $\boldsymbol{A}\boldsymbol{x}=\boldsymbol{b}$, $\boldsymbol{x} > 0$
        \item $\nabla f(\boldsymbol{x}) + \boldsymbol{A}^\top \boldsymbol{v} \geq -\varepsilon$
        \item $\Vert \boldsymbol{X}(\nabla f(\boldsymbol{x}) + \boldsymbol{A}^\top \boldsymbol{v}) \Vert_\infty \leq \varepsilon$
    \end{enumerate}
\end{proposition}

\begin{proposition}[{\cite[Proposition 2]{haeser2019optimality}}]
\label{prop:kkt2}
Let $\varepsilon_1,\varepsilon_2 > 0$. A point $\boldsymbol{x}$ satisfies the sufficient condition for being an $(\varepsilon_1,\varepsilon_2)$-KKT2 point if there exists a vector $\boldsymbol{v} \in \mathbb{R}^m$ such that:
    \begin{enumerate}
        \item $\boldsymbol{A}\boldsymbol{x}=\boldsymbol{b}$, $\boldsymbol{x} > 0$
        \item $\nabla f(\boldsymbol{x}) + \boldsymbol{A}^\top \boldsymbol{v} \geq -\varepsilon_1$
        \item $\Vert \boldsymbol{X}(\nabla f(\boldsymbol{x}) + \boldsymbol{A}^\top \boldsymbol{v}) \Vert_\infty \leq \varepsilon_1$
        \item $\lambda_{\min}(\boldsymbol{X}\nabla^2f(\boldsymbol{x})\boldsymbol{X})_{\boldsymbol{A}\boldsymbol{X}} \geq -\varepsilon_2$
    \end{enumerate}
    
\end{proposition}
The proofs of these propositions are omitted for brevity and can be found in \cite{haeser2019optimality}.

\subsection{Assumptions}
\label{sec:assump}
This section details the assumptions imposed on the objective function. Rather than assuming global Lipschitz continuity of the gradient or the Hessian in the original variable $\boldsymbol{x}$, we require Lipschitz-type bounds only along the locally feasible displacement directions $\boldsymbol{d}$. These displacements are measured in the scaled coordinate system defined by the diagonal matrix $\boldsymbol{X}$ and restricted to satisfy the linearized feasibility condition $\boldsymbol{A}\boldsymbol{X}\boldsymbol{d}=0$.
The following assumptions formalize these scaled Lipschitz properties for the gradient and the Hessian.

\begin{assumption}
\label{assum:grad-lip}
Suppose $f(\boldsymbol{x})$ is twice differentiable on $\Omega^\circ$. Then, for any $\boldsymbol{x} \in \Omega^\circ$ and any $\boldsymbol{d},\boldsymbol{d}^\prime \in \{\boldsymbol{d}: \Vert \boldsymbol{d} \Vert \leq \gamma ,\ \boldsymbol{A}\boldsymbol{X}\boldsymbol{d}=\boldsymbol{0} \}$, where $\gamma<1$ is a fixed constant, the following conditions hold: 
\begin{equation}
    \Vert \boldsymbol{X}\left(\nabla f(\boldsymbol{X}(\boldsymbol{e}+\boldsymbol{d}))-\nabla f(\boldsymbol{X}(\boldsymbol{e}+\boldsymbol{d}^\prime))\right) \Vert \leq l \Vert\boldsymbol{d}- \boldsymbol{d}^\prime\Vert.
\end{equation}
\end{assumption}
Under \assum{grad-lip}, define $g(\boldsymbol{d}) := f(\boldsymbol{X}(\boldsymbol{e}+\boldsymbol{d}))$.
Then $\nabla g(\boldsymbol{d}) = \boldsymbol{X}\nabla f(\boldsymbol{X}(\boldsymbol{e}+\boldsymbol{d}))$, and the assumption ensures that
$\nabla g$ is $l$-Lipschitz on the feasible displacement set. Consequently,
$\|\nabla^2 g(\boldsymbol{d})\| \le l$, and the standard quadratic upper bound holds:
\begin{equation}
g(\boldsymbol{d})
\le g(\boldsymbol{0})
+ \langle \nabla g(\boldsymbol{0}), \boldsymbol{d}\rangle
+ \frac{l}{2}\|\boldsymbol{d}\|^2 . 
\end{equation}
Expressed in the original variables, this implies $\|\boldsymbol{X}\nabla^2 f(\boldsymbol{x})\boldsymbol{X}\|\le l$ and
\begin{equation}
\label{eq:2nd-oder-upper}
f(\boldsymbol{X}(\boldsymbol{e}+\boldsymbol{d})) \le f(\boldsymbol{x}) + \langle \boldsymbol{X}\nabla f(\boldsymbol{x}), \boldsymbol{d}\rangle + \frac{l}{2}\|\boldsymbol{d}\|^2 .
\end{equation}

\begin{assumption}
\label{assum:hessian-lip}
Suppose $f(\boldsymbol{x})$ is twice differentiable on $\Omega^\circ$. Then, for any $\boldsymbol{x} \in \Omega^\circ$ and any $\boldsymbol{d},\boldsymbol{d}^\prime \in \{\boldsymbol{d}: \Vert \boldsymbol{d} \Vert \leq \gamma ,\ \boldsymbol{A}\boldsymbol{X}\boldsymbol{d}=\boldsymbol{0} \}$, where $\gamma<1$ is a fixed constant, the following conditions hold:
\begin{equation}
    \Vert \boldsymbol{X}\left(\nabla^2 f(\boldsymbol{X}(\boldsymbol{e}+\boldsymbol{d}))-\nabla^2 f(\boldsymbol{X}(\boldsymbol{e}+\boldsymbol{d}^\prime))\right)\boldsymbol{X} \Vert \leq \rho \Vert\boldsymbol{d}- \boldsymbol{d}^\prime\Vert.
\end{equation}
\end{assumption}
Analogous to the second-order bound \eq{2nd-oder-upper}, 
\assum{hessian-lip} implies the following third-order upper bound: 
\begin{equation}
    f(\boldsymbol{X}(\boldsymbol{e}+\boldsymbol{d})) \leq f(\boldsymbol{x}) + \langle  \boldsymbol{X} \nabla f(\boldsymbol{x}),\boldsymbol{d}\rangle + \frac{1}{2} \boldsymbol{d} \boldsymbol{X}\nabla^2f(\boldsymbol{x})\boldsymbol{X} \boldsymbol{d}+ \frac{\rho}{6} \Vert \boldsymbol{d} \Vert^3.
\end{equation}

All first-order algorithms presented in this paper rely on \assum{grad-lip}. When such first-order methods are augmented with a negative-curvature finding procedure to obtain second-order KKT points, an additional \assum{hessian-lip} is required. In comparison, Haeser et al.~\cite{haeser2019optimality} impose a slightly different and stronger set of assumptions to find approximate second-order KKT points, which we specify below.
\begin{assumption}[{\cite[Assumption 3(b)(c) and Assumption 4]{haeser2019optimality}}]
\label{assum:assum-haeser}
\leavevmode
\begin{enumerate}
    \item[(a)] Given $\boldsymbol{x}_0$ in $\Omega^\circ$, there exists $R \geq 1$ such that $\sup\{\Vert \boldsymbol{x} \Vert_\infty: f(\boldsymbol{x})\leq f(\boldsymbol{x}_0)\}\leq R$.
    \item[(b)] There exists $L\in \mathbb{R}$ such that $f(\boldsymbol{x}) \geq L$ for all $\boldsymbol{x} \in \Omega^\circ$.
    \item[(c)] Suppose $f(\boldsymbol{x})$ is twice differentiable on $\Omega^\circ$. Then, for any $\boldsymbol{x} \in \Omega^\circ$ and any $\boldsymbol{d},\boldsymbol{d}^\prime \in \{\boldsymbol{d}: \Vert \boldsymbol{d} \Vert \leq \gamma ,\ \boldsymbol{A}\boldsymbol{X}\boldsymbol{d}=\boldsymbol{0} \}$, where $\gamma<1$ is a fixed constant, the following conditions hold:
\begin{equation}
    \Vert \boldsymbol{X}\left(\nabla^2 f(\boldsymbol{X}(\boldsymbol{e}+\boldsymbol{d}))-\nabla^2 f(\boldsymbol{X}(\boldsymbol{e}+\boldsymbol{d}^\prime))\right) \Vert \leq \eta \Vert\boldsymbol{d}- \boldsymbol{d}^\prime\Vert.
\end{equation}
\end{enumerate}
\end{assumption}
It is easy to verify that \assum{assum-haeser}(c) implies \assum{hessian-lip}, as $\Vert\boldsymbol{X}\Vert$ is locally bounded. In this sense, \assum{hessian-lip} is a strictly weaker requirement. We note that \cite[Theorem 4]{haeser2019optimality} is based only on \assum{assum-haeser}(a) and (b) together with \assum{hessian-lip}. The algorithm converges to points that satisfy only Conditions 1, 3, and 4 of \prop{kkt2}, lacking the dual feasibility (Condition 2) required for an approximate second-order KKT point.

\subsection{Potential function and analytic center}
\label{sec:potential-center}
Let $0 < \varepsilon \leq \min\{\gamma, 1\}$, where $\gamma$ is the constant specified in \assum{grad-lip} and \assum{hessian-lip}. The first-order IPTR algorithm employs the following potential function:
\begin{equation}
    \phi(\boldsymbol{x})  = f(\boldsymbol{x}) - \varepsilon \sum_{i=1}^n \ln(x_i).
\end{equation}

The following lemma provides a useful inequality for controlling the logarithmic barrier term in the potential function.
\begin{lemma}
If $\boldsymbol{x} > \boldsymbol{0}$ and $\Vert \boldsymbol{d} \Vert \leq \beta < 1$, then
\begin{equation}
    -\sum_{i=1}^{n} \ln(x_i + x_i d_i) + \sum_{i=1}^{n} \ln(x_i) \leq \boldsymbol{e}^\top \boldsymbol{d} + \frac{\beta^2}{2(1-\beta)}.
\end{equation}
    
\end{lemma}

To initialize the algorithm, we require an approximate analytic center of the feasible region. Specifically, we assume that the initial point $\boldsymbol{x}_0 > \boldsymbol{0}$ satisfies
\begin{equation}
\label{eq:initial}
    -\sum_{i=1}^n \ln(x_i) \geq -\sum_{i=1}^n \ln(x_{0,i})-C_0 \qquad \forall \boldsymbol{x} \in \Omega^\circ.
\end{equation}
Finding an approximate analytic center amounts to solving $\min -\sum_{i=1}^n \log(x_i),\ \text{s.t.} \boldsymbol{A}\boldsymbol{x}=\boldsymbol{b}, \boldsymbol{x}\geq 0$. Since the objective function $-\sum_{i=1}^n \log(x_i)$ is an $n$-self-concordant barrier, an approximate analytic center can be computed in $\widetilde{\mathcal{O}}(\sqrt{n})$ damped Newton steps \cite{nemirovski2004interior,nesterov1994interior}. Each step requires $\mathcal{O}(n)$ time, leading to an overall time complexity of $\widetilde{\mathcal{O}}(n^{1.5})$. Throughout the paper, we assume that the algorithm is initialized at such an approximate analytic center.

\section{\Robustalgo First-order Interior Point Trust Region Algorithm}
\label{sec:robust-1st}

In this section, we present our \robustalgo first-order IPTR algorithm. Following the standard IPTR methodology, the search direction at iteration $t$ is ideally obtained by solving a trust-region subproblem over the constraints:
\begin{equation}
\begin{aligned}
\label{eq:1st-IPTR}
   \min\quad &\nabla \phi(\boldsymbol{x}_t)^\top \boldsymbol{X}_t \boldsymbol{d} \\
   \mathrm{s.t.}\quad &\boldsymbol{A}\boldsymbol{X}_t \boldsymbol{d} = 0,\ \Vert \boldsymbol{d} \Vert \leq \beta,
\end{aligned}
\end{equation}
where $\beta$ denotes the trust-region radius. When $\boldsymbol{P}_t \boldsymbol{X}_t \nabla \phi(\boldsymbol{x}_t) \neq \boldsymbol{0}$, the exact solution to \eq{1st-IPTR} admits the closed-form expression:
\begin{equation}
\label{eq:exact-d}
\boldsymbol{d}_{t} := - \beta \frac{\boldsymbol{P}_t \boldsymbol{X}_t \nabla \phi(\boldsymbol{x}_t)}{\Vert \boldsymbol{P}_t \boldsymbol{X}_t \nabla \phi(\boldsymbol{x}_t) \Vert},
\end{equation}
where $\boldsymbol{P}_t \coloneqq \boldsymbol{I} - \boldsymbol{X}_t \boldsymbol{A}^\top (\boldsymbol{A} \boldsymbol{X}_t^2 \boldsymbol{A}^\top)^{-1} \boldsymbol{A} \boldsymbol{X}_t$ is the orthogonal projection onto the null space of $\boldsymbol{A}\boldsymbol{X}_t$. 
After computing the search direction $\boldsymbol{d}_t$, the next iterate is updated as $\boldsymbol{x}_{t+1} = \boldsymbol{x}_t+\boldsymbol{X}_t\boldsymbol{d}_t$. 
The main computational challenge in evaluating \eq{exact-d} lies in computing the projection matrix $\boldsymbol{P}_t$. Since the scaling matrix $\boldsymbol{X}_t$ changes at every iteration, computing this exact projection from scratch incurs a prohibitive complexity of $\mathcal{O}(nm^{\omega-1})$ at each iteration. This forms a significant computational bottleneck for large-scale problems.

Our \algo{robust-1st-order} overcomes this bottleneck by reducing the average per-iteration complexity to $\mathcal{O}(mn)$. This makes each iteration as cheap as a single matrix-vector multiplication. The key idea is to replace the frequently changing exact diagonal $\boldsymbol{X}_t$ in the inversion term $(\boldsymbol{A} \boldsymbol{X}_t^2 \boldsymbol{A}^\top)^{-1}$ with a sparsely updated diagonal approximation $\overline{\boldsymbol{X}}_{t}$. This yields an approximate projection matrix $\boldsymbol{R}_t$:
\begin{equation}
\label{eq:R_t}
    \boldsymbol{R}_t \coloneqq  \boldsymbol{I} - \boldsymbol{X}_{t}^{-1} \overline{\boldsymbol{X}}^2_{t} \boldsymbol{A}^\top (\boldsymbol{A} \overline{\boldsymbol{X}}^2_{t} \boldsymbol{A}^\top)^{-1} \boldsymbol{A}\boldsymbol{X}_{t}.
\end{equation}
By design, this specific construction of $\boldsymbol{R}_t$ not only ensures that the search direction strictly resides within the null space of $\boldsymbol{A}\boldsymbol{X}_t$, but also allows rapid updates since the diagonal matrix $\overline{\boldsymbol{X}}_t$ is modified sparsely. Based on this efficient projection, we scale the step to a norm of $\beta$ whenever the projected gradient is non-zero, and set it to zero otherwise. Accordingly, the approximate search direction $\widetilde{\boldsymbol{d}}_{t}$ is computed as
\begin{equation}
\label{eq:d-tilde}
\widetilde{\boldsymbol{d}}_{t} := 
\begin{cases}
    \boldsymbol{0}, & \text{if } \boldsymbol{R}_t \boldsymbol{X}_t \nabla \phi(\boldsymbol{x}_t) = \boldsymbol{0}, \\
    - \beta \frac{\boldsymbol{R}_t \boldsymbol{X}_t \nabla \phi(\boldsymbol{x}_t)}{\Vert \boldsymbol{R}_t \boldsymbol{X}_t \nabla \phi(\boldsymbol{x}_t) \Vert}, & \text{otherwise.}
\end{cases}
\end{equation}

The sparse update $\overline{\boldsymbol{X}}_t$ can be viewed as a lazy update of $\boldsymbol{X}_t$. The motivation is that the iterate is updated as $\boldsymbol{x}_{t+1} = \boldsymbol{x}_t + \boldsymbol{X}_t \widetilde{\boldsymbol{d}}_t$. Since the Euclidean norm of $\widetilde{\boldsymbol{d}}_t$ is bounded by $\beta$, where $\beta$ is typically controlled by the approximation tolerance $\varepsilon$ for approximate KKT points, the relative change in $\boldsymbol{x}_t$ is small. Consequently, only a small subset of coordinates undergo significant changes between $\boldsymbol{x}_t$ and $\boldsymbol{x}_{t+1}$. When updating the approximation $\overline{\boldsymbol{X}}_t$, we therefore modify $\overline{\boldsymbol{X}}_t$ only on those coordinates with relatively large changes, while leaving the remaining coordinates unchanged. The updated entries are taken from $\boldsymbol{x}_{t+1}$, which gives the new approximation $\overline{\boldsymbol{X}}_{t+1}$.

This lazy-update scheme is useful because it turns the change from $\overline{\boldsymbol{X}}_t$ to $\overline{\boldsymbol{X}}_{t+1}$ into a sparse diagonal modification. To implement it, we employ the $\mathtt{SelectVector}$ algorithm from \cite[Algorithm~4]{lee2021tutorial} to identify the coordinates with noticeable relative changes and refresh only those entries of $\overline{\boldsymbol{X}}_t$. As a result, the difference $\overline{\boldsymbol{X}}_{t+1}^2-\overline{\boldsymbol{X}}_t^2$ is supported on only a small number of coordinates, and hence $\boldsymbol{A}\overline{\boldsymbol{X}}_{t+1}^2\boldsymbol{A}^\top$ is obtained from $\boldsymbol{A}\overline{\boldsymbol{X}}_t^2\boldsymbol{A}^\top$ by a low-rank update. Therefore, its inverse can be maintained efficiently by the Sherman-Morrison-Woodbury formula, which substantially reduces the per-iteration computational cost compared to standard first-order IPTR methods.

Our \robustalgo first-order IPTR algorithm is presented in \algo{robust-1st-order}.  The algorithm is initialized by selecting an approximate analytic center $\boldsymbol{x}_0$ and setting $\overline{\boldsymbol{x}}_0=\boldsymbol{x}_0$. At iteration $t$, it approximately solves the subproblem \eq{1st-IPTR} using the projection matrix $\boldsymbol{R}_t$, computes the direction $\widetilde{\boldsymbol{d}}_t$ by \eq{d-tilde}, and updates $\boldsymbol{x}_{t+1}=\boldsymbol{x}_t+\boldsymbol{X}_t\widetilde{\boldsymbol{d}}_t$. It then checks whether the potential function decreases sufficiently. If not, the algorithm returns $\boldsymbol{x}_t$. Otherwise, it updates $\overline{\boldsymbol{x}}_{t+1}$ by applying $\mathtt{SelectVector}$ to the logarithms of the iterates.

\vspace{4mm}
\begin{algorithm}[!htb]
    \SetAlgoLined
    \caption{\Robustalgo First-order Interior Point Trust Region Algorithm}
    \label{algo:robust-1st-order}
    
    Initialize $\boldsymbol{x}_0$ as an approximate analytic center\;
    
    $T \leftarrow \frac{\left(f(\boldsymbol{x}_0)- f(\boldsymbol{x}^*) + (C_0-1)\varepsilon\right)(l + 2\varepsilon + 2)}{\varepsilon^2} , \beta\leftarrow \varepsilon/(l+2\varepsilon+2), \delta \leftarrow \min(\varepsilon/(15L_\phi), \beta / (92 L_{\phi})) $ and $\overline{\boldsymbol{x}}_{0} \leftarrow \boldsymbol{x}_{0}$\;
    
    \For{$t=0,\ldots,T-1$}{
    Approximate the subproblem \eq{1st-IPTR} using the projection matrix
    \begin{equation*}
        \boldsymbol{R}_t: =  \boldsymbol{I} - \boldsymbol{X}_{t}^{-1} \overline{\boldsymbol{X}}^2_{t} \boldsymbol{A}^\top 
    (\boldsymbol{A} \overline{\boldsymbol{X}}^2_{t} \boldsymbol{A}^\top)^{-1} \boldsymbol{A}\boldsymbol{X}_{t};
    \end{equation*}
    \If{$\boldsymbol{R}_t \boldsymbol{X}_t \nabla \phi(\boldsymbol{x}_t)=\boldsymbol{0}$}{$\widetilde{\boldsymbol{d}}_{t} := \boldsymbol{0}$\;}
    \Else{
    \begin{equation*}
    \widetilde{\boldsymbol{d}}_{t}:=- \beta \frac{\boldsymbol{R}_t \boldsymbol{X}_t \nabla \phi(\boldsymbol{x}_t)}{\Vert \boldsymbol{R}_t \boldsymbol{X}_t \nabla \phi(\boldsymbol{x}_t) \Vert};
    \end{equation*}
    }
    
    $\boldsymbol{x}_{t+1} \leftarrow \boldsymbol{x}_t + \boldsymbol{X}_t \widetilde{\boldsymbol{d}}_{t}$\;
    
    \If{$\phi(\boldsymbol{x}_{t+1})-\phi(\boldsymbol{x}_{t}) >- \frac{\varepsilon^2}{2l + 4\varepsilon + 4} $}{Return $\boldsymbol{x}_t$\;}
    
    $\ln\overline{\boldsymbol{x}}_{t+1} = \mathtt{SelectVector}(\ln\overline{\boldsymbol{x}}_{t},\ln\boldsymbol{x}_{0},\ln\boldsymbol{x}_{1},\ldots,\ln\boldsymbol{x}_{t+1},\delta) $\;
    }
\end{algorithm}
\vspace{4mm}

\subsection{Sparse update of \texorpdfstring{$\overline{\boldsymbol{X}}$}{Xbar}}

In this subsection, we analyze the sparse update scheme for $\overline{\boldsymbol{X}}_t$ and the resulting time complexity of \algo{robust-1st-order}. We first show that the change of $\ln \boldsymbol{x}_t$ across iterations is small. We then describe how we maintain a sparsely updated approximation $\overline{\boldsymbol{x}}_t$ of $\boldsymbol{x}_t$. Finally, we present \lem{time-inverse} and \prop{robust-time} that establish the time complexity of our algorithm.

\begin{lemma}
\label{lem:log-change}
    Suppose $\boldsymbol{x}_t \in \mathbb{R}_{++}^n$ and $\boldsymbol{x}_{t+1} = \boldsymbol{X}_t(\boldsymbol{e} + \boldsymbol{d}_t)$ with $\Vert \boldsymbol{d}_t \Vert_2 \leq \beta$ and $0 < \beta < \frac{1}{2}$. Then it holds that $\Vert \ln \boldsymbol{x}_{t+1} - \ln \boldsymbol{x}_{t} \Vert_2 \leq 2\beta$.
\end{lemma}
\begin{proof}
For each $i$, we have $x_{t+1,i} = x_{t,i}(1 + d_{t,i})$.
Taking logarithms gives $
\ln x_{t+1,i}-\ln x_{t,i}=\ln(1+d_{t,i})$. Using the standard inequality for $|d_{t,i}| < \tfrac{1}{2}$, 
\begin{equation}
   -2\vert d_{t,i} \vert  \leq-\frac{\vert d_{t,i} \vert}{1-\vert d_{t,i} \vert} \leq \ln(1+d_{t,i}) \leq \vert d_{t,i} \vert, 
\end{equation}
we obtain
\begin{equation}
    \Vert \ln\boldsymbol x_{t+1}-\ln\boldsymbol x_t \Vert_2
=\Bigl(\sum_i \bigl|\ln(1+d_{t,i})\bigr|^2\Bigr)^{1/2}
\le 2\Bigl(\sum_i |d_{t,i}|^2\Bigr)^{1/2}
=2\Vert\boldsymbol d_t\Vert_2 \leq 2\beta.
\end{equation}
\end{proof}

\lem{log-change} shows that the change in $\ln \boldsymbol{x}_t$ across successive iterations is small. Based on this property, we apply the $\mathtt{SelectVector}$ algorithm from \cite[Algorithm~4]{lee2021tutorial} to maintain a sparse update of $\overline{\boldsymbol{X}}_t$. The algorithm takes $(\ln \overline{\boldsymbol{x}}_{t-1}, \ln \boldsymbol{x}_0, \ln \boldsymbol{x}_1, \ldots, \ln \boldsymbol{x}_t, \delta)$ as input and outputs $\ln \overline{\boldsymbol{x}}_t$ satisfying the following properties.
\begin{lemma}[{\cite[Lemma 19]{lee2021tutorial}}]
\label{lem:SelectVector}
Given vectors $\ln\boldsymbol{x}_0,\ln\boldsymbol{x}_1,\ln\boldsymbol{x}_2,\ldots$
arriving in a stream, and satisfies that $\Vert \ln\boldsymbol{x}_{t+1}-\ln\boldsymbol{x}_{t} \Vert_{2}\leq 2\beta$
for all $t$. For any $\frac{1}{2}>\delta>0$, define the vector $\ln\overline{\boldsymbol{x}}_0=\ln\boldsymbol{x}_0$
and $\ln\overline{\boldsymbol{x}}_{t}=\mathtt{SelectVector}(\ln\overline{\boldsymbol{x}}_{t-1},\ln\boldsymbol{x}_0,\ln\boldsymbol{x}_1,\ldots,\ln\boldsymbol{x}_t,\delta)$.
Then, we have that
\begin{itemize}
\item[(i)] $\Vert \ln\overline{\boldsymbol{x}}_{t}-\ln\boldsymbol{x}_t\Vert_{\infty}\leq\delta$ for all $k$.
\item[(ii)] $\Vert \ln\overline{\boldsymbol{x}}_{t} -\ln\overline{\boldsymbol{x}}_{t-1} \Vert_{0}\leq O(2^{2l_{t}}(2\beta/\delta)^{2}\log^{2}n)$
where $l_{t}$ is the largest integer $l$ with $t=0\mod2^{l}$. 
\end{itemize}
\end{lemma}
The $\ell_\infty$-bound in \lem{SelectVector}(i) immediately yields the following component wise comparison between $\overline{\boldsymbol{x}}_t$ and $\boldsymbol{x}_t$:
\begin{equation}
     e^{-\delta} \boldsymbol{x}_t\leq \overline{\boldsymbol{x}}_t   \leq e^{\delta} \boldsymbol{x}_t. 
\end{equation}
Hence, $\overline{\boldsymbol{X}}_t$ remains a multiplicative approximation of $\boldsymbol{X}_t$. Meanwhile, \lem{SelectVector}(ii) bounds the number of diagonal entries that change between consecutive iterations. This sparsity in the updates of $\overline{\boldsymbol{X}}_t$ will allow us to improve the time complexity of computing the matrix inverse $(\boldsymbol{A}\overline{\boldsymbol{X}}_t^{2}\boldsymbol{A}^\top)^{-1}$, as formalized in \lem{time-inverse}.
\begin{lemma}
\label{lem:time-inverse}
Let $\Vert \overline{\boldsymbol{x}}_{t+1} - \overline{\boldsymbol{x}}_{t} \Vert_0 = q_t$.
Given $(\boldsymbol{A} \overline{\boldsymbol{X}}_{t}^{2} \boldsymbol{A}^\top)^{-1}$, the inverse $(\boldsymbol{A} \overline{\boldsymbol{X}}_{t+1}^{2} \boldsymbol{A}^\top)^{-1}$ can be updated in time $\mathcal{O}(m^{2} q_t^{\omega - 2})$ when $q_t \le m$, and in time $\mathcal{O}(nm^{\omega-1})$ when $m< q_t \le n$, where $\omega$ denotes the exponent of matrix multiplication.
\end{lemma}
\begin{proof}
    Denote $\boldsymbol{K}_t:=\boldsymbol{A} \overline{\boldsymbol{X}}_{t}^{2} \boldsymbol{A}^\top$ and  $\boldsymbol{K}_{t+1}:=\boldsymbol{A} \overline{\boldsymbol{X}}_{t+1}^{2} \boldsymbol{A}^\top = \boldsymbol{A} (\overline{\boldsymbol{X}}_{t}^{2}+\overline{\boldsymbol{X}}_{t+1}^{2}-\overline{\boldsymbol{X}}_{t}^{2}) \boldsymbol{A}^\top$.  
    Let $\mathcal{I}_t :=\{i:\overline{\boldsymbol{x}}_{t+1,i} \neq \overline{\boldsymbol{x}}_{t,i}\}$ be the index set of the updated coordinates, so that $\vert \mathcal{I}_t\vert = q_t$. Since only the coordinates in $\mathcal{I}_t$ are updated, the diagonal matrix $\overline{\boldsymbol{X}}_{t+1}^{2} - \overline{\boldsymbol{X}}_{t}^{2}$ can be written as $\boldsymbol{U}\boldsymbol{C}\boldsymbol{U}^\top$, 
     where $\boldsymbol{U} \in \mathbb{R}^{n \times q_t}$ consists of the columns of the identity matrix $\boldsymbol{I}_n$ indexed by $\mathcal{I}_t$, and $\boldsymbol{C} \in \mathbb{R}^{q_t \times q_t}$ is diagonal. 
     Hence, $\boldsymbol{K}_{t+1}=\boldsymbol{A} (\overline{\boldsymbol{X}}_{t}^{2}+\boldsymbol{U}\boldsymbol{C}\boldsymbol{U}^\top) \boldsymbol{A}^\top = \boldsymbol{K}_t + \boldsymbol{A} \boldsymbol{U}\boldsymbol{C}\boldsymbol{U}^\top \boldsymbol{A}^\top$. The product $\boldsymbol{A} \boldsymbol{U}$ corresponds to a block of $\boldsymbol{A}$ containing $m$ rows and $q_t$ selected columns. We denote this submatrix explicitly as $\boldsymbol{A}_{:, \mathcal{I}_t} := \boldsymbol{A}\boldsymbol{U}$, which is the $m \times q_t$ submatrix of $\boldsymbol{A}$ consisting of the columns indexed by $\mathcal{I}_t$. Applying the Woodbury matrix identity gives

    \begin{equation}
    \begin{aligned}
        \boldsymbol{K}_{t+1}^{-1} 
        &= \left(\boldsymbol{K}_t + \boldsymbol{A}_{:, \mathcal{I}_t}\boldsymbol{C} \boldsymbol{A}_{:, \mathcal{I}_t}^\top\right)^{-1} \\
        &= \boldsymbol{K}_t^{-1} - \boldsymbol{K}_t^{-1}\boldsymbol{A}_{:, \mathcal{I}_t}\left( \boldsymbol{C}^{-1} +\boldsymbol{A}_{:, \mathcal{I}_t}^\top \boldsymbol{K}_t^{-1}  \boldsymbol{A}_{:, \mathcal{I}_t} \right)^{-1} \boldsymbol{A}_{:, \mathcal{I}_t}^\top \boldsymbol{K}_t^{-1}.
    \end{aligned}
    \end{equation}
    When $q_t \leq m$, computing $\left( \boldsymbol{C}^{-1} +\boldsymbol{A}_{:, \mathcal{I}_t}^\top \boldsymbol{K}_t^{-1}  \boldsymbol{A}_{:, \mathcal{I}_t} \right)^{-1}$ requires $\mathcal{O}(m^2q_t^{\omega-2} + mq_t^{\omega-1} + q_t^{\omega}) = \mathcal{O}(m^2q_t^{\omega-2})$ time, which accounts for one matrix multiplication of size $q_t \times m$  with  $m \times m$, one multiplication of size $q_t \times m$ with $m \times q_t$, and one inversion of size $(q_t \times q_t)$. The time complexity of multiplying an $m\times n$ matrix with an $n \times q$ matrix is $\mathcal{O}(mnq\min\{m,n,q\}^{\omega-3})$ because the rectangular matrix multiplication can be decomposed into blocks of square submatrices along its smallest dimension. Computing $\boldsymbol{A}_{:, \mathcal{I}_t}^\top \boldsymbol{K}_t^{-1}$ and $\boldsymbol{K}_t^{-1}\boldsymbol{A}_{:, \mathcal{I}_t}$ also requires $\mathcal{O}(m^2q_t^{\omega-2})$ time. Finally, multiplying $\boldsymbol{K}_t^{-1}\boldsymbol{A}_{:, \mathcal{I}_t}$, $\left( \boldsymbol{C}^{-1} +\boldsymbol{A}_{:, \mathcal{I}_t}^\top \boldsymbol{K}_t^{-1}  \boldsymbol{A}_{:, \mathcal{I}_t} \right)^{-1}$, and $\boldsymbol{A}_{:, \mathcal{I}_t}^\top \boldsymbol{K}_t^{-1}$ together has a time complexity of $\mathcal{O}(mq_t^{\omega-1} + m^2q_t^{\omega-2}) = \mathcal{O}(m^2q_t^{\omega-2})$, due to one multiplication of size $m \times q_t$ with $q_t \times q_t$ and one multiplication of size $m \times q_t$ with $q_t \times m$. In summary, when $q_t \leq m$, the overall time complexity for computing the inverse $\boldsymbol{K}_{t+1}^{-1}$ is $\mathcal{O}(m^2q_t^{\omega-2})$.

    Now consider the case \(m < q_t \le n\). In this case, computing
    \begin{equation*}
        \left( \boldsymbol{C}^{-1} + \boldsymbol{A}_{:, \mathcal{I}_t}^\top \boldsymbol{K}_t^{-1} \boldsymbol{A}_{:, \mathcal{I}_t} \right)^{-1}
    \end{equation*}
    requires \(\mathcal{O}(q_t^\omega)\) time, as it involves the inversion of a \(q_t \times q_t\) matrix. As this time complexity may exceed the cost of recomputing the inverse from scratch, we instead directly form \(\boldsymbol{K}_{t+1}=\boldsymbol{A}\overline{\boldsymbol{X}}_{t+1}^2\boldsymbol{A}^\top\) and compute its inverse anew. Forming \(\boldsymbol{K}_{t+1}\) requires \(\mathcal{O}(nm^{\omega-1})\) time, which accounts for one matrix multiplication of size \(m\times n\) with \(n\times m\), and computing the inverse of \(\boldsymbol{K}_{t+1}\) requires \(\mathcal{O}(m^\omega)\) time. Therefore, in this case, the overall time complexity is \(\mathcal{O}(nm^{\omega-1})\).
\end{proof}

We now combine \lem{SelectVector} and \lem{time-inverse} to derive the overall running time of \algo{robust-1st-order}. The former characterizes the sparsity of the updates to $\overline{\boldsymbol{X}}_t$, and the latter shows how this sparsity translates into a reduced cost for maintaining the inverse $(\boldsymbol{A}\overline{\boldsymbol{X}}_t^2\boldsymbol{A}^\top)^{-1}$ in the approximate projection $\boldsymbol{R}_t$. 
\begin{proposition}
\label{prop:robust-time}
Suppose \algo{robust-1st-order} runs for $T$ iterations, and let $\delta = \Theta(\beta)$. Then, the overall time complexity is upper bounded by 
\begin{equation}
    \widetilde{\mathcal{O}}\left(nm^{\omega-1} + mnT\right).
\end{equation}
\end{proposition}
\begin{proof}
    The total computational cost of \algo{robust-1st-order} arises from three main components: (1) updating the inverse in the projection matrix $\boldsymbol{R}_t$, (2) computing the approximate direction $\widetilde{\boldsymbol{d}}_t$ and updating the iterates $\boldsymbol{x}_{t+1}$, and (3) executing the $\mathtt{SelectVector}$ procedure that determines the sparse update pattern of $\overline{\boldsymbol{X}}_t$. We analyze each component in turn.
    \begin{enumerate}
        \item Cost of updating inverse  $(\boldsymbol{A} \overline{\boldsymbol{X}}^2_{t} \boldsymbol{A}^\top)^{-1}$: 
        According to \lem{SelectVector}, the maximum number of coordinate changes between 
        $\ln \overline{\boldsymbol{x}}_t$ and $\ln \overline{\boldsymbol{x}}_{t-1}$ is bounded by 
        $\mathcal{O}(4^{l_t}\log^2 n)$.  
        Let $q_t$ denote the actual number of changed coordinates at iteration $t$. Then, there exists a constant $C \ge 1$ such that
        \begin{equation}
        \label{eq:q_t}
            q_t \leq \min \left\{ n,C4^{l_t}\log^2 n \right\}\text{, for $t=1,2,\ldots,T$}.
        \end{equation}
        We categorize $q_t$ according to the value of $l_t$. Recall that $l_t$ is defined as the largest integer $l \leq \lceil \log n \rceil$ such that $t \equiv 0 \bmod{2^{l}}$. Let $N_l$ denote the number of indices $t$ that share the same value of $l$, i.e.,
        \begin{equation}
        \label{eq:N_l}
        \begin{aligned}
            N_l&:=\vert\{t\in[1,T],\; t \equiv 0 \bmod 2^l \text{ and } t \not \equiv 0 \bmod 2^{l+1}\}\vert \text{ for $0\leq l \leq \lceil \log n \rceil-1$,} \\ N_{\lceil \log n \rceil}&:=\vert\{t\in[1,T],\; t \equiv 0 \bmod 2^{\lceil \log n \rceil}\}\vert.
        \end{aligned}
        \end{equation}
        For $0\leq l \leq \lceil \log n \rceil-1$, the size of $N_l$  equals the number of integers in $[1, T]$ that are divisible by $2^l$ but not by $2^{l+1}$
        \begin{equation}
             N_l  = \left\lfloor  \frac{T}{2^l}\right\rfloor - \left\lfloor  \frac{T}{2^{l+1}}\right\rfloor.
        \end{equation}
        In addition, $N_{\lceil \log n \rceil} = \left\lfloor  \frac{T}{2^{\lceil \log n \rceil}}\right\rfloor$. Combining these bounds, we obtain $N_l \leq \frac{T}{2^l}$ for $0\leq l \leq \lceil \log n \rceil$.

        At each iteration $t$, if $q_t$ coordinates are updated, then by \lem{time-inverse}, when $q_t \le m$, the inverse $(\boldsymbol{A} \overline{\boldsymbol{X}}_{t}^2 \boldsymbol{A}^\top)^{-1}$ can be updated from the previous inverse via a rank-$q_t$ Woodbury update in time $\mathcal{O}(m^2 q_t^{\omega-2})$, where $\omega$ denotes the exponent of matrix multiplication. When $m < q_t \le n$, we compute $(\boldsymbol{A} \overline{\boldsymbol{X}}_{t}^2 \boldsymbol{A}^\top)^{-1}$ directly rather than applying the Woodbury update. The time complexity of this direct computation is $\mathcal{O}(nm^{\omega-1})$, which accounts for forming $\boldsymbol{A} \overline{\boldsymbol{X}}_{t}^2 \boldsymbol{A}^\top$ and computing its inverse. Therefore, we apply the Woodbury update only when $q_t \le m$, and compute the inverse directly otherwise.
        We set $l^*$ to ensure that for $l = 0,1,\ldots, l^*$, $C4^{l}\log^2 n \leq m$.
        \begin{equation}
            l^* := \frac{1}{2} \left\lfloor \left( \log_2 \left(\frac{m}{C\log^2 n}\right)\right) \right\rfloor.
        \end{equation}

        For $q_t \leq m$, the time complexity is 
        \begin{equation}
        \begin{aligned}
            \sum_{t=1}^{T} m^2 q_t^{\omega-2} 
            &\leq \sum_{l=0}^{\lceil \log n \rceil} m^2 N_l \left(C4^l \log^2n\right)^{\omega-2} \\
            & \leq C^{\omega-2} m^2 T \log^{2(\omega-2)}n \sum_{l=0}^{\lceil \log n \rceil} 2^{(2\omega-5)l} \\
            & = \widetilde{\mathcal{O}}\left(m^2T\right).
        \end{aligned}
        \end{equation}

        For $q_t > m$, which implies $l_t > l^*$ and $T > \sqrt{\frac{m}{C\log^2 n}}$ (otherwise we would have $N_{l_t} = 0$), the time complexity is 
        \begin{equation}
        \begin{aligned}
            \sum_{t=1}^{T} 1_{\{q_t > m\}} nm^{\omega-1}
            & \leq \sum_{t=1}^{T} 1_{\{C4^{l_t}\log^2 n > m\}} nm^{\omega-1}\\
           &=\sum_{l=l^*+1}^{\lceil \log n \rceil}  N_l  nm^{\omega-1} \\
            & \leq T nm^{\omega-1}  \sum_{l=l^*+1}^{\lceil \log n \rceil} \frac{1}{2^l} \\
            &  \leq T nm^{\omega-1} \left( \frac{1}{2^{l^*}} \right) \\
            & \leq T nm^{\omega-1} \frac{\sqrt{C}\log n}{\sqrt{m}} \\
            & =\widetilde{\mathcal{O}}\left(nm^{\omega-1.5} T\right).
        \end{aligned}
        \end{equation}

        \item  Cost of computing $\widetilde{\boldsymbol{d}}_t$ and updating $\boldsymbol{x}_{t+1}$. After computing the inverse $(\boldsymbol{A}\overline{\boldsymbol{X}}_{t}^{2}\boldsymbol{A}^\top)^{-1}$, we obtain $\widetilde{\boldsymbol{d}}_t$ according to \eq{d-tilde}. This step requires a time complexity of $\mathcal{O}(mn)$ per iteration due to the matrix–vector multiplication. Updating $\boldsymbol{x}_{t+1}$ incurs an additional cost of $\mathcal{O}(n)$ per iteration. Therefore, the overall time complexity of this part over $T$ iterations is $\mathcal{O}(mnT)$.
        
        \item Cost of $\mathtt{SelectVector}$ algorithm: 
        For each iteration $t \in [1,T]$, the $\mathtt{SelectVector}$ algorithm in \cite[Algorithm~4]{lee2021tutorial} examines every level $l\in[0,\lceil\log_2 n\rceil]$ satisfying $t =0 \bmod  2^l$, and compares each coordinate $i\in [n]$ between $\ln \boldsymbol{x}_{t,i}$ and $\ln \boldsymbol{x}_{(t-2^l),i}$. This process requires a total computational cost of
        \begin{equation}
            \sum_{l=1}^{\lceil\log_2 n\rceil} n l N_l \leq  nT\sum_{l=1}^{\lceil\log_2 n\rceil} \frac{l}{2^l} \leq 2nT,
        \end{equation}
        where $N_l$ is defined in \eq{N_l}. Therefore, the time complexity of this part is $\mathcal{O}(nT)$.
    \end{enumerate}
Combining the above bounds and including the initial cost $\widetilde{\mathcal{O}}(nm^{\omega-1})$ of forming $(\boldsymbol{A}\overline{\boldsymbol{X}}_{0}^{2}\boldsymbol{A}^\top)^{-1}$ gives the claimed complexity bound $\widetilde{\mathcal{O}}(nm^{\omega-1}+mnT)$.
\end{proof}

\subsection{Convergence of potential function}

We now show that the approximate projector $\boldsymbol{R}_t$ used in \algo{robust-1st-order} still preserves the descent property of the exact IPTR step. The argument has two parts. First, we compare $\boldsymbol{R}_t$ with the exact orthogonal projector $\boldsymbol{P}_t$ and bound their difference using the multiplicative closeness between $\overline{\boldsymbol{X}}_t$ and $\boldsymbol{X}_t$ in \lem{R-P}. Then based on this comparison, we show that each iteration either yields a sufficient decrease in the potential function or certifies the approximate KKT optimality in \prop{potential-decre}. Finally, by combining the per-iteration complexity bound in \prop{robust-time} with the convergence guarantee in \prop{potential-decre}, we obtain \thm{robust-complexity} for \algo{robust-1st-order}.

\begin{lemma}
\label{lem:R-P}
Let $\boldsymbol{A} \in \mathbb{R}^{m \times n}$ be a full-row-rank matrix, and let $\boldsymbol{x} \in \mathbb{R}^{n}_{++}$ satisfy $\boldsymbol{A}\boldsymbol{x} = \boldsymbol{b}$. Denote $\boldsymbol{X} := \mathrm{diag}(\boldsymbol{x})$. Suppose there exists another positive diagonal matrix $\overline{\boldsymbol{X}} \in \mathbb{R}_{++}^{n \times n}$ that approximates $\boldsymbol{X}$ such that, for some $0<\delta<\frac{1}{2} $, $e^{-\delta}\boldsymbol{X} \leq \overline{\boldsymbol{X}} \leq e^{\delta}\boldsymbol{X}$. 
Define the orthogonal projector onto $\ker(\boldsymbol{A X})$ as
\begin{equation}
    \boldsymbol{P} = \boldsymbol{I} - \boldsymbol{X A}^\top (\boldsymbol{A X}^2 \boldsymbol{A}^\top)^{-1} \boldsymbol{A X},
\end{equation}
and define
\begin{equation}
    \boldsymbol{R} = \boldsymbol{I} - \boldsymbol{X}^{-1} \overline{\boldsymbol{X}}^2 \boldsymbol{A}^\top 
(\boldsymbol{A} \overline{\boldsymbol{X}}^2 \boldsymbol{A}^\top)^{-1} \boldsymbol{A X}.
\end{equation}
Then $\boldsymbol{R}$ is a projection matrix satisfying $\boldsymbol{A X R h} = \boldsymbol{0}$ for any $\boldsymbol{h} \in \mathbb{R}^{n}$, and $\Vert \boldsymbol{R} - \boldsymbol{P} \Vert \le 46 \delta $.
\end{lemma}
\begin{proof}
For clarity, let us introduce the shorthand notation $\boldsymbol{M} := \boldsymbol{X}^{-1}\overline{\boldsymbol{X}}^{2}\boldsymbol{X}^{-1}$, $\boldsymbol{B} := \boldsymbol{A}\boldsymbol{X}$. Under this notation, we can rewrite $\boldsymbol{R} = \boldsymbol{I} - \boldsymbol{M}\boldsymbol{B}^{\top}(\boldsymbol{B}\boldsymbol{M}\boldsymbol{B}^{\top})^{-1}\boldsymbol{B}$, and $\boldsymbol{P}= \boldsymbol{I} - \boldsymbol{B}^{\top}(\boldsymbol{B}\boldsymbol{B}^{\top})^{-1}\boldsymbol{B}$.

First, we show that $\boldsymbol{R}$ is a projection. It suffices to verify that $\boldsymbol{R}$ is idempotent, i.e., $\boldsymbol{R}^2 = \boldsymbol{R}$. By direct calculation,
\begin{equation}
\begin{aligned}
\boldsymbol{R}^{2}
&= \big(\boldsymbol{I}-\boldsymbol{M}\boldsymbol{B}^{\top}
(\boldsymbol{B}\boldsymbol{M}\boldsymbol{B}^{\top})^{-1}\boldsymbol{B}\big)^{2} \\
&= \boldsymbol{I}
-2\boldsymbol{M}\boldsymbol{B}^{\top}(\boldsymbol{B}\boldsymbol{M}\boldsymbol{B}^{\top})^{-1}\boldsymbol{B}
+\boldsymbol{M}\boldsymbol{B}^{\top}(\boldsymbol{B}\boldsymbol{M}\boldsymbol{B}^{\top})^{-1}
\boldsymbol{B}\,\boldsymbol{M}\boldsymbol{B}^{\top}(\boldsymbol{B}\boldsymbol{M}\boldsymbol{B}^{\top})^{-1}\boldsymbol{B} \\
&=\boldsymbol{I}
-\boldsymbol{M}\boldsymbol{B}^{\top}(\boldsymbol{B}\boldsymbol{M}\boldsymbol{B}^{\top})^{-1}\boldsymbol{B} \\
&=\boldsymbol{R}.
\end{aligned}
\end{equation}
Thus, $\boldsymbol{R}$ is indeed a projection matrix.

Next, we show that the range of $\boldsymbol{R}$ is contained in $\ker(\boldsymbol{B})$.
To see this, note that 
\begin{equation}
\boldsymbol{B}\boldsymbol{R}
= \boldsymbol{B}
- \boldsymbol{B}\boldsymbol{M}\boldsymbol{B}^{\top}(\boldsymbol{B}\boldsymbol{M}\boldsymbol{B}^{\top})^{-1}\boldsymbol{B} 
= \boldsymbol{B} - \boldsymbol{B} = \boldsymbol{0}.
\end{equation}
Hence, for any $\boldsymbol{h}$, we have $\boldsymbol{B}\boldsymbol{R}\boldsymbol{h} = \boldsymbol{0}$, and consequently $\boldsymbol{A X R h} = \boldsymbol{0}$.

Having established the basic algebraic properties of $\boldsymbol{R}$, we now proceed to bound the spectral norm difference between $\boldsymbol{R}$ and $\boldsymbol{P}$. By expanding their definitions and applying the triangle inequality along with the sub-multiplicativity of the spectral norm, we obtain
\begin{equation}
\label{eq:|R-P|}
\begin{aligned}
    \Vert \boldsymbol{R} - \boldsymbol{P} \Vert
    &=\Vert(\boldsymbol{I}-\boldsymbol{M})\boldsymbol{B}^{\top}(\boldsymbol{B}\boldsymbol{B}^{\top})^{-1}\boldsymbol{B}  + \boldsymbol{M}\left(\boldsymbol{B}^{\top}(\boldsymbol{B}\boldsymbol{B}^{\top})^{-1}\boldsymbol{B} - \boldsymbol{B}^{\top}(\boldsymbol{B}\boldsymbol{M}\boldsymbol{B}^{\top})^{-1}\boldsymbol{B} \right) \Vert \\
    &\leq \Vert(\boldsymbol{I}-\boldsymbol{M})\boldsymbol{B}^{\top}(\boldsymbol{B}\boldsymbol{B}^{\top})^{-1}\boldsymbol{B}  \Vert +  \Vert \boldsymbol{M}\left(\boldsymbol{B}^{\top}(\boldsymbol{B}\boldsymbol{B}^{\top})^{-1}\boldsymbol{B} - \boldsymbol{B}^{\top}(\boldsymbol{B}\boldsymbol{M}\boldsymbol{B}^{\top})^{-1}\boldsymbol{B} \right) \Vert \\
    &\leq \Vert\boldsymbol{I}-\boldsymbol{M}\Vert \Vert\boldsymbol{B}^{\top}(\boldsymbol{B}\boldsymbol{B}^{\top})^{-1}\boldsymbol{B}  \Vert +  \Vert \boldsymbol{M}\Vert \Vert \boldsymbol{B}^{\top}(\boldsymbol{B}\boldsymbol{B}^{\top})^{-1}\boldsymbol{B}(\boldsymbol{M}-\boldsymbol{I}) \boldsymbol{B}^{\top}(\boldsymbol{B}\boldsymbol{M}\boldsymbol{B}^{\top})^{-1}\boldsymbol{B} \Vert \\
    &\leq \Vert \boldsymbol{I}-\boldsymbol{M} \Vert +  \Vert \boldsymbol{M}\Vert\Vert\boldsymbol{B}^{\top}(\boldsymbol{B}\boldsymbol{B}^{\top})^{-1}\boldsymbol{B}\Vert\Vert \boldsymbol{I}-\boldsymbol{M}\Vert\Vert \boldsymbol{B}^{\top}(\boldsymbol{B}\boldsymbol{M}\boldsymbol{B}^{\top})^{-1}\boldsymbol{B}  \Vert \\
    &=\Vert\boldsymbol{I}-\boldsymbol{M}\Vert +  \Vert \boldsymbol{M}\Vert\Vert \boldsymbol{I}-\boldsymbol{M} \Vert\Vert \boldsymbol{B}^{\top}(\boldsymbol{B}\boldsymbol{M}\boldsymbol{B}^{\top})^{-1}\boldsymbol{B}  \Vert.
\end{aligned}
\end{equation}
where the last equality holds because $\boldsymbol{B}^{\top}(\boldsymbol{B}\boldsymbol{B}^{\top})^{-1}\boldsymbol{B}$ is an orthogonal projection matrix, meaning its spectral norm is exactly 1. 

To evaluate the remaining terms in this bound, we observe that since $\boldsymbol{M}$ is a positive diagonal matrix, its spectral norm equals its maximum eigenvalue, i.e., $\Vert \boldsymbol{M}\Vert = \lambda_{\max}(M)$. It then remains to bound the norm  $\Vert \boldsymbol{B}^{\top}(\boldsymbol{B}\boldsymbol{M}\boldsymbol{B}^{\top})^{-1}\boldsymbol{B} \Vert$. Noting that $ \boldsymbol{M} \succeq \lambda_{\min}(\boldsymbol{M})\boldsymbol{I}$, it follows that
\begin{equation}
    \boldsymbol{B}\boldsymbol{M}\boldsymbol{B}^\top \succeq \lambda_{\min}(\boldsymbol{M}) \boldsymbol{B}\boldsymbol{B}^\top. 
\end{equation}
Taking the inverse reverses the Loewner order:
\begin{equation}
(\boldsymbol{B}\boldsymbol{M}\boldsymbol{B}^\top)^{-1} \preceq \frac{1}{\lambda_{\min}(\boldsymbol{M})} (\boldsymbol{B}\boldsymbol{B}^\top)^{-1}.
\end{equation}
Consequently,
\begin{equation}
    \boldsymbol{B}^\top(\boldsymbol{B}\boldsymbol{M}\boldsymbol{B}^\top)^{-1}\boldsymbol{B} \preceq \frac{1}{\lambda_{\min}(\boldsymbol{M})} \boldsymbol{B}^\top(\boldsymbol{B}\boldsymbol{B}^\top)^{-1}\boldsymbol{B} \preceq \frac{1}{\lambda_{\min}(\boldsymbol{M})} \boldsymbol{I},
\end{equation}
where we again used the fact that $\boldsymbol{B}^\top(\boldsymbol{B}\boldsymbol{B}^\top)^{-1}\boldsymbol{B}$ is an orthogonal projector whose eigenvalues are at most 1. This implies
\begin{equation}
\label{eq:BMB}
    \Vert \boldsymbol{B}^{\top}(\boldsymbol{B}\boldsymbol{M}\boldsymbol{B}^{\top})^{-1}\boldsymbol{B} \Vert \leq \frac{1}{\lambda_{\min}(\boldsymbol{M})}.
\end{equation}
Substituting this bound back into our norm inequality \eq{|R-P|} yields 
\begin{equation}
\label{eq:|R-P|-kappa}
     \Vert \boldsymbol{R} - \boldsymbol{P} \Vert \leq 
\Vert\boldsymbol{I}-\boldsymbol{M}\Vert (1+\kappa(\boldsymbol{M})),
\end{equation}
where $\kappa(\boldsymbol{M}) := {\lambda_{\max}(\boldsymbol{M})}/{\lambda_{\min}(\boldsymbol{M})}$ denotes the condition number of $\boldsymbol{M}$.

Finally, we express this bound in terms of the approximation error $\delta$. By the assumption that $e^{-\delta}\boldsymbol{X} \leq \overline{\boldsymbol{X}} \leq e^{\delta}\boldsymbol{X}$, we equivalently have $e^{-2\delta}\boldsymbol{I} \leq \boldsymbol{M} \leq e^{2\delta}\boldsymbol{I}$. This relation immediately implies that  $\kappa(\boldsymbol{M}) \leq e^{4\delta}$, and the spectral norm $\Vert\boldsymbol{I}-\boldsymbol{M}\Vert $ is bounded by $\max\{e^{2\delta}-1,1-e^{-2\delta}\}=e^{2\delta}-1$. Using the elementary inequality $e^x - 1 \leq xe^x$ for $x>0$, we can further bound $\Vert\boldsymbol{I}-\boldsymbol{M}\Vert \leq 2\delta e^\delta$. Substituting these bounds into \eq{|R-P|-kappa} gives
\begin{equation}
     \Vert \boldsymbol{R} - \boldsymbol{P} \Vert \leq 
\Vert\boldsymbol{I}-\boldsymbol{M}\Vert (1+\kappa(\boldsymbol{M})) \leq 2\delta e^{2\delta}(1+e^{4\delta}).
\end{equation}
Plugging in $\delta <\frac{1}{2}$ bounds the constant factor by $2 e(1+e^{2})\leq 46$, yielding $\Vert \boldsymbol{R} - \boldsymbol{P} \Vert \leq 46\delta$ and completing the proof.
\end{proof}

Having established the error bound of $\Vert \boldsymbol{R} - \boldsymbol{P} \Vert$ in \lem{R-P}, we now analyze the decrease of the potential function. To guarantee that the approximate direction $\widetilde{\boldsymbol{d}}$ still yields a sufficient decrease when $\delta$ is small, the scaled gradient of the potential function must not grow unboundedly. Therefore, we introduce the following assumption.
\begin{assumption}
\label{assum:xphi}
There exists a constant $L_\phi>0$ such that
$\Vert\boldsymbol{X}\nabla\phi(\boldsymbol{x})\Vert\le L_{\phi}$ for all $\boldsymbol{x}\in\Omega^\circ$.
\end{assumption}
\begin{remark}
This assumption is mild in practice.
A simple sufficient condition is when  the iterates $\boldsymbol{x}_t$ remain bounded, i.e., there exists $R\ge 1$ such that $\sup\{\Vert \boldsymbol{x}_t \Vert_\infty : f(\boldsymbol{x}) \leq f(\boldsymbol{x}_0) , \boldsymbol{x} \in \Omega^\circ, t\in [T] \} \leq R$. If, in addition, the function $f$ is $L$-Lipschitz continuous on $\Omega^\circ$, then its gradient is bounded as $\Vert\nabla f(\boldsymbol{x})\Vert \le L$. Hence $\Vert\boldsymbol{X}\nabla\phi(\boldsymbol{x})\Vert = \Vert\boldsymbol{X}\nabla f(\boldsymbol{x}) + \varepsilon \boldsymbol{e}\Vert  \leq  \Vert\boldsymbol{X}\nabla f(\boldsymbol{x}) \Vert + \Vert\varepsilon \boldsymbol{e}\Vert \leq RL + \varepsilon\sqrt{n}$. If we further set $\varepsilon \le 1/\sqrt{n}$, then $\Vert \boldsymbol{X}\nabla\phi(\boldsymbol{x}) \Vert$ is bounded by a constant $K := RL + 1$.
\end{remark}

We now analyze the decrease of the potential function at each iteration.
\begin{proposition}
\label{prop:potential-decre}
    Under \assum{grad-lip} and \assum{xphi}, for any $\varepsilon \in (0, \min\{\gamma, \frac{1}{2}\}]$, let $\beta = \varepsilon / (l + 2\varepsilon + 2)$ and $\delta = \min(\varepsilon/(15L_\phi), \beta / (92 L_{\phi}))$. 
    Suppose that $e^{-\delta}\boldsymbol{x}_t \leq \overline{\boldsymbol{x}}_t \leq e^{\delta}\boldsymbol{x}_t$, and define the next iterate by $\boldsymbol{x}_{t+1} := \boldsymbol{X}_t(\boldsymbol{e} + \widetilde{\boldsymbol{d}}_{t})$, where $\widetilde{\boldsymbol{d}}_{t}$ is given in \eq{d-tilde}. 
    Then, at iteration $t$, one of the following holds:
    \begin{equation}
        \phi(\boldsymbol{x}_{t+1}) - \phi(\boldsymbol{x}_{t}) \leq - \frac{\varepsilon^2}{2l + 4\varepsilon + 4},
    \end{equation}
    or there exists $\boldsymbol{v}_{t} \in \mathbb{R}^m$ such that
\begin{equation}
    \Vert \boldsymbol{X}_t (\nabla f(\boldsymbol{x}_t)
        + \boldsymbol{A}^\top \boldsymbol{v}_{t}) \Vert_\infty < 2\varepsilon,
        \quad \text{and} \quad
        \nabla f(\boldsymbol{x}_t) + \boldsymbol{A}^\top \boldsymbol{v}_{t} > 0.
\end{equation}
\end{proposition}

\begin{proof}
Let $\boldsymbol{d}_t$ denote the exact optimal solution of \eq{1st-IPTR}. 
By the necessary and sufficient optimality conditions, there exist $\lambda_t \ge 0$ and $\boldsymbol{v}_t \in \mathbb{R}^m$ such that
\begin{equation}
\label{eq:ness-suff}
    \boldsymbol{X}_t \nabla f(\boldsymbol{x}_t) - \varepsilon \boldsymbol{e} 
    + \boldsymbol{X}_t\boldsymbol{A}^\top \boldsymbol{v}_{t} 
    + \lambda_t\boldsymbol{d}_t = 0.
\end{equation}

We now examine the decrease of the potential function when the update 
$\boldsymbol{x}_{t+1} := \boldsymbol{X}_t(\boldsymbol{e} + \widetilde{\boldsymbol{d}}_{t})$ is applied:
\begin{equation}
\begin{aligned}
\label{eq:phi-decre-ineq}
    \phi(\boldsymbol{x}_{t+1}) - \phi(\boldsymbol{x}_{t})
    &\leq \langle \boldsymbol{X}_t \nabla \phi(\boldsymbol{x}_t), \widetilde{\boldsymbol{d}}_{t}\rangle 
      + \frac{l}{2} \Vert \widetilde{\boldsymbol{d}}_{t} \Vert^2 
      + \varepsilon \beta^2 \\
    &= \langle -\boldsymbol{X}_t\boldsymbol{A}^\top \boldsymbol{v}_{t} 
      - \lambda_t \boldsymbol{d}_t, \widetilde{\boldsymbol{d}}_{t}\rangle
      + \left(\frac{l}{2}+\varepsilon \right)\beta^2  \\
    &= -\lambda_t \Vert\boldsymbol{d}_{t}\Vert^2 
      + \lambda_t \boldsymbol{d}_{t}^\top (\boldsymbol{d}_{t} - \widetilde{\boldsymbol{d}}_{t}) 
      + \left(\frac{l}{2}+\varepsilon \right)\beta^2.
\end{aligned}
\end{equation}
The second equality follows from \eq{ness-suff}, and the third is due to $\boldsymbol{A}\boldsymbol{X}_t\widetilde{\boldsymbol{d}}_t = \boldsymbol{0}$ . This expresses the potential decrease in terms of the exact direction $\boldsymbol{d}_{t}$ and its approximation $\widetilde{\boldsymbol{d}}_{t}$. 
We evaluate this bound under two cases: (i) $\lambda_t = 0$ and $\Vert \boldsymbol{d}_t \Vert < \beta$, or (ii) $\lambda_t > 0$ and $\Vert \boldsymbol{d}_t \Vert = \beta$.

\medskip
\noindent\textbf{Case 1.} $\lambda_t = 0$ and $\Vert \boldsymbol{d}_t \Vert < \beta$.  
In this case, \eq{ness-suff} yields  $\boldsymbol{X}_t (\nabla f(\boldsymbol{x}_t) + \boldsymbol{A}^\top \boldsymbol{v}_{t}) = \varepsilon \boldsymbol{e}$. It then follows that $\Vert \boldsymbol{X}_t (\nabla f(\boldsymbol{x}_t) + \boldsymbol{A}^\top \boldsymbol{v}_{t}) \Vert_\infty =\varepsilon < 2\varepsilon$. Furthermore, since the diagonal matrix is strictly positive, this equality directly implies $\nabla f(\boldsymbol{x}_t)+\boldsymbol{A}^\top \boldsymbol{v}_{t} > 0$.

\medskip
\noindent\textbf{Case 2.} Under the conditions $\lambda_t > 0$ and $\Vert \boldsymbol{d}_t \Vert = \beta$, we distinguish between two subcases according to the norm of the projected gradient $\Vert \widetilde{\boldsymbol{d}}_{t}\Vert$.

\smallskip
\noindent\textbf{Case 2.1.} If $\boldsymbol{R}_t \boldsymbol{X}_t \nabla \phi(\boldsymbol{x}_t) = 0$, then, by the definition of $\boldsymbol{R}_t$,
\[
    \boldsymbol{X}_t \nabla \phi(\boldsymbol{x}_t)
    = \left(\boldsymbol{X}_{t}^{-1} \overline{\boldsymbol{X}}^2_{t} \boldsymbol{A}^\top 
      (\boldsymbol{A} \overline{\boldsymbol{X}}^2_{t} \boldsymbol{A}^\top)^{-1} 
      \boldsymbol{A}\boldsymbol{X}_{t} \right)\boldsymbol{X}_t \nabla \phi(\boldsymbol{x}_t).
\]
Let $\boldsymbol{M}_t := \boldsymbol{X}_{t}^{-1} \overline{\boldsymbol{X}}^2_{t}\boldsymbol{X}_{t}^{-1}$ and 
$\boldsymbol{v}_t := (\boldsymbol{A} \overline{\boldsymbol{X}}^2_{t} \boldsymbol{A}^\top)^{-1} 
\boldsymbol{A}\boldsymbol{X}_{t}^2 \nabla \phi(\boldsymbol{x}_t)$. Then
\[
    \boldsymbol{X}_t \nabla \phi(\boldsymbol{x}_t)
    = \boldsymbol{M}_t \boldsymbol{X}_t \boldsymbol{A}^\top\boldsymbol{v}_t,
\]
and hence
\begin{equation}
\begin{aligned}
\label{eq:case2_1}
    \Vert \boldsymbol{X}_t \nabla f(\boldsymbol{x}_t) - \varepsilon \boldsymbol{e}+\boldsymbol{X}_t\boldsymbol{A}^\top \boldsymbol{v}_{t} \Vert_\infty 
    &= \Vert (\boldsymbol{I}-\boldsymbol{M}_t)\boldsymbol{X}_t \boldsymbol{A}^\top\boldsymbol{v}_t \Vert_\infty \\
    & \leq \Vert (\boldsymbol{I}-\boldsymbol{M}_t)\boldsymbol{X}_t \boldsymbol{A}^\top\boldsymbol{v}_t \Vert \\
    &\leq  \Vert \boldsymbol{I}-\boldsymbol{M}_t\Vert \Vert \boldsymbol{X}_t \boldsymbol{A}^\top  (\boldsymbol{A} \overline{\boldsymbol{X}}^2_{t} \boldsymbol{A}^\top)^{-1} \boldsymbol{A}\boldsymbol{X}_{t} \Vert \Vert \boldsymbol{X}_t \nabla\phi(\boldsymbol{x}_t) \Vert \\
    & \leq 2\delta e^{2\delta} \cdot e^{2\delta} \cdot L_\phi \\
    & \leq \varepsilon.
\end{aligned}
\end{equation}
The fourth line follows from \eq{BMB} in \lem{R-P}, and the last line follows from $\delta\leq 1/2$ and $\delta \leq \varepsilon/(15L_\phi)$.  Thus, there exists $\boldsymbol{v}_t$ such that $\Vert \boldsymbol{X}_t \nabla f(\boldsymbol{x}_t) - \varepsilon \boldsymbol{e} + \boldsymbol{X}_t\boldsymbol{A}^\top \boldsymbol{v}_{t} \Vert_\infty \leq \varepsilon$. This implies that $0\leq \boldsymbol{X}_t (\nabla f(\boldsymbol{x}_t) + \boldsymbol{A}^\top \boldsymbol{v}_{t})  \leq 2\varepsilon$ holds elementwise, and $\nabla f(\boldsymbol{x}_t)+\boldsymbol{A}^\top \boldsymbol{v}_{t}\geq 0$ since $\boldsymbol{X}_t$ is nonnegative.

\smallskip
\noindent\textbf{Case 2.2.} When $\boldsymbol{R}_t \boldsymbol{X}_t \nabla \phi(\boldsymbol{x}_t) \neq 0$, to bound the decrease in the potential function, it remains to control the error term $ \lambda_t \boldsymbol{d}_{t}^\top (\boldsymbol{d}_{t} - \widetilde{\boldsymbol{d}}_{t})$ in \eq{phi-decre-ineq}. By the definition of $\widetilde{\boldsymbol{d}}_t$, we have $\Vert \widetilde{\boldsymbol{d}}_{t}\Vert=\beta$. Therefore,
\begin{equation}
\begin{aligned}
    \Vert \lambda_t \boldsymbol{d}_{t}^\top (\boldsymbol{d}_{t} - \widetilde{\boldsymbol{d}}_{t}) \Vert
    &\leq \lambda_t \beta^2 
    \left\Vert\frac{\boldsymbol{P}_t \boldsymbol{X}_t \nabla \phi(\boldsymbol{x}_t)}{\Vert \boldsymbol{P}_t \boldsymbol{X}_t \nabla \phi(\boldsymbol{x}_t) \Vert}
    -\frac{\boldsymbol{R}_t \boldsymbol{X}_t \nabla \phi(\boldsymbol{x}_t)}{\Vert \boldsymbol{R}_t \boldsymbol{X}_t \nabla \phi(\boldsymbol{x}_t) \Vert} \right\Vert \\
    &\leq 2\beta \Vert \boldsymbol{P}_t - \boldsymbol{R}_t\Vert \Vert  \boldsymbol{X}_t \nabla \phi(\boldsymbol{x}_t) \Vert \\
    &\leq 2\beta \cdot 46 \delta \cdot L_\phi \leq \beta^2.
\end{aligned}
\end{equation}
The second line follows from the unit-vector difference bound 
$\Vert\boldsymbol{a}/\Vert\boldsymbol{a}\Vert - \boldsymbol{b}/\Vert\boldsymbol{b}\Vert\Vert 
\leq 2 \Vert\boldsymbol{a}-\boldsymbol{b}\Vert/\Vert\boldsymbol{a}\Vert$ and the identity $\Vert \boldsymbol{P}_t \boldsymbol{X}_t \nabla \phi(\boldsymbol{x}_t) \Vert = \Vert \boldsymbol{P}_t\boldsymbol{X}_t\boldsymbol{A}^\top \boldsymbol{v}_{t} 
    + \lambda_t\boldsymbol{P}_t\boldsymbol{d}_t \Vert= \lambda_t \Vert \boldsymbol{d}_t \Vert$, 
while the last line follows from \lem{R-P} and $\delta \leq  \beta / (92 L_{\phi})$.

To simplify the notation in the subsequent analysis, define $p(\boldsymbol{x}_t,\boldsymbol{v}_t):= \boldsymbol{X}_t \nabla f(\boldsymbol{x}_t) - \varepsilon \boldsymbol{e} + \boldsymbol{X}_t\boldsymbol{A}^\top \boldsymbol{v}_t$. From \eq{ness-suff}, $\Vert p(\boldsymbol{x}_t,\boldsymbol{v}_t) \Vert  = \lambda_t \beta$. Then the potential decrease satisfies
\begin{align}
    \phi(\boldsymbol{x}_{t+1}) - \phi(\boldsymbol{x}_{t})
    &\leq -\lambda_t \beta^2 + \left(\frac{l}{2}+\varepsilon +1 \right)\beta^2= -\Vert p(\boldsymbol{x}_t,\boldsymbol{v}_t) \Vert \beta  
       + \left(\frac{l}{2}+\varepsilon +1 \right)\beta^2.
\end{align}

\noindent\textbf{Case 2.2.1.} If $\Vert p(\boldsymbol{x}_t,\boldsymbol{v}_t) \Vert \geq \varepsilon$, then
\[
    \phi(\boldsymbol{x}_{t+1}) - \phi(\boldsymbol{x}_{t}) 
    \leq -\varepsilon\beta + \left(\frac{l}{2}+\varepsilon +1 \right)\beta^2 
    \leq - \frac{\varepsilon^2}{2l+4\varepsilon +4}.
\]
\textbf{Case 2.2.2.} If $\Vert p(\boldsymbol{x}_t,\boldsymbol{v}_t) \Vert < \varepsilon$,
then by the definition of $p(\boldsymbol{x}_t,\boldsymbol{v}_t)$, we have $\Vert \boldsymbol{X}_t \nabla f(\boldsymbol{x}_t) - \varepsilon \boldsymbol{e}
+ \boldsymbol{X}_t\boldsymbol{A}^\top \boldsymbol{v}_{t} \Vert_\infty < \varepsilon$. This again implies that $0\leq \boldsymbol{X}_t (\nabla f(\boldsymbol{x}_t) + \boldsymbol{A}^\top \boldsymbol{v}_{t})  \leq 2\varepsilon$ holds elementwise, and $\nabla f(\boldsymbol{x}_t)+\boldsymbol{A}^\top \boldsymbol{v}_{t}\geq 0$ since $\boldsymbol{X}_t$ is nonnegative.
\end{proof}

\prop{potential-decre} naturally provides a stopping criterion for \algo{robust-1st-order}. At each iteration, the algorithm checks the decrease in the potential function. If the sufficient decrease of $-{\varepsilon^2}/{(2l+4\varepsilon +4)}$ is met, the algorithm proceeds to the next step; otherwise, it terminates. In the latter case, the proposition guarantees that the current iterate $\boldsymbol{x}_t$ is already a $2\varepsilon$-KKT point. While this mechanism bounds the maximum number of iterations, \prop{robust-time} ensures that each individual update can be computed efficiently. By combining the iteration bound with the per-iteration computational cost, we can now establish the overall convergence and time complexity of the algorithm.
\begin{theorem}
\label{thm:robust-complexity}
Suppose that \assum{grad-lip} and \assum{xphi} hold. For any $\varepsilon \in (0, \min\{\gamma, \tfrac{1}{2}\}]$, the \algo{robust-1st-order} is guaranteed to find a $2\varepsilon$-KKT point within $\mathcal{O}( {l (f(\boldsymbol{x}_0) - f(\boldsymbol{x}^*))}/{\varepsilon^2})$ gradient evaluations. 
Otherwise, it holds that $ f(\boldsymbol{x}_t) - f(\boldsymbol{x}^*) \leq \varepsilon$. Moreover, the overall time complexity of the algorithm is $\widetilde{\mathcal{O}}(nm^{\omega-1} + {nm}/{\varepsilon^2} )$.
\end{theorem}
\begin{proof}
    By \prop{potential-decre}, under the choice of parameters $\beta = \varepsilon / (l + 2\varepsilon + 2)$ and $\delta = \min\{\varepsilon/(15L_\phi), \, \beta/(92L_\phi)\}$, each iteration of \algo{robust-1st-order} either ensures a sufficient decrease in the potential function by at least ${\varepsilon^2}/{(2l + 4\varepsilon + 4)}$ or indicates that the current iterate $\boldsymbol{x}_t$ is already a $2\varepsilon$-KKT point. 

    Suppose first that every iteration falls into the decrease case. Since the initial point $\boldsymbol{x}_0 > \boldsymbol{0}$ is assumed to satisfy $ -\sum_{i=1}^n \ln(x_i) \ge -\sum_{i=1}^n \ln(x_{0,i}) - C_0 $ for all $\boldsymbol{x} \in \Omega^\circ$ as stated in \eq{initial}, after $t$ iterations we obtain
    \begin{equation}
    f(\boldsymbol{x}_t)-f(\boldsymbol{x}_0) = \phi(\boldsymbol{x}_t) + \varepsilon\sum_{i=1}^n \ln(x_{t,i})-\phi(\boldsymbol{x}_0) + \varepsilon\sum_{i=1}^n \ln(x_{0,i}) \leq -\frac{t\varepsilon^2}{2l + 4\varepsilon + 4} + \varepsilon C_0.
    \end{equation} 
    Consequently, the total number of iterations is bounded by
    $T = \frac{\left(f(\boldsymbol{x}_0)- f(\boldsymbol{x}^*) + (C_0-1)\varepsilon\right)(2l + 4\varepsilon + 4)}{\varepsilon^2}$, where $f(\boldsymbol{x}^*)$ denotes the global minimum value. If the algorithm reaches this maximum number of iterations, it must hold that $ f(\boldsymbol{x}_t) - f(\boldsymbol{x}^*) \leq \varepsilon$.

    Finally, by \prop{robust-time}, the overall time complexity under the above iteration bound is $\widetilde{\mathcal{O}}(nm^{\omega-1} + {nm}/{\varepsilon^2} )$, which completes the proof.
\end{proof}

\begin{remark}
    In practice, the trust-region radius $\beta$ in \algo{robust-1st-order} can be adjusted adaptively as in \cite{jiang2026beyond}. Indeed, the proof of \prop{potential-decre} shows that the decrease of the potential function satisfies
    \begin{align*}
    \phi(\boldsymbol{x}_{t+1}) - \phi(\boldsymbol{x}_{t})
    \leq -\Vert p(\boldsymbol{x}_t,\boldsymbol{v}_t) \Vert \beta  
       + \left(\frac{l}{2}+\varepsilon +1 \right)\beta^2.
    \end{align*}
    Here,
    $
    \Vert p(\boldsymbol{x}_t,\boldsymbol{v}_t) \Vert
    = \Vert \boldsymbol{P}_t \boldsymbol{X}_t \nabla\phi(\boldsymbol{x}_t)\Vert,
    $
    which is close to $\Vert \boldsymbol{R}_t \boldsymbol{X}_t \nabla\phi(\boldsymbol{x}_t)\Vert$ under the approximate projection. This suggests that, in practical implementations, one may choose $\beta$ on the order of
    $
    \beta = \mathcal{O}({\Vert \boldsymbol{R}_t \boldsymbol{X}_t \nabla\phi(\boldsymbol{x}_t)\Vert}/{l})
    $. Such a choice allows a larger step and potentially faster decrease when the projected gradient is large. For a small projected gradient, $\beta$ is correspondingly small, which reduces the potential decrease. Therefore, the current stopping criterion still applies.
\end{remark}

\subsection{Improved complexity with a concave objective function}
\label{sec:concave}

We now consider the special case where $f$ is concave on $\Omega^\circ$, i.e.,
\begin{align*}
    f(\boldsymbol{y}) \leq f(\boldsymbol{x})
    + \nabla f(\boldsymbol{x})^\top(\boldsymbol{y}-\boldsymbol{x}),
    \quad \forall \boldsymbol{x},\boldsymbol{y}\in\Omega^\circ.
\end{align*}
In this setting, $f$ admits a first-order upper bound along every feasible scaled direction, so the quadratic term in \eq{2nd-oder-upper} is no longer needed. Consequently, the decrease of the potential function becomes linear in $\varepsilon$, improving the iteration complexity of first-order IPTR algorithms from $\mathcal{O}(1/\varepsilon^2)$ to $\mathcal{O}(1/\varepsilon)$. We first analyze the case with the exact projection $\boldsymbol{P}_t$, and then turn to \algo{robust-1st-order} with the approximate projection $\boldsymbol{R}_t$.

In the exact-projection case, at each iteration we solve \eq{1st-IPTR} with trust-region radius $\beta$, compute the search direction according to \eq{exact-d}, and update $\boldsymbol{x}_{t+1}=\boldsymbol{x}_t+\boldsymbol{X}_t\boldsymbol{d}_t$. The stopping rule follows the same logic as before: we check the decrease of the potential function, and if $\phi(\boldsymbol{x}_{t+1})-\phi(\boldsymbol{x}_t)>-\varepsilon\beta(1-\beta)$, then one can show that $\boldsymbol{x}_t$ is already a $2\varepsilon$-KKT point. We refer to this procedure as the exact first-order IPTR method. The resulting iteration and time complexity bounds are stated in the following theorem.

\begin{theorem}
\label{thm:concave-exact}
Suppose that $f$ is concave on $\Omega^\circ$. Consider the exact first-order IPTR algorithm described above, with trust-region radius $\beta \in (0,1)$. Then, for any $\varepsilon \in (0,1]$, the algorithm either returns a $2\varepsilon$-KKT point within $\mathcal{O}(1/\varepsilon)$ iterations, or returns an iterate $\boldsymbol{x}_t$ such that $f(\boldsymbol{x}_t)-f(\boldsymbol{x}^*)\le \varepsilon$. The overall time complexity is upper bounded by $\mathcal{O}({nm^{\omega-1}}/{\varepsilon})$.
\end{theorem}

\begin{proof}
Let $\boldsymbol{d}_t$ be defined by \eq{exact-d}, and set $\boldsymbol{x}_{t+1}=\boldsymbol{X}_t(\boldsymbol{e}+\boldsymbol{d}_t)$. By the necessary and sufficient optimality conditions for \eq{1st-IPTR}, there exist $\lambda_t \ge 0$ and $\boldsymbol{v}_t \in \mathbb{R}^m$ such that \eq{ness-suff} holds. Repeating the derivation of \eq{phi-decre-ineq} with $\widetilde{\boldsymbol{d}}_t=\boldsymbol{d}_t$, and using the concavity of $f$ in place of \eq{2nd-oder-upper}, we obtain
\begin{equation}
    \phi(\boldsymbol{x}_{t+1})-\phi(\boldsymbol{x}_t) \le -\lambda_t\|\boldsymbol{d}_t\|^2+\varepsilon\beta^2.
\end{equation}

We distinguish two cases. If $\lambda_t=0$ and $\|\boldsymbol{d}_t\|<\beta$, then as in Case 1 of \prop{potential-decre}, there exists $\boldsymbol{v}_t\in\mathbb{R}^m$ such that $\boldsymbol{X}_t (\nabla f(\boldsymbol{x}_t) + \boldsymbol{A}^\top \boldsymbol{v}_{t}) = \varepsilon \boldsymbol{e}$.
Thus $\boldsymbol{x}_t$ is a $2\varepsilon$-KKT point. Otherwise, we have $\lambda_t>0$ and $\|\boldsymbol{d}_t\|=\beta$. Define $p(\boldsymbol{x}_t,\boldsymbol{v}_t):= \boldsymbol{X}_t\nabla f(\boldsymbol{x}_t)-\varepsilon\boldsymbol{e} +\boldsymbol{X}_t\boldsymbol{A}^\top \boldsymbol{v}_t$. By \eq{ness-suff}, we have $\|p(\boldsymbol{x}_t,\boldsymbol{v}_t)\|=\lambda_t \Vert \boldsymbol{d}_t \Vert$. If $\|p(\boldsymbol{x}_t,\boldsymbol{v}_t)\|<\varepsilon$, then as in Case 2.2.2 of \prop{potential-decre}, the same $2\varepsilon$-KKT conclusion follows. Otherwise,
\begin{equation}
    \phi(\boldsymbol{x}_{t+1})-\phi(\boldsymbol{x}_t)\le-\lambda_t \Vert \boldsymbol{d}_t \Vert^2 +\varepsilon\beta^2\le-\varepsilon\beta+\varepsilon\beta^2=-\varepsilon\beta(1-\beta).
\end{equation}

Therefore, at each iteration, either $\boldsymbol{x}_t$ is a $2\varepsilon$-KKT point, or the potential function decreases by at least $\varepsilon\beta(1-\beta)$. Combining this with the initialization in \eq{initial}, after $t$ iterations we have
\begin{equation}
    f(\boldsymbol{x}_t)-f(\boldsymbol{x}_0) = \phi(\boldsymbol{x}_t) + \varepsilon\sum_{i=1}^n \ln(x_{t,i})-\phi(\boldsymbol{x}_0) + \varepsilon\sum_{i=1}^n \ln(x_{0,i}) \leq -t\varepsilon\beta(1-\beta) + \varepsilon C_0.
\end{equation}
It follows that the total number of iterations is bounded by $T = \frac{f(\boldsymbol{x}_0)- f(\boldsymbol{x}^*) + (C_0-1)\varepsilon}{\varepsilon \beta (1-\beta)}$, where $f(\boldsymbol{x}^*)$ denotes the global minimum value. If the algorithm reaches this bound without returning a $2\varepsilon$-KKT point, then it must hold that $f(\boldsymbol{x}_t) - f(\boldsymbol{x}^*) \leq \varepsilon$. The time complexity of each iteration is $\mathcal{O}(nm^{\omega-1})$ as discussed after \eq{exact-d}. Multiplying by the iteration bound yields the overall time complexity $\mathcal{O}({nm^{\omega-1}}/{\varepsilon})$.
\end{proof}

We next analyze the time and iteration complexity of \algo{robust-1st-order} for concave functions. The descent analysis closely follows the previous arguments and still gives the $\mathcal{O}(1/\varepsilon)$ iteration complexity. The approximate update scheme differs in one respect. In the concave setting, the trust-region radius satisfies $\beta=\Theta(1)$. To maintain the accuracy of the approximate projection $\boldsymbol{R}_t$, \lem{R-P} gives $\|\boldsymbol{R}_t-\boldsymbol{P}_t\|\le 46\delta$ under $e^{-\delta}\boldsymbol{X}_t \leq \overline{\boldsymbol{X}_t} \leq e^{\delta}\boldsymbol{X}_t$. Since the projection error now needs to be controlled at the $\Theta(\varepsilon)$ level, we take $\delta=\Theta(\varepsilon)$. With a constant stepsize but an update tolerance of order $\Theta(\varepsilon)$, the approximate update scheme needs to update more coordinates as shown in \lem{SelectVector}. We first establish the corresponding one-step guarantee for the \algo{robust-1st-order}.

\begin{proposition}
\label{prop:concave-approx}
Suppose that $f$ is concave on $\Omega^\circ$ and \assum{xphi} holds. For any $\varepsilon \in (0,1]$, $\beta \in (0,1)$, choose $\delta \le {(1-\beta)\varepsilon}/{184L_\phi}$. At iteration $t$, suppose that $e^{-\delta}\boldsymbol{x}_t \le \overline{\boldsymbol{x}}_t \le e^{\delta}\boldsymbol{x}_t$, define $\widetilde{\boldsymbol{d}}_t$ by \eq{d-tilde}, and let $\boldsymbol{x}_{t+1}=\boldsymbol{X}_t(\boldsymbol{e}+\widetilde{\boldsymbol{d}}_t)$. Then exactly one of the following two statements holds:
\begin{equation}
    \phi(\boldsymbol{x}_{t+1})-\phi(\boldsymbol{x}_t)\le -\frac12\,\varepsilon\beta(1-\beta),
\end{equation}
or there exists $\boldsymbol{v}_t\in\mathbb{R}^m$ such that
\begin{equation}
    \|\boldsymbol{X}_t(\nabla f(\boldsymbol{x}_t)+\boldsymbol{A}^\top \boldsymbol{v}_t)\|_\infty \le 2\varepsilon,
    \qquad
    \nabla f(\boldsymbol{x}_t)+\boldsymbol{A}^\top \boldsymbol{v}_t \ge \boldsymbol{0}.
\end{equation}
\end{proposition}
\begin{proof}
Let $\boldsymbol{d}_t$ denote the exact optimal solution of \eq{1st-IPTR}. By the necessary and sufficient optimality conditions for \eq{1st-IPTR}, there exist $\lambda_t \ge 0$ and $\boldsymbol{v}_t \in \mathbb{R}^m$ such that \eq{ness-suff} holds. Repeating the derivation of \eq{phi-decre-ineq}, with the concavity of $f$ in place of \eq{2nd-oder-upper}, we obtain
\begin{equation}
    \phi(\boldsymbol{x}_{t+1})-\phi(\boldsymbol{x}_t)
    \le
    -\lambda_t\|\boldsymbol{d}_t\|^2+\lambda_t\boldsymbol{d}_t^\top(\boldsymbol{d}_t-\widetilde{\boldsymbol{d}}_t)+\varepsilon\beta^2.
\end{equation}

If $\lambda_t=0$ and $\|\boldsymbol{d}_t\|<\beta$, then, as in Case 1 of \prop{potential-decre}, we can find $\boldsymbol{v}_t\in\mathbb{R}^m$ such that $\boldsymbol{X}_t (\nabla f(\boldsymbol{x}_t) + \boldsymbol{A}^\top \boldsymbol{v}_{t}) = \varepsilon \boldsymbol{e}$. Thus $\boldsymbol{x}_t$ is a $2\varepsilon$-KKT point. 

If $\boldsymbol{R}_t\boldsymbol{X}_t\nabla\phi(\boldsymbol{x}_t)=\boldsymbol{0}$, then, by the same argument as in Case 2.1 of \prop{potential-decre}, one can choose $\boldsymbol{v}_t$ such that
\begin{equation*}
    \|\boldsymbol{X}_t\nabla f(\boldsymbol{x}_t)-\varepsilon\boldsymbol{e}+\boldsymbol{X}_t\boldsymbol{A}^\top\boldsymbol{v}_t\|_\infty \le 2\delta e^{4\delta} L_{\phi}.
\end{equation*}
Since $\delta \le {(1-\beta)\varepsilon}/{184L_\phi}$ and $\beta \in (0,1)$, it follows that $2\delta e^{4\delta} L_{\phi} \le {\varepsilon}$. Thus $\boldsymbol{x}_t$ is again a $2\varepsilon$-KKT point.

Now suppose that $\boldsymbol{R}_t\boldsymbol{X}_t\nabla\phi(\boldsymbol{x}_t)\neq \boldsymbol{0}$ and $\lambda_t>0$, $\Vert \boldsymbol{d}_t \Vert = \beta$. By the definition of $\widetilde{\boldsymbol{d}}_t$, we have $\|\widetilde{\boldsymbol{d}}_t\|=\beta$. As in Case 2.2 of \prop{potential-decre},
\begin{equation}
| \lambda_t\boldsymbol{d}_t^\top(\boldsymbol{d}_t-\widetilde{\boldsymbol{d}}_t)|
    \le 2\beta\|\boldsymbol{P}_t-\boldsymbol{R}_t\|\,\|\boldsymbol{X}_t\nabla\phi(\boldsymbol{x}_t)\|
    \le 92\beta\delta L_\phi
    \le \frac12\,\varepsilon\beta(1-\beta).
\end{equation}
Define $p(\boldsymbol{x}_t,\boldsymbol{v}_t):=\boldsymbol{X}_t\nabla f(\boldsymbol{x}_t)-\varepsilon\boldsymbol{e}+\boldsymbol{X}_t\boldsymbol{A}^\top\boldsymbol{v}_t$. By \eq{ness-suff}, we have $\|p(\boldsymbol{x}_t,\boldsymbol{v}_t)\|=\lambda_t\beta$. If $\|p(\boldsymbol{x}_t,\boldsymbol{v}_t)\|<\varepsilon$, then, by the same argument as in Case 2.2.2 of \prop{potential-decre}, the first-order KKT condition holds. Otherwise, $\lambda_t\beta\ge \varepsilon$, and hence
\begin{equation}
\begin{aligned}
    \phi(\boldsymbol{x}_{t+1})-\phi(\boldsymbol{x}_t)
    &\le
    -\lambda_t \Vert \boldsymbol{d}_t \Vert^2+\varepsilon\beta^2+\frac12\,\varepsilon\beta(1-\beta) \\
    &\le
    -\varepsilon\beta+\varepsilon\beta^2+\frac12\,\varepsilon\beta(1-\beta) \\
    &=
    -\frac12\,\varepsilon\beta(1-\beta).
\end{aligned}
\end{equation}
\end{proof}

Based on \prop{concave-approx}, we can now derive the iteration and time complexity bounds of \algo{robust-1st-order} in the concave setting.
\begin{theorem}
\label{thm:concave-approximate}
Suppose $f$ is concave on $\Omega^\circ$ and \assum{xphi} holds. Let $\varepsilon\in(0,1]$. Consider \algo{robust-1st-order} with the trust-region radius $\beta \in (0,1)$ and $\delta= {(1-\beta)\varepsilon}/{(184L_\phi)}$. Then, the algorithm terminates within $\mathcal{O}(1/\varepsilon)$ iterations, returning either a $2\varepsilon$-KKT point or an iterate $\boldsymbol{x}_t$ satisfying $f(\boldsymbol{x}_t)-f(\boldsymbol{x}^*)\le \varepsilon$. Furthermore, the overall time complexity is $\widetilde{\mathcal{O}}(nm^{\omega-1}+{nm}/{\varepsilon^2})$.
\end{theorem}

\begin{proof}
According to \prop{concave-approx}, each iteration either yields a $2\varepsilon$-KKT point or decreases the potential function by at least $\frac12\varepsilon\beta(1-\beta)$. Following the same potential argument used in \thm{concave-exact}, the total number of iterations $T$ is bounded by
\begin{equation}
T=\frac{2(f(\boldsymbol{x}_0)-f(\boldsymbol{x}^*)+(C_0-1)\varepsilon)}{\varepsilon\beta(1-\beta)}.
\end{equation}

To analyze the time complexity, let $q_t$ be the number of coordinates updated between $\ln\overline{\boldsymbol{x}}_{t-1}$ and $\ln\overline{\boldsymbol{x}}_t$. By \lem{SelectVector}(ii), there exists a constant $C \ge 1$ such that
\begin{equation}
\label{eq:q_t-concave}
q_t \le \min\left\{n,\ C4^{l_t}(\beta/\delta)^2\log^2 n\right\},
\qquad t=1,2,\ldots,T.
\end{equation}
Substituting $\delta = (1-\beta)\varepsilon/(184L_\phi)$ yields $(\beta/\delta)^2 = (184\beta L_\phi)^2/((1-\beta)^2\varepsilon^2)$. This allows us to rewrite \eqref{eq:q_t-concave} as
\begin{equation}
    q_t \le \min\left\{n,\ {C'4^{l_t}\log^2 n}/{\varepsilon^2}\right\},
\qquad t=1,2,\ldots,T,
\end{equation}
where $C' = C({184\beta L_\phi}/{(1-\beta)})^2$.

Let $N_l$ denote the number of iterations sharing the same value of $l$, as defined in \eq{N_l}. Computationally, the Woodbury matrix identity is more efficient when $q_t \le m$, whereas direct projection is preferable when $m < q_t \le n$. To formalize this transition, we define a threshold $l^*_\varepsilon$ such that ${C'4^l\log^2 n}/{\varepsilon^2} \le m$ holds for all $l \le l^*_\varepsilon$. Specifically, we set
\begin{equation}
\label{eq:l-star-concave}
l^*_\varepsilon := \frac{1}{2}\left\lfloor \log_2\left(\frac{m\varepsilon^2}{C'\log^2 n}\right)\right\rfloor .
\end{equation}

For iterations where $q_t \le m$ (and thus $l_t \le l^*_\varepsilon$), the cumulative time complexity is bounded by
\begin{equation}
\begin{aligned}
\sum_{t:\,l_t\le l^*_\varepsilon} m^2 q_t^{\omega-2}
&\le \sum_{l=0}^{l^*_\varepsilon} m^2 N_l \left(\frac{C'4^l\log^2 n}{\varepsilon^2}\right)^{\omega-2} \\
&\le \frac{C'^{\omega-2}m^2T\log^{2(\omega-2)}n}{\varepsilon^{2\omega-4}}
\sum_{l=0}^{l^*_\varepsilon} 2^{(2\omega-5)l} \\
&= \widetilde{\mathcal{O}}\left(\frac{m^2}{\varepsilon^{2\omega-3}}\right).
\end{aligned}
\end{equation}

For iterations where $q_t > m$, the total cost is
\begin{equation}
\begin{aligned}
\sum_{t:\,l_t>l^*_\varepsilon} nm^{\omega-1}
&= \sum_{l=l^*_\varepsilon+1}^{\lceil \log n \rceil} N_l\,nm^{\omega-1} \\
&\le Tnm^{\omega-1}\sum_{l=l^*_\varepsilon+1}^{\lceil \log n \rceil} \frac{1}{2^l} \\
&\le Tnm^{\omega-1}\frac{1}{2^{l^*_\varepsilon}} \\
&\le \widetilde{\mathcal{O}}\left(\frac{nm^{\omega-1}T}{\varepsilon\sqrt m}\right) \\
&= \widetilde{\mathcal{O}}\left(\frac{nm^{\omega-1.5}}{\varepsilon^2}\right).
\end{aligned}
\end{equation}

Beyond the matrix updates, computing $\widetilde{\boldsymbol{d}}_t$ and updating $\boldsymbol{x}_{t+1}$ requires $\mathcal{O}(mn)$ operations per iteration, contributing $\mathcal{O}(mn/\varepsilon)$ across all $T$ iterations. Factoring in the $\widetilde{\mathcal{O}}(nm^{\omega-1})$ initial cost to construct $(\boldsymbol{A}\overline{\boldsymbol{X}}_0^2\boldsymbol{A}^\top)^{-1}$, the aggregate time complexity becomes 
\begin{equation}
\widetilde{\mathcal{O}}\left(
nm^{\omega-1}
+
\frac{mn}{\varepsilon}
+
\frac{m^2}{\varepsilon^{2\omega-3}}
+
\frac{nm^{\omega-1.5}}{\varepsilon^2}
\right).
\end{equation}
Since $m \le n$ and $\omega < 2.5$, the above bound reduces to $\widetilde{\mathcal{O}}(nm^{\omega-1}+nm/\varepsilon^2)$.
\end{proof}

\begin{remark}
In fact, as noted above, the sparse-update analysis applies only to iterations with $q_t \le m$, equivalently, to those satisfying $l_t \le l^*_\varepsilon$. Hence, when $l^*_\varepsilon<0$, this regime is empty, and the overall complexity bound reduces to the direct per-iteration cost $\mathcal{O}(nm^{\omega-1})$. Therefore, in the concave case, \algo{robust-1st-order} admits the overall time complexity $\min\{\widetilde{\mathcal{O}}(nm^{\omega-1}+nm/\varepsilon^2),\, \mathcal{O}(nm^{\omega-1}/\varepsilon)\}$. Up to logarithmic factors, the comparison is determined by $nm/\varepsilon^2$ and $nm^{\omega-1}/\varepsilon$. Hence, the bound $\widetilde{\mathcal{O}}(nm^{\omega-1}+nm/\varepsilon^2)$ is better when $\varepsilon=\Omega(m^{2-\omega})$, whereas the bound $\mathcal{O}(nm^{\omega-1}/\varepsilon)$ is better when $\varepsilon=O(m^{2-\omega})$. Thus, the approximate update implementation is preferable when moderate accuracy suffices, whereas the direct per-iteration implementation is preferable when higher accuracy is required.
\end{remark}

\section{Finding Second-Order KKT Points by First-Order Methods}
\label{sec:2nd-KKT}

In this section, we show how to find approximate second-order KKT points using solely first-order information. Comparing \prop{kkt1} and \prop{kkt2}, the additional requirement for an approximate second-order KKT point is that the minimum eigenvalue of the projected scaled Hessian $\boldsymbol{X}\nabla^2 f(\boldsymbol{x})\boldsymbol{X}$ restricted to $\ker(\boldsymbol{A}\boldsymbol{X})$ is lower bounded. Since the existing first-order IPTR algorithm only guarantees convergence to a first-order KKT point, it may get trapped at saddle points of the Lagrangian. Therefore, we need to verify this second-order condition and escape saddle points without explicit Hessian computation.

To this end, we aim to detect whether $\lambda_{\min}\{\boldsymbol{X}\nabla^2 f(\boldsymbol{x})\boldsymbol{X}\}_{\boldsymbol{A}\boldsymbol{X}}<-\sqrt{\varepsilon}$ using only gradient evaluations, and if so, to find a corresponding negative curvature direction. Specifically, for any such $\boldsymbol{x}$, we seek a unit vector $\widehat{\boldsymbol{e}}$ satisfying $\boldsymbol{A}\boldsymbol{X}\widehat{\boldsymbol{e}}=0$ and $\widehat{\boldsymbol{e}}^\top \boldsymbol{X}\nabla^2 f(\boldsymbol{x})\boldsymbol{X}\widehat{\boldsymbol{e}} \leq -\sqrt{\varepsilon}/4$. In the following, we let $\boldsymbol{P}$ denote the orthogonal projector onto $\ker(\boldsymbol{A}\boldsymbol{X})$.

\begin{procedureblock}[!htb]
    \SetAlgoLined
    \caption{Negative Curvature Finding($\boldsymbol{x}$)}
    \label{proc:nega-curvature}
    
    $\boldsymbol{y}_0 \leftarrow$ Uniform ball with radius $r$ centered at $\boldsymbol{x}$ in the $\ker(\boldsymbol{A}\boldsymbol{X})$\;
    \For{$t=0,\ldots,\mathcal{T}-1$}{
       $\boldsymbol{y}_{t+1}\leftarrow \boldsymbol{y}_{t}-\frac{\Vert\boldsymbol{y}_{t} \Vert}{lr} \boldsymbol{P}\boldsymbol{X}\left(\nabla f(\boldsymbol{X}(\boldsymbol{e}+\frac{r\boldsymbol{y}_{t}}{\Vert\boldsymbol{y}_{t} \Vert}))-\nabla f(\boldsymbol{x})\right)$\;
    }
\end{procedureblock}

The \proc{nega-curvature} can be viewed as a power method for identifying a direction associated with the minimum eigenvalue. Indeed, since $\nabla f(\boldsymbol{X}(\boldsymbol{e}+\frac{r\boldsymbol{y}_{t}}{\Vert\boldsymbol{y}_{t} \Vert}))-\nabla f(\boldsymbol{x}) \approx \nabla^2f(\boldsymbol{x})\boldsymbol{X}\frac{r\boldsymbol{y}_{t}}{\Vert\boldsymbol{y}_{t} \Vert}$, the update satisfies $\boldsymbol{y}_{t+1} \approx \left(I-\frac{1}{l}\boldsymbol{P}\boldsymbol{X}\nabla^2 f(\boldsymbol{x})\boldsymbol{X}\right)\boldsymbol{y}_{t}$. Therefore, the iteration amplifies the component corresponding to sufficiently negative eigenvalues. \prop{nega-curva} shows that \proc{nega-curvature} identifies a negative curvature direction with high probability.

\begin{proposition}
\label{prop:nega-curva}
  Under \assum{grad-lip} and \assum{hessian-lip}, if $\lambda_{\min}\{\boldsymbol{X}\nabla^2f(\boldsymbol{x})\boldsymbol{X}\}_{\boldsymbol{A}\boldsymbol{X}} < -\sqrt{ \varepsilon}$, then \proc{nega-curvature} with parameters $\dim(\ker{(\boldsymbol{A}\boldsymbol{X})}) = k$, $\mathcal{T} = \frac{8l}{\sqrt{\varepsilon}}\log \left( \frac{8l}{\delta_0} \sqrt{\frac{n}{\pi \varepsilon}}\right) $ and $r = \frac{(\frac{1+\frac{7\sqrt{ \varepsilon}}{8l}}{2})^\mathcal{T}}{8\rho} \sqrt{\frac{\pi\varepsilon}{k}} \delta_0$ 
  finds a negative curvature $\widehat{\boldsymbol{e}}$ satisfying $\boldsymbol{A}\boldsymbol{X} \widehat{\boldsymbol{e}}=0$ and $ \widehat{\boldsymbol{e}}^\top \boldsymbol{X} \nabla^2f(\boldsymbol{x})\boldsymbol{X}\widehat{\boldsymbol{e}} \leq -\sqrt{ \varepsilon}/4$ with probability at least $1-\delta_0$.
\end{proposition}
\begin{proof}
We first show that $\boldsymbol{y}_{t} \in \ker(\boldsymbol{A}\boldsymbol{X})$ for all $t \geq 0$. We proceed by induction. The base case holds since $\boldsymbol{y}_{0} \in \ker(\boldsymbol{A}\boldsymbol{X})$ by initialization. Assume $\boldsymbol{y}_{t} \in \ker(\boldsymbol{A}\boldsymbol{X})$. The update rule gives:
\begin{equation*}
    \boldsymbol{A}\boldsymbol{X}\boldsymbol{y}_{t+1} = \boldsymbol{A}\boldsymbol{X}\boldsymbol{y}_{t} - \frac{\Vert\boldsymbol{y}_{t} \Vert}{lr} \boldsymbol{A}\boldsymbol{X}\boldsymbol{P}\left(\nabla f\left(\boldsymbol{X}\left(\boldsymbol{e}+\frac{r\boldsymbol{y}_{t}}{\Vert\boldsymbol{y}_{t} \Vert}\right)\right)-\nabla f(\boldsymbol{x})\right).
\end{equation*}
Since $\boldsymbol{P}$ is the orthogonal projection onto $\ker(\boldsymbol{A}\boldsymbol{X})$, we have $\boldsymbol{A}\boldsymbol{X}\boldsymbol{P}=\boldsymbol{0}$. Thus, $\boldsymbol{A}\boldsymbol{X}\boldsymbol{y}_{t+1} = \boldsymbol{A}\boldsymbol{X}\boldsymbol{y}_{t} = \boldsymbol{0}$, meaning $\boldsymbol{y}_{t+1} \in \ker(\boldsymbol{A}\boldsymbol{X})$. Moreover, as long as $r<1$, the point $\boldsymbol{x}^\prime:=\boldsymbol{X}\left(\boldsymbol{e}+{r\boldsymbol{y}_{t}}/{\Vert\boldsymbol{y}_{t} \Vert}\right)$ satisfies $x^\prime_i = x_i\left(1+{r{y}_{{t},i}}/{\Vert\boldsymbol{y}_{t} \Vert}\right)>0$. We also have $\boldsymbol{A}\boldsymbol{x}^\prime = \boldsymbol{A}\boldsymbol{x}+\boldsymbol{A}\boldsymbol{X}{r\boldsymbol{y}_{t}}/{\Vert\boldsymbol{y}_{t} \Vert}=\boldsymbol{b}$. Therefore, $\boldsymbol{x}^\prime\in \Omega^{\circ}$.

Let us define the approximation error $\Delta$ as
\label{append:pf-nega-curva}
\begin{equation}
    \Delta:=\frac{\Vert\boldsymbol{y}_{t} \Vert}{lr}\boldsymbol{P}\boldsymbol{X}\left(\nabla f(\boldsymbol{X}(\boldsymbol{e}+\frac{r\boldsymbol{y}_{t}}{\Vert\boldsymbol{y}_{t} \Vert}))-\nabla f(\boldsymbol{x})-\nabla^2 f(\boldsymbol{x}) \boldsymbol{X}\frac{r\boldsymbol{y}_{t}}{\Vert\boldsymbol{y}_{t} \Vert}\right).
\end{equation}
With this definition, the update rule can be rewritten as
\begin{align}
    \boldsymbol{y}_{t+1} = \boldsymbol{y}_{t}-(\Delta+\frac{1}{l}\boldsymbol{P}\boldsymbol{X}\nabla^2 f(\boldsymbol{x}) \boldsymbol{X}\boldsymbol{y}_{t})
    =\left(I-\frac{1}{l}\boldsymbol{P}\boldsymbol{X}\nabla^2 f(\boldsymbol{x})\boldsymbol{X}\right) \boldsymbol{y}_{t} - \Delta.
\end{align}
By Taylor's theorem, there exists $\xi \in [0,1]$ such that the norm of $\Delta$ can be bounded using \assum{hessian-lip}
\begin{equation}
\begin{aligned}
    \frac{\Vert\Delta \Vert}{\Vert\boldsymbol{y}_{t} \Vert} &= \frac{1}{lr} \left\Vert \boldsymbol{P}\boldsymbol{X}\left(\nabla f(\boldsymbol{X}(\boldsymbol{e}+\frac{r\boldsymbol{y}_{t}}{\Vert\boldsymbol{y}_{t} \Vert}))-\nabla f(\boldsymbol{x})-\nabla^2 f(\boldsymbol{x}) \boldsymbol{X}\frac{r\boldsymbol{y}_{t}}{\Vert\boldsymbol{y}_{t} \Vert}\right)\right\Vert \\
    &\leq \frac{1}{lr} \left\Vert \boldsymbol{X}\left(\nabla^2 f(\boldsymbol{X}(\boldsymbol{e}+\frac{\xi r\boldsymbol{y}_{t}}{\Vert\boldsymbol{y}_{t} \Vert}))\frac{r\boldsymbol{X}\boldsymbol{y}_{t}}{\Vert\boldsymbol{y}_{t} \Vert}-\nabla^2 f(\boldsymbol{x}) \boldsymbol{X}\frac{r\boldsymbol{y}_{t}}{\Vert\boldsymbol{y}_{t} \Vert}\right)\right\Vert \\
    &\leq \frac{1}{l} \left\Vert \left(\boldsymbol{X}\nabla^2 f(\boldsymbol{X}(\boldsymbol{e}+\frac{\xi r\boldsymbol{y}_{t}}{\Vert\boldsymbol{y}_{t} \Vert}))\boldsymbol{X}-\boldsymbol{X}\nabla^2 f(\boldsymbol{x}) \boldsymbol{X}\right)\right\Vert \\
    & \leq \frac{\rho}{l} \left\Vert \frac{\xi r\boldsymbol{y}_{t}}{\Vert\boldsymbol{y}_{t} \Vert} \right\Vert \\
    & \leq \frac{\rho r}{l}.
\end{aligned}
\end{equation}

The eigenvectors of the projected scaled Hessian $\boldsymbol{P}\boldsymbol{X}\nabla^2 f(\boldsymbol{x})\boldsymbol{X}\boldsymbol{P}$ restricted to $\ker(\boldsymbol{A}\boldsymbol{X})$ form an orthogonal basis. Let $\lambda_1 \leq \lambda_2 \leq \cdots \leq \lambda_k$ be its eigenvalues, with corresponding eigenvectors $\boldsymbol{u}_1,\ldots,\boldsymbol{u}_k$. By assumption, there exist indices $p$ and $p^{\prime}$ such that
\begin{equation}
    \lambda_p \leq -\sqrt{\varepsilon} \leq \lambda_{p+1}, \quad \text{and} \quad \lambda_{p^\prime} \leq -\sqrt{\varepsilon}/2 < \lambda_{p^\prime+1}.
\end{equation}
We define the subspaces $\mathcal{S}_{\parallel} = \mathrm{span}\{\boldsymbol{u}_1,\ldots,\boldsymbol{u}_p\}$ and $\mathcal{S}_{\perp} = \mathrm{span}\{\boldsymbol{u}_{p+1},\ldots,\boldsymbol{u}_k\}$. Similarly, we define $\mathcal{S}_{\parallel^{\prime}} = \mathrm{span}\{\boldsymbol{u}_1,\ldots,\boldsymbol{u}_{p^{\prime}}\}$ and $\mathcal{S}_{\perp^{\prime}} = \mathrm{span}\{\boldsymbol{u}_{p^{\prime}+1},\ldots,\boldsymbol{u}_k\}$. 

Before proceeding, we establish the properties of the initial vector $\boldsymbol{y}^0$. Since $\boldsymbol{y}^0$ is initialized uniformly at random from the unit sphere in the $k$-dimensional subspace $\ker(\boldsymbol{A}\boldsymbol{X})$, standard results on random projections guarantee that its projection onto the subspace $\mathcal{S}_{\parallel}$ (which has dimension $p \geq 1$) satisfies
\begin{equation}
    \mathbb{P}\left( \frac{\Vert \boldsymbol{y}_{0,\parallel} \Vert}{\Vert \boldsymbol{y}_{0} \Vert} \geq \sqrt{\frac{\pi}{k}} \delta_0 \right) \geq 1-\delta_0.
\end{equation}
We condition on this high-probability event for the remainder of the proof and define $\alpha_0 :=\frac{\Vert \boldsymbol{y}_{0,\parallel} \Vert}{\Vert \boldsymbol{y}_{0} \Vert} \geq \sqrt{\frac{\pi}{k}} \delta_0$. 

Projecting $\boldsymbol{y}_{t+1}$ onto the subspace $\mathcal{S}_{\parallel}$ yields 
\begin{equation}
\begin{aligned}
    \Vert \boldsymbol{y}_{{t+1},{\parallel}} \Vert &\geq \left(1+\frac{\sqrt{ \varepsilon}}{l}\right) \Vert \boldsymbol{y}_{{t},{\parallel}} \Vert -\Vert \Delta \Vert\\
    &\geq \left(1+\frac{\sqrt{ \varepsilon}}{l} -\frac{\Vert \Delta \Vert}{\Vert \boldsymbol{y}_{t} \Vert} \frac{\Vert \boldsymbol{y}_{t} \Vert}{\Vert \boldsymbol{y}_{{t},{\parallel}} \Vert } \right) \Vert \boldsymbol{y}_{{t},{\parallel}} \Vert \\
    & \geq \left(1+\frac{\sqrt{ \varepsilon}}{l} -\frac{\rho r}{l} \frac{\Vert \boldsymbol{y}_{t} \Vert}{\Vert \boldsymbol{y}_{{t},{\parallel}} \Vert } \right) \Vert \boldsymbol{y}_{{t},{\parallel}} \Vert .
\end{aligned}
\end{equation}
Let $\alpha_t := \frac{\Vert \boldsymbol{y}_{t,\parallel} \Vert}{\Vert \boldsymbol{y}_{t} \Vert}$. The parameter $r$ is chosen such that $r \leq \frac{\alpha_0}{8\rho} \sqrt{\varepsilon} \left(\frac{1+7\sqrt{\varepsilon}/(8l)}{2}\right)^\mathcal{T}$, which ensures $\frac{\rho r}{l \alpha_0 \left(\frac{1+7\sqrt{\varepsilon}/(8l)}{2}\right)^\mathcal{T}} \leq \frac{\sqrt{\varepsilon}}{8l}$. We will prove by induction that $\alpha_t \geq \alpha_0 \left(\frac{1+7\sqrt{\varepsilon}/(8l)}{2}\right)^{t}$. Suppose this holds for step $t$, which implies $\frac{\rho r}{l \alpha_t} \leq \frac{\sqrt{\varepsilon}}{8l}$. Then at step $t+1$,
\begin{equation}
\begin{aligned}
\label{eq:alphat-1}
    \Vert \boldsymbol{y}_{{t+1},{\parallel}} \Vert
    & \geq \left(1+\frac{\sqrt{ \varepsilon}}{l} -\frac{\rho r}{l} \frac{\Vert \boldsymbol{y}_{t} \Vert}{\Vert \boldsymbol{y}_{{t},{\parallel}} \Vert } \right) \Vert \boldsymbol{y}_{{t},{\parallel}} \Vert, \\
    \frac{\Vert \boldsymbol{y}_{{t+1},{\parallel}} \Vert}{\Vert \boldsymbol{y}_{t+1} \Vert} &\geq \frac{\Vert \boldsymbol{y}_{{t},{\parallel}} \Vert}{\Vert \boldsymbol{y}_{t} \Vert} \frac{\Vert \boldsymbol{y}_{t} \Vert}{\Vert \boldsymbol{y}_{t+1} \Vert}\left(1+\frac{\sqrt{ \varepsilon}}{l} -\frac{\rho r}{l} \frac{\Vert \boldsymbol{y}_{t} \Vert}{\Vert \boldsymbol{y}_{{t},{\parallel}} \Vert } \right), \\
    \alpha_{t+1} &\geq \alpha_t \frac{\Vert \boldsymbol{y}_{t} \Vert}{\Vert \boldsymbol{y}_{t+1} \Vert}\left(1+\frac{\sqrt{\varepsilon}}{l} -\frac{\rho r}{l\alpha_t} \right). 
\end{aligned}
\end{equation}
Next, we establish an lower bound on the ratio $\Vert \boldsymbol{y}_{t} \Vert/{\Vert \boldsymbol{y}_{t+1} \Vert}$. Using the triangle inequality and \assum{grad-lip}, we have
\begin{equation}
\begin{aligned}
\label{eq:alphat-2}
    \Vert \boldsymbol{y}_{t+1} \Vert &\leq \Vert \boldsymbol{y}_{t} \Vert  \left( 1 +\frac{1}{lr} \left\Vert \boldsymbol{P}\boldsymbol{X}\left(\nabla f(\boldsymbol{X}(\boldsymbol{e}+\frac{r\boldsymbol{y}_{t}}{\Vert\boldsymbol{y}_{t} \Vert}))-\nabla f(\boldsymbol{x})\right)\right\Vert \right)  \\
    & \leq \Vert \boldsymbol{y}_{t} \Vert \left(1 + \frac{1}{lr} \cdot l r \right)\\
    & \leq 2\Vert \boldsymbol{y}_{t} \Vert.
\end{aligned}
\end{equation}
Combining \eqref{eq:alphat-1} and \eqref{eq:alphat-2}, we obtain
\begin{equation}
\begin{aligned}
    \alpha_{t+1} &\geq \alpha_t \frac{1+\frac{\sqrt{\varepsilon}}{l} -\frac{\rho r}{l\alpha_t}}{2}\geq \alpha_t \frac{1+\frac{7\sqrt{\varepsilon}}{8l}}{2}  \geq \alpha_0\left(\frac{1+7\sqrt {\varepsilon}/(8l)}{2}\right)^{t+1}.
\end{aligned}
\end{equation}
This completes the induction. Consequently, for all $1\leq t \leq \mathcal{T}$, we have $\frac{\rho r}{l\alpha_t} \leq \frac{\sqrt{\varepsilon}}{8l}$ and
\begin{equation}
    \Vert \boldsymbol{y}_{t,\parallel} \Vert \geq \left(1+\frac{7\sqrt{\varepsilon}}{8l} \right) \Vert \boldsymbol{y}_{t-1,\parallel} \Vert \geq \left(1+\frac{7\sqrt{\varepsilon}}{8l} \right)^t \Vert \boldsymbol{y}_{0,\parallel} \Vert.
\end{equation}

Let $\boldsymbol{y}_{t,\perp^{\prime}}$ denote the projection of $\boldsymbol{y}_{t}$ onto the subspace $\mathcal{S}_{\perp^{\prime}}$. We will show that there exists $1 \leq t_0 \leq \mathcal{T}$ such that $\Vert \boldsymbol{y}_{t_0,\perp^{\prime}} \Vert/\Vert \boldsymbol{y}_{t_0} \Vert \leq \frac{\sqrt{\varepsilon}}{8l}$. Suppose, for the sake of contradiction, that $\Vert \boldsymbol{y}_{t,\perp^{\prime}} \Vert/\Vert \boldsymbol{y}_{t} \Vert > \frac{\sqrt{\varepsilon}}{8l}$ holds for all $1 \leq t \leq \mathcal{T}$. Under this assumption, we have
\begin{equation}
\begin{aligned}
    \Vert \boldsymbol{y}_{{t+1},{\perp^{\prime}}} \Vert &\leq \left(1+\frac{\sqrt{\varepsilon}}{2l}\right) \Vert \boldsymbol{y}_{{t},{\perp^{\prime}}} \Vert + \Vert \Delta \Vert\\
    &\leq \left(1+\frac{\sqrt{\varepsilon}}{2l} +\frac{\Vert \Delta \Vert}{\Vert \boldsymbol{y}_{t} \Vert} \frac{\Vert \boldsymbol{y}_{t} \Vert}{\Vert \boldsymbol{y}_{{t},{\perp^{\prime}}} \Vert } \right) \Vert \boldsymbol{y}_{{t},{\perp^{\prime}}} \Vert \\
    & \leq \left(1+\frac{\sqrt{\varepsilon}}{2l} +\frac{\rho r}{l} \frac{\Vert \boldsymbol{y}_{t} \Vert}{\Vert \boldsymbol{y}_{{t},{\perp^{\prime}}} \Vert } \right) \Vert \boldsymbol{y}_{{t},{\perp^{\prime}}} \Vert \\
    & \leq \left(1+\frac{\sqrt{\varepsilon}}{2l} +\frac{\rho r}{l} \frac{8l}{\sqrt{\varepsilon}} \right) \Vert \boldsymbol{y}_{{t},{\perp^{\prime}}} \Vert \\
    & \leq \left(1+\frac{5\sqrt{\varepsilon}}{8l} \right) \Vert \boldsymbol{y}_{{t},{\perp^{\prime}}} \Vert .
\end{aligned}
\end{equation}
Combining this with the initial condition $\mathcal{P}\left( \Vert \boldsymbol{y}^{0}_{\parallel} \Vert / \Vert \boldsymbol{y}^{0} \Vert \geq \sqrt{\frac{\pi}{k}} \delta_0  \right) \geq 1-\delta_0 $  and $\mathcal{T} = \frac{8l}{\sqrt{\varepsilon}}\log \left( \frac{8l}{\delta_0} \sqrt{\frac{n}{\pi \varepsilon}}\right) $, we obtain
\begin{equation}
    \frac{\Vert \boldsymbol{y}_{{\mathcal{T}},{\perp^{\prime}}} \Vert}{\Vert \boldsymbol{y}_{{\mathcal{T}},{\parallel}} \Vert} \leq \frac{\Vert \boldsymbol{y}_{{0},{\perp^{\prime}}} \Vert\left(1+\frac{5\sqrt{\varepsilon}}{8l} \right)^\mathcal{T}}{\Vert \boldsymbol{y}_{{0},{\parallel}} \Vert \left(1+\frac{7\sqrt{\varepsilon}}{8l}  \right)^\mathcal{T}} 
    \leq \frac{1}{\delta_0} \sqrt{\frac{k}{\pi}} \left( \frac{1+\frac{5\sqrt{\varepsilon}}{8l} }{1+\frac{7\sqrt{\varepsilon}}{8l} } \right)^\mathcal{T} 
    \leq \frac{1}{\delta_0} \sqrt{\frac{n}{\pi}} \left( \frac{1+\frac{5\sqrt{\varepsilon}}{8l} }{1+\frac{7\sqrt{\varepsilon}}{8l} } \right)^\mathcal{T} \leq \frac{\sqrt{\varepsilon}}{8l}.
\end{equation}
This contradicts the assumption that $\Vert \boldsymbol{y}_{{t},{\perp^{\prime}}} \Vert/\Vert \boldsymbol{y}_{t} \Vert > \frac{\sqrt{\varepsilon}}{8l}$ for all $1 \leq t \leq \mathcal{T}$. Therefore, there exists $1 \leq t_0 \leq \mathcal{T}$ such that $\Vert \boldsymbol{y}_{{t_0},{\perp^{\prime}}} \Vert/\Vert \boldsymbol{y}_{t_0} \Vert \leq \frac{\sqrt{\varepsilon}}{8l}$. Let $\widehat{\boldsymbol{e}}:=\boldsymbol{y}_{t_0}/\Vert \boldsymbol{y}_{t_0} \Vert$ denotes the normalized vector. Then $\boldsymbol{A}\boldsymbol{X}\widehat{\boldsymbol{e}} = \boldsymbol{0}$, $\Vert \widehat{\boldsymbol{e}}_{\perp^{\prime}} \Vert \leq \frac{\sqrt{\varepsilon}}{8l}$, and $\Vert \widehat{\boldsymbol{e}}_{\parallel^\prime} \Vert^2 \geq 1-\frac{\varepsilon}{64l^2}$. Consequently,
\begin{equation}
    \widehat{\boldsymbol{e}}^\top \boldsymbol{X}\nabla^2f(\boldsymbol{x})\boldsymbol{X} \widehat{\boldsymbol{e}} = \widehat{\boldsymbol{e}}_{\perp^{\prime}}^\top \boldsymbol{X}\nabla^2f(\boldsymbol{x})\boldsymbol{X} \widehat{\boldsymbol{e}}_{\perp^{\prime}} + \widehat{\boldsymbol{e}}_{\parallel^\prime}^\top \boldsymbol{X}\nabla^2f(\boldsymbol{x})\boldsymbol{X} \widehat{\boldsymbol{e}}_{\parallel^\prime} \leq l \Vert \widehat{\boldsymbol{e}}_{\perp^{\prime}} \Vert^2 - \sqrt{\varepsilon}\Vert \widehat{\boldsymbol{e}}_{\parallel^{\prime}} \Vert^2 /2 \leq - \sqrt{\varepsilon}/4.
\end{equation}
This completes the proof of the proposition.
\end{proof}

\prop{nega-curva} shows that \proc{nega-curvature} successfully identifies a negative curvature direction with high probability. Once this direction is found, the next step is to use it to escape the current saddle point of the Lagrangian. \prop{func-decrease} demonstrates that taking a step along this negative curvature direction guarantees a sufficient decrease in the objective function value.
\begin{proposition}
\label{prop:func-decrease}
    Under \assum{grad-lip} and \assum{hessian-lip}, for any point $\boldsymbol{x}\in \Omega^\circ$, if there exists a vector $\widehat{\boldsymbol{e}}$ such that $\boldsymbol{A}\boldsymbol{X} \widehat{\boldsymbol{e}}=0  $ and $\widehat{\boldsymbol{e}}^\top \boldsymbol{X} \nabla^2f(\boldsymbol{x}) \boldsymbol{X} \widehat{\boldsymbol{e}} \leq -\sqrt{ \varepsilon}/4$, then the following holds:
    \begin{equation}
    f\left(\boldsymbol{X}\left(\boldsymbol{e}- \frac{\langle \nabla f(\boldsymbol{x}), \boldsymbol{X}\widehat{\boldsymbol{e}} \rangle }{\vert\langle \nabla f(\boldsymbol{x}), \boldsymbol{X} \widehat{\boldsymbol{e}} \rangle \vert} 
    \cdot \frac{3\sqrt{\varepsilon}}{ 8\rho}\widehat{\boldsymbol{e}}\right)\right)  
     \leq  f(\boldsymbol{x})
    - \frac{9\sqrt{\varepsilon^{3}}}{1024\rho^2}.
    \end{equation}
\end{proposition}
\begin{proof}
Under \assum{grad-lip} and \assum{hessian-lip}, the function admits the following third-order upper bound:
    \begin{equation}
\begin{aligned}
 &f\left(\boldsymbol{X}\left(\boldsymbol{e}- \frac{\langle \nabla f(\boldsymbol{x}), \boldsymbol{X}\widehat{\boldsymbol{e}} \rangle }{\vert\langle \nabla f(\boldsymbol{x}), \boldsymbol{X} \widehat{\boldsymbol{e}} \rangle \vert} 
 \cdot \frac{3\sqrt{\varepsilon}}{ 8\rho}\widehat{\boldsymbol{e}}\right)\right) \\
     \leq & f(\boldsymbol{x}) + \langle \boldsymbol{X}\nabla f(\boldsymbol{x}), - \frac{\langle \nabla f(\boldsymbol{x}), \boldsymbol{X}\widehat{\boldsymbol{e}} \rangle}{\vert\langle \nabla f(\boldsymbol{x}), \boldsymbol{X}\widehat{\boldsymbol{e}} \rangle \vert } \cdot \frac{3\sqrt{\varepsilon}}{ 8\rho} \widehat{\boldsymbol{e}} \rangle 
    + \frac{1}{2} \cdot \frac{9}{64} \cdot \frac{\varepsilon}{\rho^2} \widehat{\boldsymbol{e}}^\top \boldsymbol{X}\nabla^2f(\boldsymbol{x})\boldsymbol{X}\widehat{\boldsymbol{e}} + \frac{\rho }{6 } \cdot \frac{27}{512} \cdot\frac{\sqrt{\varepsilon^{3}}} {\rho^3} \\
     \leq & f(\boldsymbol{x})
    + \frac{9\varepsilon}{128\rho^2} \widehat{\boldsymbol{e}}^\top \boldsymbol{X}\nabla^2 f(\boldsymbol{x})\boldsymbol{X}\widehat{\boldsymbol{e}} + \frac{ 9\sqrt{\varepsilon^{3}}}{1024 \rho^2 } \\
     \leq & f(\boldsymbol{x})
    - \frac{9\sqrt{\varepsilon^{3}}}{512\rho^2} + \frac{ 9\sqrt{\varepsilon^{3}}}{1024 \rho^2 } \\ 
     \leq & f(\boldsymbol{x})
    - \frac{9\sqrt{\varepsilon^{3}}}{1024\rho^2},
\end{aligned}
\end{equation}
where the first inequality directly applies the third-order Taylor bound under \assum{hessian-lip}, and the third inequality follows from the negative curvature condition
\end{proof}

In summary, \prop{nega-curva} and \prop{func-decrease} provide a complete procedure for escaping saddle points of the the Lagrangian. When the algorithm reaches a first-order KKT point that does not satisfy the second-order condition, we can apply \proc{nega-curvature} to find a negative curvature direction. Updating the variable along this direction decreases the objective function by at least $\mathcal{O}(\varepsilon^{3/2})$. Since the objective function is bounded from below, this sufficient decrease ensures that the algorithm will not be trapped at first-order KKT point and will eventually converge to an approximate second-order KKT point.

\subsection{First-order IPTR with Negative Curvature Finding}

We now combine the negative curvature finding procedure with the existing first-order IPTR framework and prove that the resulting \algo{1st-order interior}, based solely on first-order information, finds an $(2\varepsilon,\sqrt{\varepsilon})$-KKT2 point within $\mathcal{O}(1/\varepsilon^2)$ iterations. In particular, when the first-order IPTR method reaches an approximate first-order KKT point that does not satisfy the second-order condition, the negative-curvature step guarantees a sufficient decrease in the objective value. This allows us to bound the number of iterations spent at such points, which in turn yields an upper bound on the total number of iterations.

\begin{algorithm}[!htb]
    \SetAlgoLined
    \caption{First-order IPTR with Negative Curvature Finding}
    \label{algo:1st-order interior}
    
    Initialize $\boldsymbol{x}_0$ as an approximate analytic center\;
    $T \leftarrow \max\left\{\frac{2048(f(\boldsymbol{x}_0)-f(\boldsymbol{x}^*)){\rho^2}}{{9\sqrt{\varepsilon^{3}}}}, \frac{16\left(f(\boldsymbol{x}_0)- f(\boldsymbol{x}^*) + (C_0-1)\varepsilon\right)(l+\varepsilon)}{\varepsilon^2}\right\}$ \;
    \For{$t=0,\ldots,T-1$}{
    Solve the following subproblem to obtain the solution $\boldsymbol{d}_{t}$:
       \begin{equation}
       \begin{aligned}
       \label{eq:lp-interior}
           \min\quad &\nabla \phi(\boldsymbol{x}_t)^\top \boldsymbol{X}_t \boldsymbol{d} \\
           \mathrm{s.t.}\quad &\boldsymbol{A}\boldsymbol{X}_t \boldsymbol{d} = 0,\ \Vert \boldsymbol{d} \Vert \leq \frac{\varepsilon}{l+2 \varepsilon}.
        \end{aligned}
       \end{equation}
       $\boldsymbol{x}_{t+1} \leftarrow \boldsymbol{x}_t + \boldsymbol{X}_t \boldsymbol{d}_{t}$\;
       \If{$\phi(\boldsymbol{x}_{t+1})-\phi(\boldsymbol{x}_t)>-\frac{\varepsilon^2}{4l+4\varepsilon}$}
       {
       $\widehat{\boldsymbol{e}}\leftarrow$ Negative Curvature Finding($\boldsymbol{x}_t$)\;
       $\boldsymbol{x}_{t+1} \leftarrow \boldsymbol{x}_t -  \frac{\langle \nabla f(\boldsymbol{x}_t), \boldsymbol{X}_t\widehat{\boldsymbol{e}} \rangle }{\vert\langle \nabla f(\boldsymbol{x}_t), \boldsymbol{X}_t \widehat{\boldsymbol{e}} \rangle \vert} 
    \cdot \frac{3\sqrt{\varepsilon}}{ 8\rho}\boldsymbol{X}_t\widehat{\boldsymbol{e}}$\;
    \If{$f(\boldsymbol{x}_{t+1})-f(\boldsymbol{x}_{t}) > - \frac{9\sqrt{\varepsilon^{3}}}{1024\rho^2}$ }
    {
    Return $\boldsymbol{x}_t$\;
    }
       }
    }
\end{algorithm}
At each iteration, \algo{1st-order interior} solves a linear programming subproblem of the form \eq{lp-interior}. As discussed in \eq{exact-d}, this subproblem admits a closed-form solution, which we denote by $\boldsymbol{d}_{t}$. The next iterate is then updated as $\boldsymbol{x}_{t+1} = \boldsymbol{x}_t + \boldsymbol{X}_t \boldsymbol{d}_{t}$. The following lemma formally characterizes the decrease of the potential function achieved by this step.

\begin{lemma}[{\cite[Theorem 2]{haeser2019optimality}}]
\label{lem:1st-interior-2eq}
    Under \assum{grad-lip} and \assum{hessian-lip}, For any $\varepsilon \in  (0, \min\{\gamma, 1\}]$,  either the following condition holds at iteration $t$:
    \begin{equation}
    \phi(\boldsymbol{x}_{t+1})-\phi(\boldsymbol{x}_t)\leq-\frac{\varepsilon^2}{4l+4\varepsilon}
    \end{equation}
    or 
    $
        \Vert \boldsymbol{X}_t \nabla f(\boldsymbol{x}_t) +\boldsymbol{X}_t\boldsymbol{A}^\top \boldsymbol{v}_{t} \Vert_\infty < 2\varepsilon \text{ and } \nabla f(\boldsymbol{x}_t)+\boldsymbol{A}^\top \boldsymbol{v}_{t}>0
    $ for some $\boldsymbol{v}_{t} \in \mathbb{R}^m$.
\end{lemma}

As established in \lem{1st-interior-2eq}, each iteration either yields a sufficient decrease in the potential function $\phi$ or identifies a $2\varepsilon$-KKT point. In the latter case, if the iterate is not an approximate second-order KKT point, \prop{nega-curva} and \prop{func-decrease} guarantee a further decrease in the objective $f$ of at least $\frac{9\sqrt{\varepsilon^{3}}}{1024\rho^2}$ via the negative curvature step. Based on these guaranteed decreases, \thm{1st-interior} establishes an upper bound on the total number of iterations required to find an approximate second-order KKT point.

\begin{theorem}
\label{thm:1st-interior}
Under \assum{grad-lip} and \assum{hessian-lip}, for any $\varepsilon \in (0, \min\{\beta, 1\}]$ and any $0 < \delta \leq 1$, the \algo{1st-order interior}  produces at least $T/4$ iterates $\boldsymbol{x}_t$ that are  $(2\varepsilon,\sqrt{\varepsilon})$-KKT2 points within
\begin{equation}
    \widetilde{\mathcal{O}}\left(\frac{l \rho^2(f(\boldsymbol{x}_0)-f(\boldsymbol{x}^*))}{\varepsilon^2} \log n\right) 
\end{equation}
gradient queries, with probability at least $1 - \delta$.
\end{theorem}
\begin{proof}
    During the iterations of \algo{1st-order interior}, the iterates $\boldsymbol{x}_t$ can be categorized into three types:
(i) points that are not $2\varepsilon$-KKT points;
(ii) points that satisfy the $2\varepsilon$-KKT conditions but not the $(2\varepsilon,\sqrt{\varepsilon})$-KKT2 conditions;
(iii) points that satisfy the $(2\varepsilon,\sqrt{\varepsilon})$-KKT2 conditions.

For points of the first type, the potential function decreases by at least $\frac{\varepsilon^2}{4l + 4\varepsilon}$. If all iterates belong to this category, then based on the initialization in \eq{initial}, we have
\begin{equation}
    f(\boldsymbol{x}_t)-f(\boldsymbol{x}_0) \leq -\frac{t\varepsilon^2}{4l + 4\varepsilon} + \varepsilon C_0.
\end{equation}
Consequently, the number of such iterations is bounded by $T_1 = \frac{\left(f(\boldsymbol{x}_0)- f(\boldsymbol{x}^*) + (C_0-1)\varepsilon\right)(4l+4\varepsilon)}{\varepsilon^2}$, where $f(\boldsymbol{x}^*)$ denotes the optimal value of $f$. Exceeding this bound would imply that $f(\boldsymbol{x}_t)$ falls below the optimal value, which is impossible.

For points of the second type, invoking \prop{nega-curva} with
\begin{equation}
    \delta_0 = \frac{{9\sqrt{\varepsilon^{3}}}}{1024(f(\boldsymbol{x}_0)-f(\boldsymbol{x}^*)){\rho^2}} \delta,
\end{equation}
with probability at least $1 - \delta_0$, \proc{nega-curvature} will find a negative curvature direction $\widehat{\boldsymbol{e}}$. According to \prop{func-decrease}, moving along this direction decreases the function value by at least $\frac{9\sqrt{\varepsilon^{3}}}{1024\rho^2}$. Therefore, the number of points satisfying the $2\varepsilon$-KKT conditions but not the $(2\varepsilon,\sqrt{\varepsilon})$-KKT2 conditions is bounded by $T_2 = \frac{1024(f(\boldsymbol{x}_0)-f(\boldsymbol{x}^*)){\rho^2}}{{9\sqrt{\varepsilon^{3}}}}$ with probability at least $1-\delta$. Combining both cases, define
\begin{equation}
    T = \max\left\{2T_2, 4T_1\right\} = \max\left\{\frac{2048(f(\boldsymbol{x}_0)-f(\boldsymbol{x}^*)){\rho^2}}{{9\sqrt{\varepsilon^{3}}}}, \frac{16\left(f(\boldsymbol{x}_0)- f(\boldsymbol{x}^*) + (C_0-1)\varepsilon\right)(l+\varepsilon)}{\varepsilon^2}\right\}.
\end{equation}
Under this bound, at most $T/4$ iterations correspond to points of the first type, and at most $T/2$ iterations correspond to points of the second type with probability at least $1-\delta$. Therefore, at least $T/4$ iterations correspond to points of the third type, i.e., $(2\varepsilon,\sqrt{\varepsilon})$-KKT2 points with probability at least $1-\delta$.

In terms of gradient queries, the number required by the first-order interior point method is bounded by $\mathcal{O}(T)=  \mathcal{O}(\frac{l\rho^2(f(\boldsymbol{x}_0)- f(\boldsymbol{x}^*)) }{\varepsilon^2})$. For the negative curvature steps, the total number of gradient queries depend on the number of negative curvature searches performed, where each search requires $\widetilde{\mathcal{O}}(\frac{l }{\sqrt{\varepsilon}})$ gradient queries. Consequently, the total number of gradient queries for points of the second type is bounded by 
\begin{equation}
    \frac{2048(f(\boldsymbol{x}_0)-f(\boldsymbol{x}^*)){\rho^2}}{{9\sqrt{\varepsilon^{3}}}}  \cdot \widetilde{\mathcal{O}}\left(\frac{l }{\sqrt{\varepsilon}}\right) = \widetilde{\mathcal{O}}\left(\frac{l \rho^2(f(\boldsymbol{x}_0)-f(\boldsymbol{x}^*))}{\varepsilon^2}  \right)
\end{equation} 
with probability at least $1 - \delta$. Therefore, the overall number of gradient queries is upper bounded by $\widetilde{\mathcal{O}}(\frac{l \rho^2(f(\boldsymbol{x}_0)-f(\boldsymbol{x}^*))}{\varepsilon^2}) $ also with probability at least $1-\delta$.
\end{proof}

Having established the iteration and gradient query complexities in \thm{1st-interior}, we now evaluate the overall time complexity of \algo{1st-order interior}. The total computational cost comes directly from solving the linear programming subproblem \eq{lp-interior} and the negative curvature finding procedure. The following proposition provides an upper bound on this overall time complexity.

\begin{proposition}
    The overall time complexity of \algo{1st-order interior} is uppper bounded by
        $\widetilde{\mathcal{O}}\left( {nm^{\omega-1}}/{\varepsilon^2}  \right)$.
\end{proposition}
\begin{proof}
At each iteration $t$, we need to solve a linear programming problem subject to the equality constraint $\boldsymbol{A}\boldsymbol{X}_t \boldsymbol{d} = \boldsymbol{0}$ and the ball constraint $\Vert \boldsymbol{d} \Vert \leq \frac{\varepsilon}{l + 2\varepsilon}$. This linear programming problem admits a closed-form solution as defined in \eq{exact-d}. 

Computing this solution requires the projection matrix $\boldsymbol{P}_t$ onto the null space of $\boldsymbol{A}\boldsymbol{X}_t$. To construct $\boldsymbol{P}_t$, we first compute the orthogonal basis $\boldsymbol{Z}_t$ of $\ker(\boldsymbol{A}\boldsymbol{X}_t)$. The time complexity of this step is $\mathcal{O}(n m^{\omega-1})$ due to the cost of rectangular matrix multiplication \cite{demmel2007fast}, where $\omega$ is the matrix-multiplication exponent. Once $\boldsymbol{Z}_t$ is obtained, the projection matrix can be computed as $\boldsymbol{P}_t = \boldsymbol{Z}_t \boldsymbol{Z}_t^\top$. Subsequently, the computation of $\boldsymbol{d}_{t}$ involves a matrix-vector multiplication, which incurs a cost of $\mathcal{O}(mn)$. Updating $\boldsymbol{x}_{t+1}$ requires multiplying a diagonal matrix with a vector, with a time complexity of $\mathcal{O}(n)$. Therefore, each iteration of solving the interior point trust region problem has an overall time complexity of $\mathcal{O}(n m^{\omega-1})$. According to \thm{1st-interior}, \eq{lp-interior} is invoked at most $ \frac{\left(f(\boldsymbol{x}_0)- f(\boldsymbol{x}^*) + (C_0-1)\varepsilon\right)(4l+4\varepsilon)}{\varepsilon^2}$ times. Consequently, the overall time complexity of solving \eq{lp-interior} is bounded by $\mathcal{O}(n m^{\omega-1}/\varepsilon^2)$.

Then we bound the time complexity for negative curvature finding. In \proc{nega-curvature}, the first step is to compute the projection matrix $\boldsymbol{P}$ onto $\ker(\boldsymbol{A}\boldsymbol{X})$, which requires $\mathcal{O}(n m^{\omega-1})$ time. In each subsequent iteration of \proc{nega-curvature}, updating $\boldsymbol{y}^t$ only involves a matrix-vector multiplication, which takes $\mathcal{O}(mn)$ time. The algorithm proceeds for $\mathcal{T} = \frac{8l}{\sqrt{\varepsilon}}\log \left( \frac{8l}{\delta_0} \sqrt{\frac{n}{\pi \varepsilon}}\right)$ steps in total. Hence, a single invocation of \proc{nega-curvature} incurs a time complexity of $\widetilde{\mathcal{O}}(n m^{\omega-1} + mn/\sqrt{\varepsilon})$. According to \thm{1st-interior}, \proc{nega-curvature} will be invoked $\frac{1024(f(\boldsymbol{x}_0)-f(\boldsymbol{x}^*)){\rho^2}}{{9\sqrt{\varepsilon^{3}}}}$ times with high probability. Therefore, the overall time complexity for calling \proc{nega-curvature} is $\widetilde{\mathcal{O}}(n m^{\omega-1}/\sqrt{\varepsilon^3} + mn/\varepsilon^2)$.

Combining both parts, the total time complexity is bounded by $\widetilde{\mathcal{O}}\left( {nm^{\omega-1}}/{\varepsilon^2}  \right)$, which completes the proof.
\end{proof}

\subsection{\Robustalgo First-order IPTR with Negative Curvature Finding}

To further reduce the overall time complexity, we combine the \robustalgo IPTR framework (\algo{robust-1st-order}) with the negative-curvature-finding procedure (\proc{nega-curvature}). The main idea is that the costly exact projection is needed only when the algorithm arrives at an approximate first-order KKT point and seeks a direction of negative curvature. Since \proc{nega-curvature} is invoked at most $\mathcal{O}(1/\varepsilon^{1.5})$ times, as established in \thm{1st-interior}, we may use the cheaper approximate projection  in most of the $\mathcal{O}(1/\varepsilon^2)$ iterations. This leads to \algo{robust-nega}, which guarantees convergence to a $(2\varepsilon,\sqrt{\varepsilon})$-KKT2 point while substantially reducing the overall time complexity.

\vspace{4mm}
\begin{algorithm}[!htb]
    \SetAlgoLined
    \caption{\Robustalgo First-order IPTR with Negative Curvature Finding}
    \label{algo:robust-nega}
    
    Initialize $\boldsymbol{x}_0$ as an approximate analytic center\;
    $T \leftarrow \max\left\{\frac{2048(f(\boldsymbol{x}_0)-f(\boldsymbol{x}^*)){\rho^2}}{{9\sqrt{\varepsilon^{3}}}}, \frac{4\left(f(\boldsymbol{x}_0)- f(\boldsymbol{x}^*) + (C_0-1)\varepsilon\right)(l+2\varepsilon+2)}{\varepsilon^2}\right\}$\;
    $\beta\leftarrow \varepsilon/(l+2\varepsilon+2), \delta_{\text{err}} \leftarrow \min(\varepsilon/(15L_\phi), \beta / (92 L_{\phi})) $ and $\overline{\boldsymbol{x}}_{0} \leftarrow \boldsymbol{x}_{0}$\;
    
    \For{$t=1,\ldots,T$}{
    Approximate the subproblem \eq{lp-interior} using the projection matrix
    \begin{equation}
        \boldsymbol{R}_t: =  \boldsymbol{I} - \boldsymbol{X}_{t}^{-1} \overline{\boldsymbol{X}}^2_{t} \boldsymbol{A}^\top 
    (\boldsymbol{A} \overline{\boldsymbol{X}}^2_{t} \boldsymbol{A}^\top)^{-1} \boldsymbol{A}\boldsymbol{X}_{t};
    \end{equation}
    \If{$\boldsymbol{R}_t \boldsymbol{X}_t \nabla \phi(\boldsymbol{x}_t)=\boldsymbol{0}$}{$\widetilde{\boldsymbol{d}}_{t} := \boldsymbol{0}$\;}
    \Else{
    \begin{equation}
    \widetilde{\boldsymbol{d}}_{t}:=- \beta \frac{\boldsymbol{R}_t \boldsymbol{X}_t \nabla \phi(\boldsymbol{x}_t)}{\Vert \boldsymbol{R}_t \boldsymbol{X}_t \nabla \phi(\boldsymbol{x}_t) \Vert};
    \end{equation}
    }

    $\boldsymbol{x}_{t+1} \leftarrow \boldsymbol{x}_t + \boldsymbol{X}_t \widetilde{\boldsymbol{d}}_{t}$\;
    
    \If{$\phi(\boldsymbol{x}_{t+1})-\phi(\boldsymbol{x}_{t}) >- \frac{\varepsilon^2}{2l + 4\varepsilon + 4} $}{
    $\widehat{\boldsymbol{e}}\leftarrow$ Negative Curvature Finding($\boldsymbol{x}_t$)\;
       $\boldsymbol{x}_{t+1} \leftarrow \boldsymbol{x}_t -  \frac{\langle \nabla f(\boldsymbol{x}_t), \boldsymbol{X}_t\widehat{\boldsymbol{e}} \rangle }{\vert\langle \nabla f(\boldsymbol{x}_t), \boldsymbol{X}_t \widehat{\boldsymbol{e}} \rangle \vert} 
    \cdot \frac{3\sqrt{\varepsilon}}{ 8\rho}\boldsymbol{X}_t\widehat{\boldsymbol{e}}$\;
    \If{$f(\boldsymbol{x}_{t+1})-f(\boldsymbol{x}_{t}) > - \frac{9\sqrt{\varepsilon^{3}}}{1024\rho^2}$ }
    {
    Return $\boldsymbol{x}_t$\;
    }
    }
    
    $\ln\overline{\boldsymbol{x}}_{t+1} = \mathtt{SelectVector}(\ln\overline{\boldsymbol{x}}_{t},\ln\boldsymbol{x}_{0},\ln\boldsymbol{x}_{1},\ldots,\ln\boldsymbol{x}_{t+1},\delta_{\text{err}}) $\;
    
    }
\end{algorithm}

\begin{theorem}
\label{thm:robust-nega}
Suppose that \assum{grad-lip}, \assum{hessian-lip}, and \assum{xphi} hold. For any $\varepsilon \in (0, \min\{\gamma, \tfrac{1}{2}\}]$ and $\delta\in[0,1]$, \algo{robust-nega} finds a $(2\varepsilon,\sqrt{\varepsilon})$-KKT2 point within $\mathcal{O}( {l\rho^2 (f(\boldsymbol{x}_0) - f(\boldsymbol{x}^*))}/{\varepsilon^2})$ gradient evaluations with probability $1-\delta$. Moreover, the overall time complexity of the algorithm is upper bounded by $\widetilde{\mathcal{O}}(nm^{\omega-1}/\varepsilon^{1.5} + {nm}/{\varepsilon^2} )$.
\end{theorem}
\begin{proof}
The proof follows the structure of \thm{1st-interior}. During the iterations, each iterate falls into one of the following three categories:
(i) points that are not $2\varepsilon$-KKT points;
(ii) points that satisfy the $2\varepsilon$-KKT conditions but not the $(2\varepsilon,\sqrt{\varepsilon})$-KKT2 conditions;
(iii) points that satisfy the $(2\varepsilon,\sqrt{\varepsilon})$-KKT2 conditions.

For iterates of the first type, the potential function decreases by at least $\frac{\varepsilon^2}{2l+4\varepsilon+4}$ according to \prop{potential-decre}. Hence, the number of such iterations is at most $T_1 = \frac{\left(f(\boldsymbol{x}_0)- f(\boldsymbol{x}^*) + (C_0-1)\varepsilon\right)(2l+4\varepsilon+4)}{\varepsilon^2}$.

For points of the second type, set $\delta_0 = \frac{{9\sqrt{\varepsilon^{3}}}}{1024(f(\boldsymbol{x}_0)-f(\boldsymbol{x}^*)){\rho^2}} \delta$ in \prop{nega-curva} as in \thm{1st-interior}. With probability at least $1 - \delta_0$, \proc{nega-curvature} identifies a negative curvature direction $\widehat{\boldsymbol{e}}$. By \prop{func-decrease}, moving along this direction decreases the function value by at least $\frac{9\sqrt{\varepsilon^{3}}}{1024\rho^2}$. Thus, the number of iterations corresponding to the second type is bounded by $T_2 = \frac{1024(f(\boldsymbol{x}_0)-f(\boldsymbol{x}^*)){\rho^2}}{{9\sqrt{\varepsilon^{3}}}}$ with probability at least $1-\delta$, which is exactly the same to \thm{1st-interior}.

Combining both cases, define
\begin{equation}
    T = \max\left\{2T_2, 4T_1\right\} = \max\left\{\frac{2048(f(\boldsymbol{x}_0)-f(\boldsymbol{x}^*)){\rho^2}}{{9\sqrt{\varepsilon^{3}}}}, \frac{4\left(f(\boldsymbol{x}_0)- f(\boldsymbol{x}^*) + (C_0-1)\varepsilon\right)(l+2\varepsilon+2)}{\varepsilon^2}\right\}.
\end{equation}
Under this bound, at most $T/4$ iterations correspond to points of the first type, and at most $T/2$ iterations correspond to points of the second type with probability at least $1-\delta$. Therefore, at least $T/4$ iterations correspond to points of the third type, i.e., $(2\varepsilon,\sqrt{\varepsilon})$-KKT2 points with probability at least $1-\delta$.

For the time complexity analysis, each call to \proc{nega-curvature} requires computing the projection matrix $\boldsymbol{P}_t$ onto the null space of $\boldsymbol{AX}_t$, which incurs a cost of $\mathcal{O}(nm^{\omega-1})$. As shown earlier, the algorithm invokes \proc{nega-curvature} at most $\mathcal{O}(1/\varepsilon^{1.5})$ times. Therefore, the total time complexity contributed by all invocations of \proc{nega-curvature} is $\mathcal{O}(nm^{\omega-1}/\varepsilon^{1.5})$. 

After each invocation of \proc{nega-curvature}, once $\boldsymbol{P}_t$ has been computed, we restart the $\mathtt{SelectVector}$ procedure and recompute $\boldsymbol{R}_t$ based on the updated projection. By \prop{robust-time}, the additional computational cost for $T$ iterations of sparse updates is $\mathcal{O}(mnT)$. Since the total number of iterations between all \proc{nega-curvature} calls sums to $\mathcal{O}(1/\varepsilon^{2})$, the total cost of sparsely updating $\boldsymbol{R}_t$ and computing $\boldsymbol{d}_t$ is $\mathcal{O}(mn/\varepsilon^{2})$. Therefore, the overall time complexity is $\mathcal{O}(nm^{\omega-1}/\varepsilon^{1.5} +mn/\varepsilon^{2})$.
\end{proof}

\section{Numerical Experiments}
\label{sec:expri}

\subsection{Empirical convergence of first-order methods to KKT2 points}
\label{sec:small-expri}
We demonstrate our algorithm\footnote{The complete source code for all experiments in \sec{expri} is available at \href{https://github.com/Macondooo/Approximate-First-Order-IPTR-Algorithms}{GitHub repository}.} on a visualized example with $n=3$ and $m=1$. The feasible region is constrained to the plane $x_0+x_1+x_2=1$ in the first orthant. The objective function is designed as a quartic function with a bowl-shaped outer landscape, containing one saddle point and two local minima in its interior. The examples illustrate the  differences in convergence behavior between first-order and second-order algorithms with respect to KKT optimality conditions. 

\begin{figure}[!htb]
    \centering
    \includegraphics[width=0.95\linewidth]{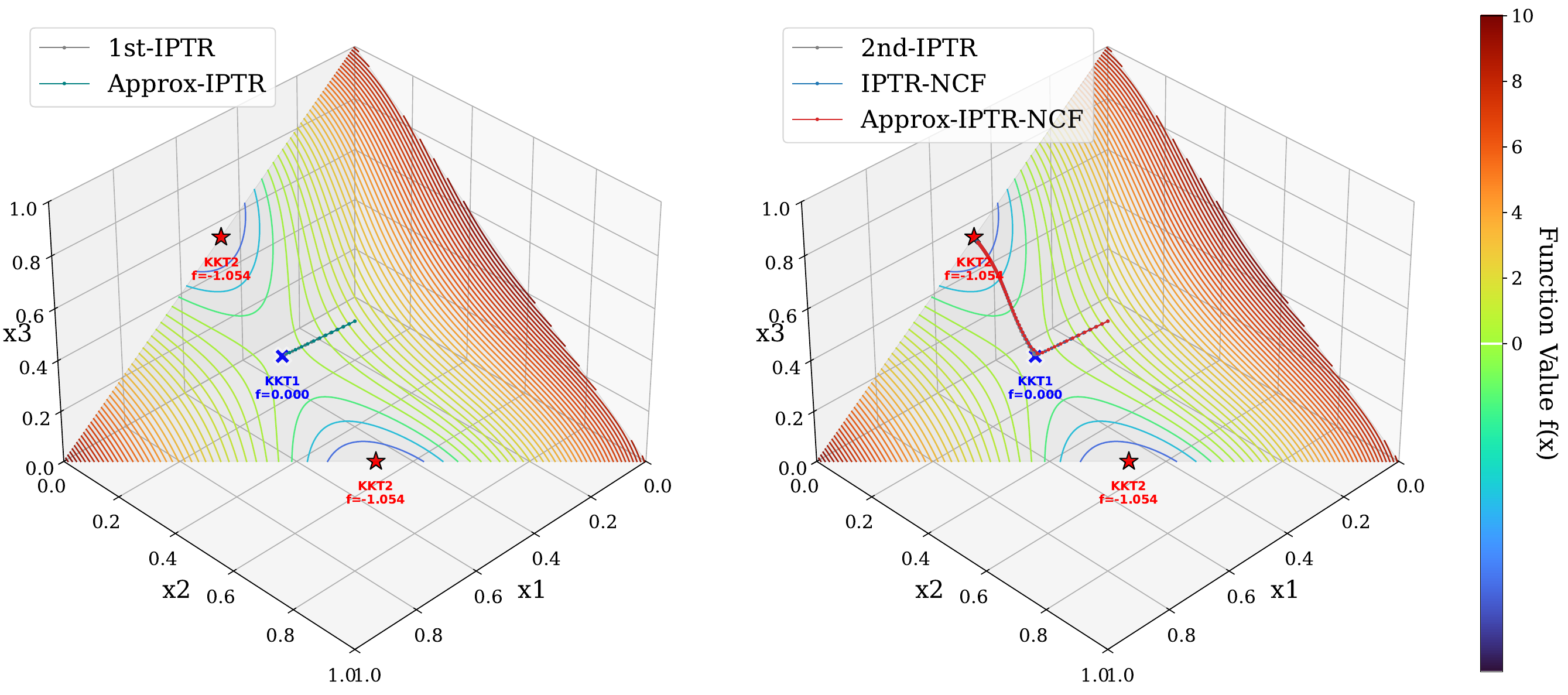}
    \caption{Iteration trajectories of IPTR-type algorithms on the constrained nonconvex problem $\min_{\boldsymbol{x}} f(\boldsymbol{x}) = 40 (x_0 - 0.5)^2 - 5 (x_1 - x_2)^2 + 4 (x_1 - x_2)^4, \text{ s.t. } x_0+x_1+x_2=1,\ \boldsymbol{x} \geq \boldsymbol{0}$. The left panel compares the first-order IPTR algorithm from \cite{haeser2019optimality} (1st-IPTR) with our \robustalgo first-order IPTR algorithm (\algo{robust-1st-order}, Approx-IPTR). The right panel compares the second-order IPTR algorithm from \cite{haeser2019optimality} (2nd-IPTR) with our first-order IPTR with negative curvature finding (\algo{1st-order interior}, IPTR-NCF) and its \robustalgo variant (\algo{robust-nega}, Approx-IPTR-NCF). Iteration points are plotted every 20 iterations.}
    \label{fig:n3m1-1}
\end{figure}

\begin{figure}[!htb]
    \centering
    \includegraphics[width=0.95\linewidth]{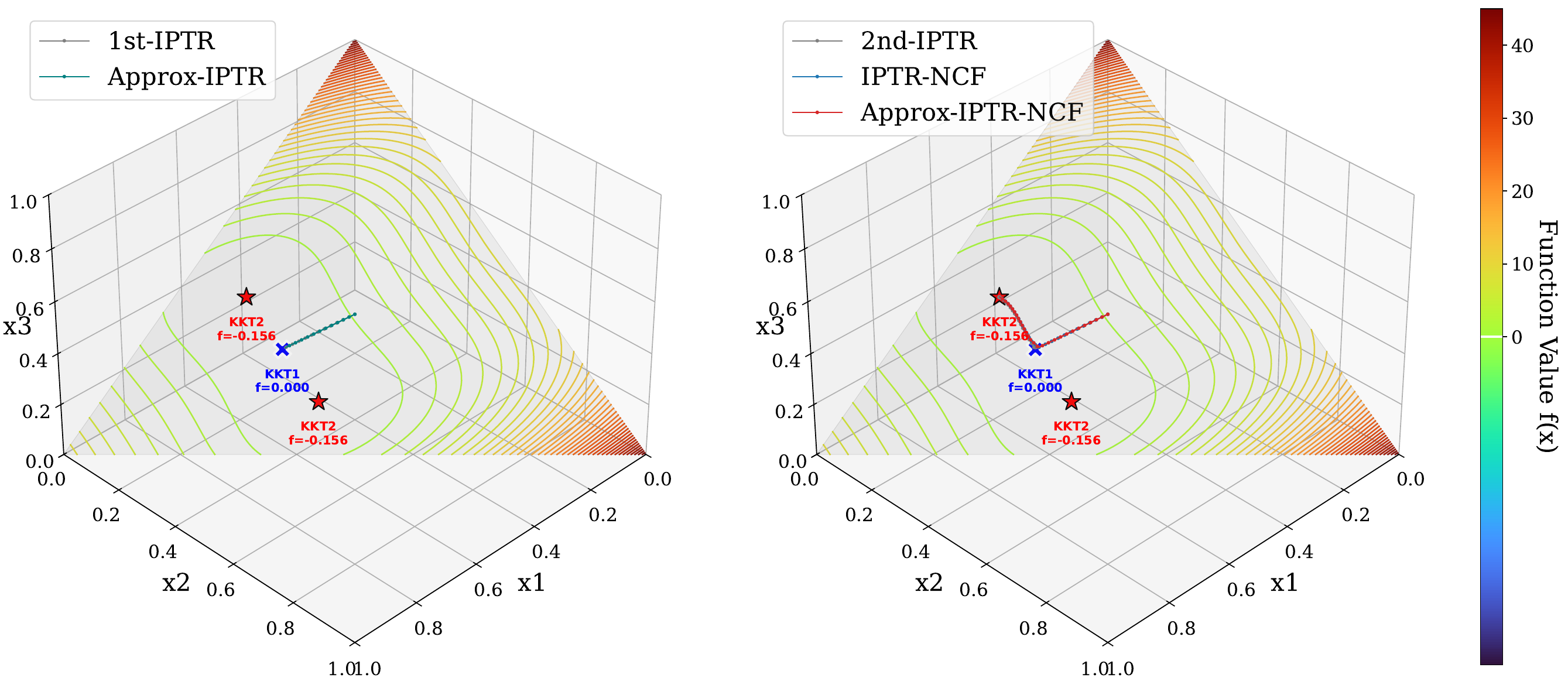}
    \caption{Iteration trajectories of IPTR-type algorithms on the constrained nonconvex problem $\min_{\boldsymbol{x}} f(x) = 40 (x_0 - 0.5)^2 - 5 (x_1 - x_2)^2 + 40 (x_1 - x_2)^4, \text{ s.t. } x_0+x_1+x_2=1,\ \boldsymbol{x} \geq \boldsymbol{0}$. The plot layout, curve styles, and performance measures follow the same conventions as in \fig{n3m1-1}.}
    \label{fig:n3m1-2}
\end{figure}

In both \fig{n3m1-1} and \fig{n3m1-2}, the exact first-order and second-order KKT points are marked on the optimization landscape. The \robustalgo and non-\robustalgo variants follow very similar trajectories, indicating that incorporating the \robustalgo acceleration mechanism does not affect the convergence behavior of the original algorithm.
In the left panel of \fig{n3m1-1}, both the first-order IPTR method of \cite{haeser2019optimality} and our \robustalgo first-order variant \algo{robust-1st-order} converge to a first-order KKT point but do not reach a second-order KKT point.
When a negative-curvature finding step is incorporated, as in \algo{1st-order interior} and \algo{robust-nega}, the algorithms are able to escape first-order KKT points and converge to second-order KKT points using only first-order information, as shown in the right panel.
\fig{n3m1-2} presents another example in which the second-order KKT point lies in the interior of the feasible region.  Despite this structural difference, the convergence behavior remains consistent with the previous example.

\subsection{Large-scale empirical evaluation}

To further assess the practical scalability of our algorithms, we conduct large-scale experiments with high-dimensional instances. The goal of this subsection is to empirically examine the time complexity and demonstrate the computational advantage of the \robustalgo variant in high-dimensional regimes. In the experiments, the stopping criteria for the IPTR algorithms exactly follow those specified in the algorithms. \algo{robust-1st-order} checks the decrease in the potential function, $\phi(\boldsymbol{x}_{t+1}) - \phi(\boldsymbol{x}_t)$, while \algo{1st-order interior} and \algo{robust-nega} check the decrease in the objective value, $f(\boldsymbol{x}_{t+1}) - f(\boldsymbol{x}_t)$ according to \prop{func-decrease}.

\begin{figure}[!h]
    \centering
    \includegraphics[width=\linewidth]{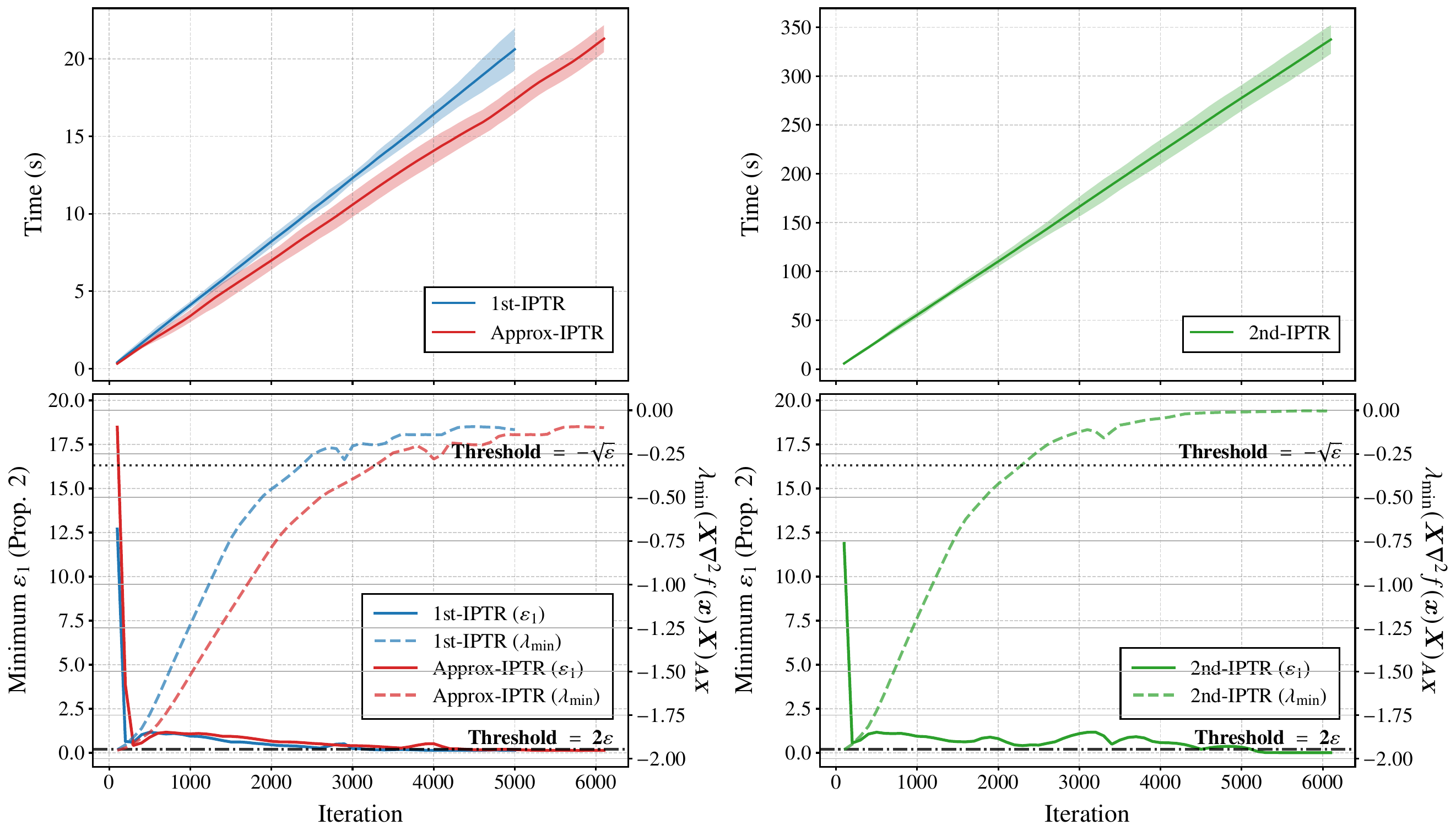}
    \caption{Empirical running time per iteration and convergence to approximate first- and second-order KKT points for an instance of \eq{exp_instance} with $n=1000$ and $m=500$. The left two plots compare the first-order IPTR method from \cite{haeser2019optimality} (1st-IPTR) with our \robustalgo first-order variant (\algo{robust-1st-order}, Approx-IPTR). The right two plots present results for the second-order IPTR method from \cite{haeser2019optimality} (2nd-IPTR). For each pair of plots, the upper panel illustrates the running time per iteration. Solid curves denote the mean over five independent runs and the shaded regions indicate $3 \times \text{std}$. The lower panel shows the convergence toward approximate first- and second-order KKT points. Solid curves correspond to the first-order optimality measure from \prop{kkt1} or \prop{kkt2}(1)--(3). The horizontal line at $2\varepsilon$ with $\varepsilon=0.1$ marks the threshold. Dashed curves correspond to the minimum-eigenvalue condition for second-order optimality from \prop{kkt2}(4), with threshold $-\sqrt{\varepsilon}$. All quantities are recorded every 100 iterations.
}
    \label{fig:IPTR-time-kkt1}
\end{figure}

We considers a quartic objective function consisting of a separable quartic term together with a quadratic component involving cross terms:
\begin{equation}
\label{eq:exp_instance}
    \min_{\boldsymbol{x}} f(\boldsymbol{x}):= \sum_i \frac{\sigma}{4}  x_i^4 + \frac{1}{2}\boldsymbol{x}^\top \boldsymbol{Q} \boldsymbol{x} + \boldsymbol{c}^\top \boldsymbol{x}, \text{ s.t. } \boldsymbol{A}\boldsymbol{x}=\boldsymbol{b}, \boldsymbol{x} \geq \boldsymbol{0}.
\end{equation}
The parameter $\sigma$ controls the relative influence of the quartic term. The constraint matrix $\boldsymbol{A} \in \mathbb{R}^{m\times n}$ is constructed by fixing its first row to be the normalized all-ones vector in order to ensure that the feasible 
region is bounded. The remaining rows are generated randomly with entries drawn from the interval $(0,1)$, while ensuring that $\boldsymbol{A}$ has full row rank. The right-hand side vector $\boldsymbol{b}$ is defined as $\boldsymbol{b}:=\boldsymbol{A}\boldsymbol{x}_0$, where $\boldsymbol{x}_0$ is a strictly feasible interior point. In our experiments, $\boldsymbol{x}_0$ is used as an approximate analytic center of the feasible set. 
The quadratic matrix $\boldsymbol{Q}$ is designed to be nonconvex on the null space of $\boldsymbol{A}$. The linear term $\boldsymbol{c}$ is chosen such that the constrained problem admits at least one stationary saddle point. This construction allows us to systematically generate nonconvex quartic objectives with linear equality constraints and controlled saddle-point geometry for evaluating algorithmic performance. 

\begin{figure}[p]
    \centering
    \includegraphics[width=\linewidth]{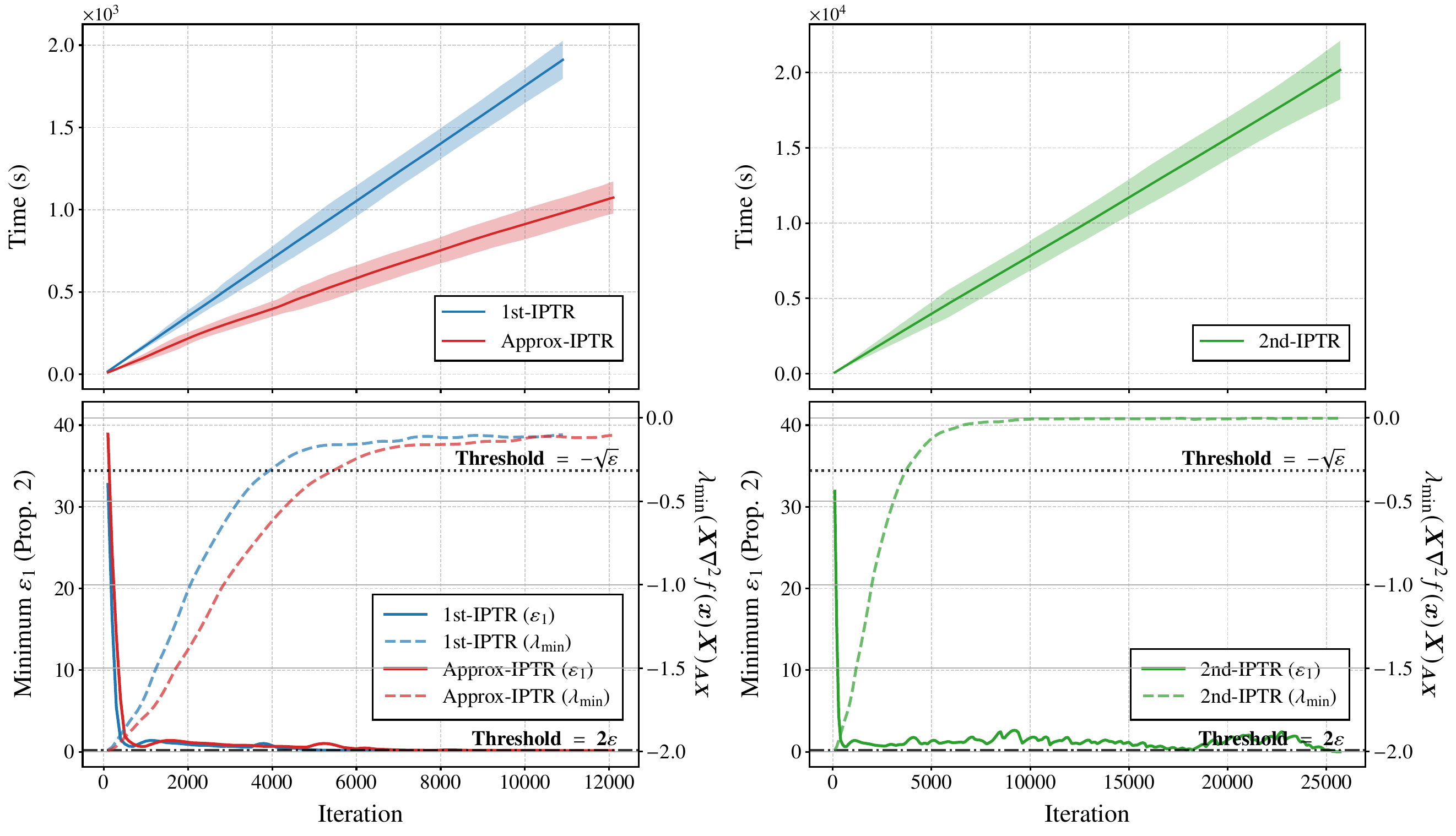}
    \caption{Empirical running time per iteration and convergence to approximate first- and second-order KKT points for an instance of \eq{exp_instance} with $n=3000$ and $m=2000$. The plot layout, curve styles, and performance measures follow the same conventions as in \fig{IPTR-time-kkt1}.}
    \label{fig:IPTR-time-kkt2}
\end{figure}
\begin{figure}[p]
    \centering
    \includegraphics[width=\linewidth]{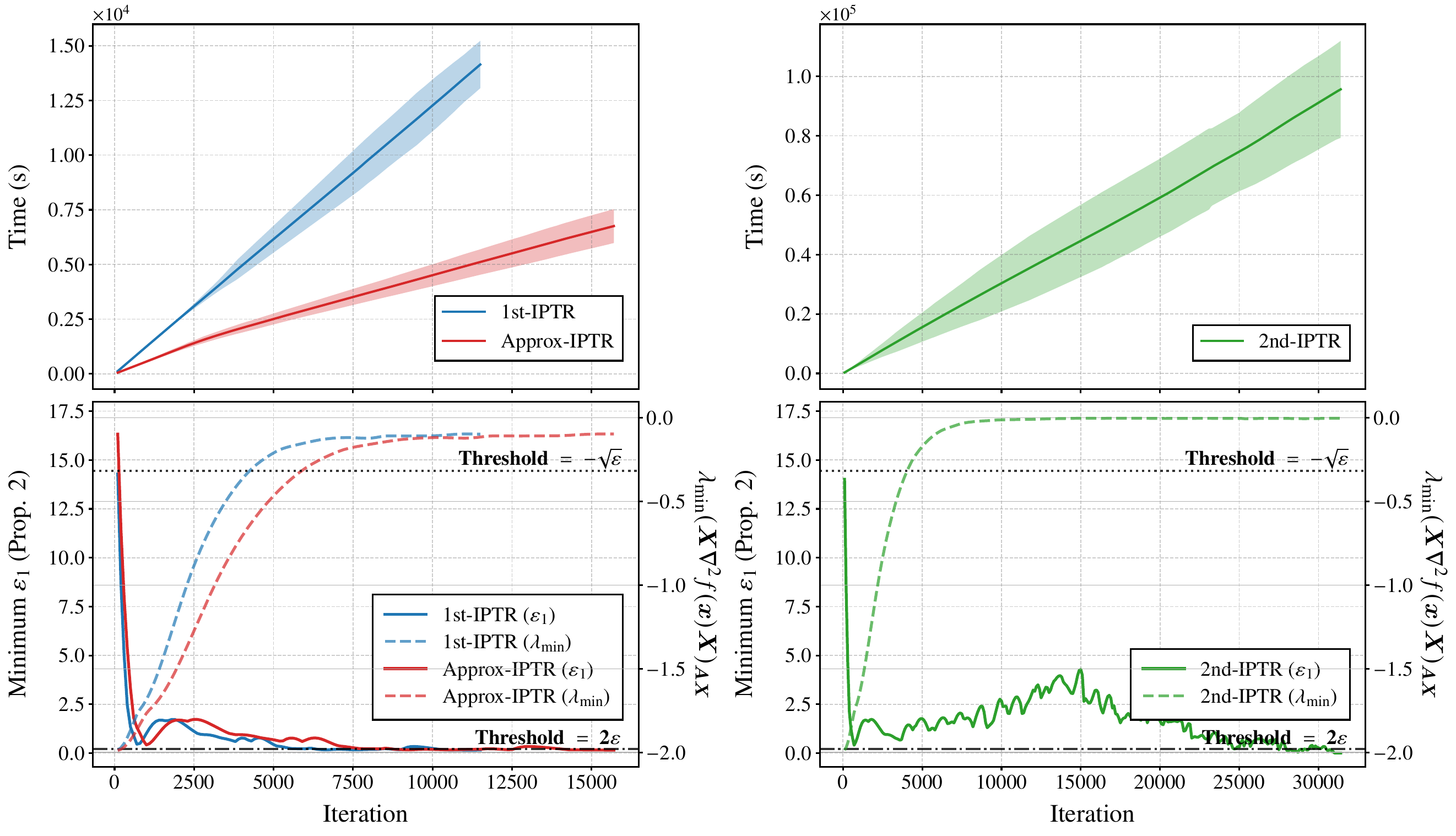}
    \caption{Empirical running time per iteration and convergence to approximate first- and second-order KKT points for an instance of \eq{exp_instance} with $n=5000$ and $m=4000$. The plot layout, curve styles, and performance measures follow the same conventions as in \fig{IPTR-time-kkt1}.}
    \label{fig:IPTR-time-kkt3}
\end{figure}

To demonstrate the speedup potential on large-scale problems, we evaluate three instances of varying sizes: $(n,m)=(1000,500)$, $(3000,2000)$ and $(5000,4000)$, as depicted in \fig{IPTR-time-kkt1}, \fig{IPTR-time-kkt2}, and \fig{IPTR-time-kkt3}, respectively. In these settings, the \robustalgo first-order IPTR algorithm (\algo{robust-1st-order}) requires significantly less computational time compared to the first-order IPTR method of \cite{haeser2019optimality}. This improvement stems from maintaining the projection via structured low-rank updates rather than explicitly recomputing it at each iteration, thereby substantially reducing the per-iteration complexity. As detailed in \tab{speedup_ratio}, the per-iteration speedup of \algo{robust-1st-order} scales favorably with problem dimensions. Specifically, the speedup increases from $1.18\times$ to $1.78\times$ and $2.48\times$ as the problem size $(n,m)$ grows from $(1000,500)$, $(3000,2000)$ and $(5000,4000)$. This demonstrates that our method is highly scalable and particularly well-suited for large-scale optimization tasks. 
For the first-order algorithms with negative curvature finding subroutine, namely \algo{1st-order interior} and \algo{robust-nega}, we also ran them in our experiments. In the large-scale instances considered here, the first-order iterations already approach approximate second-order KKT points near the boundary before the negative curvature routine is triggered. As a result, the negative curvature routine is invoked only for a few iterations near termination. This observation may be attributed to the high-dimensional structure of the problem, as the iterates rapidly approach the boundary of the feasible region and then progress along the boundary, ultimately settling at a second-order KKT point. Before that, their iterates coincide with those of the corresponding first-order methods. Hence we do not plot them separately in the figures.

\begin{table}[!htb]
    \centering
    \caption{Per-iteration speedup of Approx-IPTR relative to 1st-IPTR across different problem scales.}
    \label{tab:speedup_ratio}
    \renewcommand{\arraystretch}{1.15}
    \setlength{\tabcolsep}{12pt}
    \begin{tabular}{@{}cc@{}}
        \toprule
        Problem Size $(n,m)$ & Speedup vs.\ 1st-IPTR \\
        \midrule
        $(1000,500)$  & $1.18\times$ \\
        $(3000,2000)$ & $1.78\times$ \\
        $(5000,4000)$ & $2.48\times$ \\
        \bottomrule
    \end{tabular}
\end{table}

The second-order IPTR algorithm of \cite{haeser2019optimality} requires more iterations and longer running time than all first-order variants. In our implementation, the quadratic programming subproblem arising at each iteration is solved using the method of \cite{ye1998complexity}. Solving this QP subproblem is substantially more expensive than solving the linearized trust-region subproblem in the first-order framework, which admits a closed-form solution. This highlights the computational advantage of first-order methods, and in particular the benefit of our \algo{1st-order interior} and \algo{robust-nega}, which attain approximate second-order KKT points while retaining the lower per-iteration cost of a first-order scheme.

In summary, the large-scale experiments indicate that our approximate first-order IPTR algorithm, \algo{robust-1st-order}, is well suited for large-scale problem instances. Its advantage over the existing first-order IPTR method becomes more substantial as the problem dimension increases, showing the benefit of the proposed approximate update scheme in the large-scale regime. 
The same approximate update scheme can also be incorporated into a first-order method equipped with a negative curvature finding subroutine, thereby guaranteeing convergence to approximate second-order KKT points, as demonstrated by \algo{robust-nega} in \sec{small-expri}. 
Another practical advantage of these first-order methods is that they do not require access to Hessian information. This is particularly appealing in large-scale problems, where the computation and storage of Hessian matrices may become prohibitively expensive. Taken together, these results suggest that the proposed approach provides an effective and practical framework for large-scale constrained optimization.

\section{Conclusion}
In this paper, we developed efficient first-order IPTR algorithms for computing approximate first- and second-order KKT points of nonconvex optimization problems with affine equality and nonnegative constraints. Specifically, our algorithm for computing approximate first-order KKT points replaced the exact projection step by an approximate update scheme, thereby reducing the average per-iteration cost to essentially that of matrix--vector multiplication while preserving the convergence guarantees of existing first-order IPTR algorithms. We further showed for the first time that approximate second-order KKT points can also be computed for constrained optimization within the first-order IPTR framework by incorporating a negative-curvature finding procedure based on the projected power method and finite-difference gradient approximations, thus avoiding explicit Hessian evaluations. 
To complement with our theory results, we also conducted extensive numerical experiments to evaluate the empirical performance of the proposed algorithms. On representative examples, the results showed that our first-order IPTR algorithms with the negative-curvature finding subroutine are able to escape first-order KKT points and converge to approximate second-order KKT points. On large-scale instances, the \robustalgo first-order IPTR algorithm consistently improves upon the runtime of the existing first-order IPTR algorithm, and the advantage becomes more significant as the problem size increases. 
These results indicate that the IPTR algorithms proposed in this paper not only have provable advantage in theory, but also practically effective for large-scale constrained nonconvex optimization problems.

Our paper leaves several open questions for future investigation:
\begin{itemize}
\item In unconstrained nonconvex optimization, the best known iteration complexity for finding an approximate second-order stationary point is $\widetilde{O}(1/\varepsilon^{1.75})$. It is natural to study whether a comparable complexity bound can be established for the constrained nonconvex optimization setting considered in this paper. At present, such a result is not known for general constrained problems. The primary difficulty is that both the exploitation of negative curvature and the use of accelerated steps must respect the local geometry of the feasible region. This geometry varies with the active set and often complicates feasibility-preserving updates. Establishing an $\widetilde{O}(1/\varepsilon^{1.75})$ bound for first-order methods in general constrained nonconvex optimization therefore remains an open question. 

\item It would also be of general interest to extend the proposed approximate first-order IPTR framework beyond affine equality and nonnegativity constraints, i.e., to more general classes of constrained nonconvex optimization problems.

\item Finally, it remains to be understood whether the approximate update mechanism developed in this paper can also be leveraged to accelerate the trust-region QP algorithm \cite{ye1998complexity} that arise in second-order IPTR algorithms.
\end{itemize}

\section*{Acknowledgments}
We thank Yurii Nesterov for helpful discussions, especially for suggesting that we consider concave objective functions, which inspires the results in \sec{concave}.
YS and TL were supported by the National Natural Science Foundation of China (Grant Numbers 62372006 and 92365117).

\newcommand{\arxiv}[1]{arXiv:\href{https://arxiv.org/abs/#1}{\ttfamily{#1}}\?}\newcommand{\arXiv}[1]{arXiv:\href{https://arxiv.org/abs/#1}{\ttfamily{#1}}\?}\def\?#1{\if.#1{}\else#1\fi}
\providecommand{\bysame}{\leavevmode\hbox to3em{\hrulefill}\thinspace}

\end{document}